\documentclass[aps,floatfix,amsmath,twocolumn,notitlepage,nofootinbib,longbibliography,superscriptaddress]{revtex4-2}

\usepackage[ruled]{algorithm2e}
\SetKwFor{For}{for (}{) $\lbrace$}{$\rbrace$}
\usepackage{amssymb}
\usepackage{graphicx}
\usepackage{graphics}
\usepackage{amsmath}
\usepackage{amsthm}
\usepackage{color}
\usepackage{dsfont}
\usepackage{wrapfig}

\usepackage{longtable}

\usepackage{tabularray}

\usepackage{float}
\usepackage{mathtools}
\usepackage[colorlinks=true, linkcolor=blue, citecolor=red]{hyperref}
\usepackage{verbatim}
\usepackage{xcolor}
\usepackage{soul}

\usepackage{multirow}

\usepackage{wrapfig}
\usepackage{array}

\usepackage{multirow}

\newcolumntype{C}[1]{>{\centering\let\newline\\\arraybackslash\hspace{0pt}}m{#1}}

\newcommand{\p}{\mathrm{p}}

\usepackage{blkarray}

\usepackage{color}

\def\>{\rangle}
\def\<{\langle}

\def\id{\mathsf{id}}

\renewcommand{\geq}{\geqslant}
\renewcommand{\leq}{\leqslant}

\newcommand{\A}{\mathbf{A}}
\newcommand{\B}{\mathbf{B}}

\newtheorem*{theorem*}{Theorem}

\newtheorem{lemma}{Lemma}
\newtheorem*{lemma*}{Lemma}
\newtheorem{claim}{Claim}
\newtheorem*{claim*}{Claim}
\newtheorem{proposition}{Proposition}

\newtheorem*{definition*}{Definition}
\newtheorem{conjecture}{Conjecture}

\theoremstyle{definition}

\theoremstyle{remark}

\theoremstyle{definition}

\newtheorem*{axiom*}{Axiom}

\newcommand{\tr}{{\rm Tr}}

\newcommand{\ve}[1]{{\left\vert\kern-0.25ex\left\vert\kern-0.25ex\left\vert #1 
    \right\vert\kern-0.25ex\right\vert\kern-0.25ex\right\vert}}

\newcommand{\ket}[1]{|#1\rangle}

\newcommand{\ketbra}[2]{|#1\rangle\!\langle#2|}

\newcommand{\1}{\mathds{1}}

\usepackage[capitalise]{cleveref}

\definecolor{cool_green}{rgb}{0.0, 0.5, 0.0}

\begin{document}

\title{Quantum Nonlocality and Device-Independent Randomness are Robust to Noisy Signaling Channels}

\author{Kuntal Sengupta}
\thanks{A preliminary version of this work appears in the author's Ph.D. thesis submitted to the University of York~\cite{sengupta2025}.
\href{kuntal.sengupta@neel.cnrs.fr}{kuntal.sengupta@neel.cnrs.fr}}
\affiliation{Inria, Univ.\ Grenoble Alpes, CNRS, 38000 Grenoble, France} 

\author{Lewis Wooltorton}
\thanks{This work was partially conducted at the author's previous affiliation: the University of Bristol and the University of York, UK. \href{lewis.wooltorton@ens-lyon.fr}{lewis.wooltorton@ens-lyon.fr}}
\affiliation{Inria, ENS de Lyon, UCBL, LIP, 46 Allee d’Italie, 69364 Lyon Cedex 07, France}
\begin{abstract} 
    Given a pair of isolated devices that accept random binary inputs and return binary outputs, a user can deduce from the observed data alone if the underlying mechanism can be explained classically. Bell's theorem further states that a classical explanation can be ruled out if the devices perform certain measurements on an entangled quantum state, underpinning the security of cryptographic protocols that are device-independent (DI). For certain protocols, such as those performed in a tight space, it might be difficult to perfectly enforce the non-signaling assumption required in Bell's theorem. This prompts the question: is quantum nonlocality robust to such imperfections? We show that if a binary channel sends a noisy copy of one party's input to the other before any measurements take place, the answer is yes. We completely characterize the vertices and facets of the local polytope in this scenario, and identify Bell inequalities that certify non-signaling quantum correlations. This is possible even when a near perfect copy of the input is sent. We go on to show that the identified inequalities are more robust to depolarizing noise than the Clauser-Horne-Shimony-Holt inequality when certifying DI randomness in this scenario. In addition, we characterize the local polytope when both parties receive a noisy copy of each other's input and make similar conclusions, leaving many new potential Bell inequalities to be explored. 
\end{abstract}
\maketitle

\section{Introduction}

Bell's theorem~\cite{Bell, Brunner_review} enables an experimentalist to certify the nonlocal behavior  of their setup using observed data only. Given access to a source of free randomness and the ability to enforce a non-signaling condition between different physical components, witnessing the violation of a Bell inequality rules out any local hidden variable (LHV) explanation of the observed correlations~\cite{EPR}. This has lead to the advent of device-independent (DI) information processing, where protocols for securely amplifying and expanding randomness~\cite{ColbeckThesis,CK2,PAMBMMOHLMM,MS1,MS2,CR_free}, or distributing a secret key~\cite{Ekert,ABGMPS,PABGMS,VV2,ADFRV}, can be realized without any physical characterization of the devices used. Provided the assumptions of Bell's theorem are met, the violation of a Bell inequality certifies that the devices are functioning correctly, and this alone is sufficient to establish the protocol's security.          

However, failing to meet one of Bell's assumptions could lead to incorrect security claims. This could happen when performing a Bell test in a tight space, such as on a single quantum device where strict non-signaling (often refered to as parameter independence) between different physical regions cannot be guaranteed~\cite{Foreman2023,Fyrillas24}. For example, the authors of Ref.~\cite{Silman_2013} estimated the amount of signaling present in a Bell test featuring two superconducting qubits on a single chip~\cite{Ansmann2009}, and in Ref.~\cite{Fyrillas24}, DI randomness was certified from a single two-qubit photonic chip\footnote{See~\cite{Silman_2013} and~\cite[Appendix A]{Fyrillas24} for discussions on estimating the amount of signaling present in a given setup.}. To address this concern, investigating whether the assumptions of Bell's theorem can be relaxed while preserving the separation between classical and quantum phenomena has become an interesting avenue of research. 

The amount of communication needed for classical devices to reproduce quantum correlations has been well studied~\cite{Maudlin1992,BCT99,Steiner_2000,Pawlowski_2010,Gha21}. In particular, Toner and Bacon showed that all quantum correlations generated by performing projective measurements on the singlet can be reproduced by LHV models when one bit of classical communication is permitted~\cite{TonerBacon03}. Furthermore, two bits of communication are sufficient to reproduce all quantum correlations generated by projective measurements on arbitrary two-qudit states~\cite{TonerBacon03,RT10,VB9}. Whether outcome communication can be used to the same effect was more recently studied in Ref.~\cite{vieira2025}. A closely related line of research has investigated how much measurement dependence, i.e., the extent to which a hidden variable can be correlated with the input settings, is required to reproduce quantum correlations~\cite{BG11,Hall10,Hall11}. Notably, Barrett and Gisin showed that this type of relaxation is equivalent to models with classical communication about the input settings~\cite{BG11}. 

Rather than assuming the existence of a non-signaling quantum correlation and asking how much a given assumption must be relaxed to explain it classically (in the spirit of Bacon and Toner for parameter independence~\cite{TonerBacon03} or Barrett and Gisin for measurement independence~\cite{BG11}), one can also ask the converse. Consider an experimentalist who can only perform an imperfect Bell test, where the extent of the imperfection (e.g., the amount of signaling) has been characterized. Are there non-signaling quantum correlations realizable within their setup that cannot be explained classically? It was shown by P\"utz \textit{et al.} that an arbitrarily small amount of measurement independence is sufficient to witness nonlocality~\cite{Putz_2016,Aktas15,Supic23}, and, more recently, Vieira, Ramanathan and Cabello showed that quantum correlations cannot be simulated even in the presence of arbitrary relaxations of measurement and parameter independence simultaneously~\cite{Vieira_2025,ramanathan2025}.     

The existence of such correlations have naturally lead to new DI protocols robust to certain relaxations. For example, in DI randomness amplification, a weak source of randomness can be used to choose the inputs of the Bell test, and the outputs can be processed to generate a private random string~\cite{Colbeck2012,Gallego_2013,Foreman2023,kulikov2024}. In the case of perfectly random inputs, constrained signaling between the devices can also be tolerated in randomness expansion protocols~\cite{Silman_2013, tan2023, Fyrillas24, arqand2024, ramanathan2025}.

With the exception of Refs.~\cite{Silman_2013,tan2023} that consider quantum models for weak cross-talk and constrained leakage, respectively, the aforementioned works relax parameter (or measurement) dependence at the level of the LHV behavior . For example, in Refs.~\cite{Hall11, Vieira_2025,ramanathan2025}, LHV models with relaxed parameter dependence are those where Alice's outcome $A=a$ given an input $X=x$ and hidden variable $\Lambda=\lambda$ is allowed to depend on Bob's input $Y=y$ by some fixed amount $\epsilon > 0$, i.e.,
\begin{equation}
   \big|  \p(a|x,y=0,\lambda) - \p(a|x,y=1,\lambda) \big| \leq \epsilon. \label{eq:sig1}
\end{equation}
Note that \eqref{eq:sig1} does not assume any particular mechanism by which the signaling takes place, only that it does so in a restricted manner. Checking for membership of this set can be phrased as a linear program, and new inequalities that are satisfied by all such LHV models, yet exhibit a quantum violation, have been derived~\cite{Hall11, Vieira_2025,ramanathan2025}. However, to the best of our knowledge, a complete description of all facets and vertices of an LHV polytope under relaxed parameter independence has not yet been found. Identifying facet Bell inequalities in this scenario is important for DI randomness, since they are often more robust to noise than non-facet alternatives~\cite{WBC}. Furthermore, a complete characterization would be a useful tool for systematically searching over different inequalities to find one that is optimal for randomness generation in the presence of signaling.                

In this manuscript, we obtain complete descriptions of LHV correlations that can exploit input signaling via an explicit noisy channel. As a result, we identify Bell inequalities that certify nonlocal quantum correlations in the presence of signaling. In more detail, we consider the minimal Bell scenario\footnote{We refer to the bipartite Bell scenario with two inputs and two outputs per party as the minimal Bell scenario.} in the presence of binary channels that send a noisy copy of one party's input to the other before any measurements take place; see \cref{fig:scenario} for an example featuring one-way signaling from Alice to Bob. This enables LHV models to simulate certain nonlocal correlations depending on the crossover probability of the channel, and we obtain a complete description of this set in terms of its vertices and facets for certain channel instances. The Bell inequalities we identify are facets that are violated by non-signaling quantum correlations. Furthermore, we provide examples of such inequalities that are maximally violated by a partially entangled pair of qubits even when one party's input is almost perfectly sent to the other, a similar phenomenon to that recently reported in Refs.~\cite{Vieira_2025,ramanathan2025} using different techniques. Interestingly, the facets we identify correspond to a previously discovered family of self-tests among the set of non-signaling quantum correlations~\cite{Barizien_2024}. 

We also fully characterize the LHV polytope when both parties receive a noisy copy of each other's input, and identify exposed non-signaling quantum correlations even when the signaling is arbitrarily strong. In the two-way setting, we prove an equivalence between the set of LHV correlations induced by our signaling model and those induced by the general definition of joint parameter and measurement dependence, as considered by Refs.~\cite{Vieira_2025,ramanathan2025}. This establishes the noisy channel model as a general tool to study relaxed assumptions in the minimal Bell scenario.

\begin{figure}
    \centering
    \includegraphics[width=0.5\textwidth]{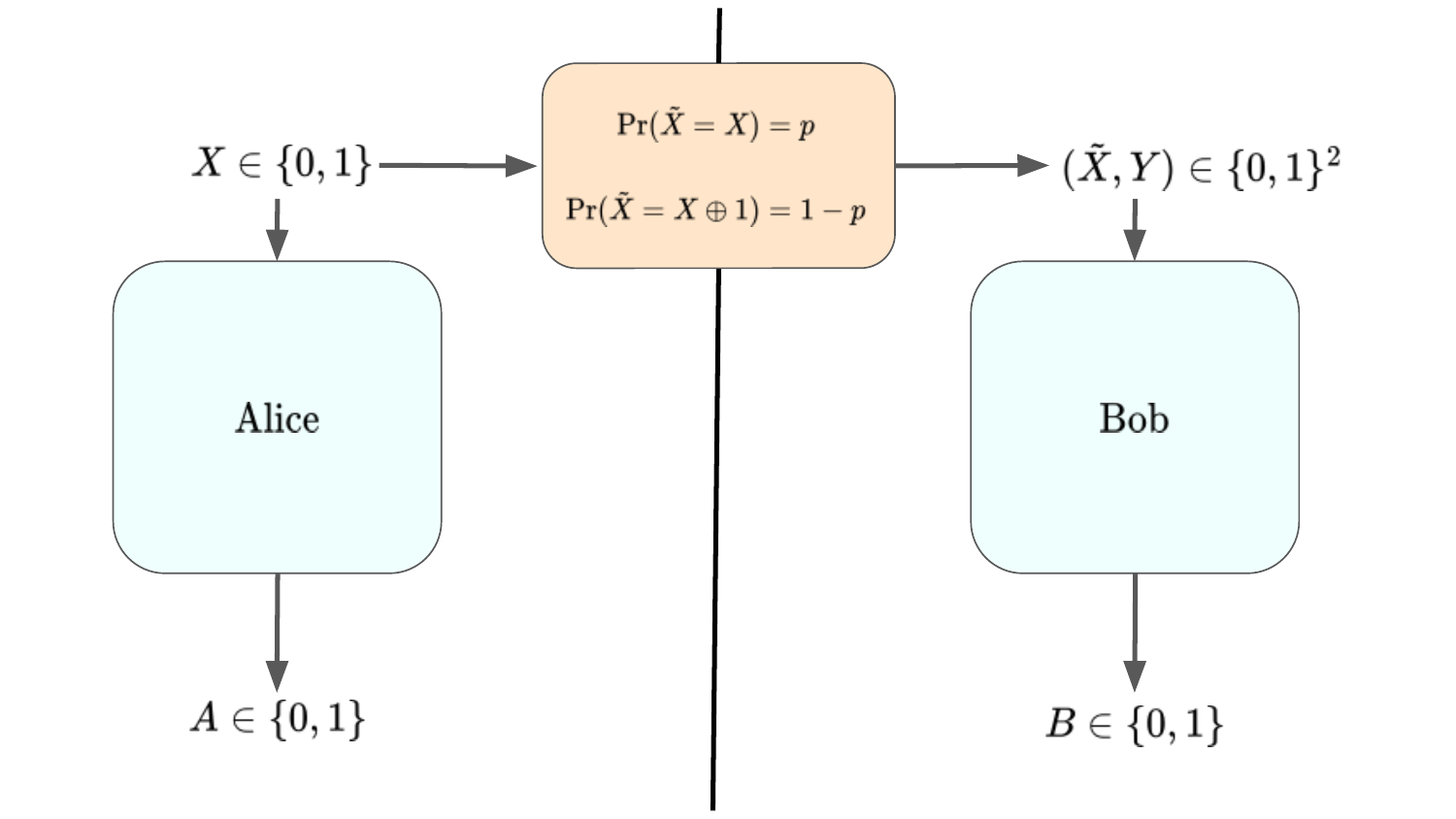}
    \caption{An illustration of the minimal Bell scenario in the presence of a one-way binary symmetric channel from Alice to Bob. Depending on the crossover probability $1-p \in (0,1/2]$, a noisy copy of Alice's input $X$ is also an input to Bob's device. When $p = 1/2$, Bob receives no information about Alice's input and we recover the standard non-signaling Bell scenario, whereas when $p = 1$, Bob's strategy can depend arbitrarily on Alice's input.}
    \label{fig:scenario}
\end{figure}

We then investigate the certification of DI randomness in the presence of noisy input signaling via binary channels. We show using the numerical technique of Ref.~\cite{BrownDeviceIndependent2} that the identified facet inequalities are more robust to depolarizing noise than the Clauser-Horne-Shimony-Holt (CHSH) inequality~\cite{CHSH} for any degree of input signaling by one party, suggesting the practical benefit of using partially entangled states in this context. Interestingly, our numerical results suggest that almost maximum local randomness can be certified in the presence of near perfect input singling from one party, and we leave an analytical investigation of this to future work. 

The remainder of the manuscript is structured as follows. We review the necessary background in \cref{sec:background}. In \cref{sec:sig_sec}, we introduce the set of LHV correlations under the noisy channel model and study its properties, including connections to general definitions of parameter and measurement dependence. In \cref{sec:exposed} we highlight certain facet inequalities that are violated by non-signaling quantum correlations, and in \cref{sec:rand} we apply them to DI randomness certification. We conclude with a discussion in \cref{sec:disc}. All proofs can be found in the Appendix.

\section{Background}
\label{sec:background}       
The minimal Bell scenario comprises of two parties, Alice and Bob, who are independently asked questions by a referee, denoted by random variables $X$ and $Y$ that take values $x \in \{0,1\}$ and $y \in \{0,1\}$, respectively. To these questions they need to return answers denoted by random variables $A$ and $B$ that take values $a \in \{0,1\}$ and $b \in \{0,1\}$, respectively. The set of probabilities $\p(a,b,x,y)$ is a \emph{non-signaling local behavior} (or admits a non-signaling LHV model) if it is of the form
\begin{equation}
    \p(a,b,x,y) = \p(x,y)\sum_{\lambda} \p(\lambda) \p(a|x,\lambda) \p(b|y,\lambda), \label{eq:p_classical}
\end{equation}
where $\Lambda$ is a common cause that takes values $\Lambda = \lambda$ according to the distribution $\p(\Lambda)$. The distribution over $XY$ is assumed to be fixed and uniform, $\p(x,y) =\p(x)\p(y) = 1/4$ for all $x$ and $y$, by the referee. The local description \eqref{eq:p_classical} takes into account three assumptions: $(i)$ $X$ and $Y$ are independent of $\Lambda$ (measurement independence), $(ii)$ $A$ is independent of $Y$ and $B$ is independent of $X$ (parameter independence) and $(iii)$ $A$ and $B$ are independent when conditioned on $XY\Lambda$ (outcome independence).

A behavior is a \emph{non-signaling quantum behavior} if there exists a pair of Hilbert spaces $\mathcal{H}_{Q_{A}}$ and $\mathcal{H}_{Q_{B}}$, sets of POVMs $\{M_{a|x}\}_{a\in \{0,1\}}$ and $\{N_{b|y}\}_{b \in \{0,1\}}$ for $x,y\in \{0,1\}$ on $\mathcal{H}_{Q_{A}}$ and $\mathcal{H}_{Q_{B}}$, respectively, and a joint density operator $\rho$ on $\mathcal{H}_{Q_{A}} \otimes \mathcal{H}_{Q_{B}}$ such that
\begin{equation}
    \p(a,b,x,y) = \p(x,y)\tr[(M_{a|x} \otimes N_{b|y})\rho]. \label{eq:p_quantum}
\end{equation}
For simplicity, we take Alice's and Bob's measurements to be projective throughout this work\footnote{This can be taken without loss of generality in the context of cryptography (see \cref{sec:rand}), however, more care should be taken in the context of self-testing~\cite{Baptista2025}. In this work, we explicitly assume projective measurements for our self-testing claim.}. Bell's theorem then states that there exist behaviors of the form  \eqref{eq:p_quantum} that cannot be represented in the form \eqref{eq:p_classical}~\cite{Bell}. 

A Bell inequality witnesses this separation, and is given by a tuple $(f,\eta)$ where $f$ is a linear functional of the behavior  $\mathrm{P}$ (see \eqref{eq:behav}) and $\eta \in \mathbb{R}$, such that $f(\mathrm{P}_{\mathrm{c}}^{\text{NS}}) \leq \eta$ for all non-signaling local behaviors $\mathrm{P}_{\mathrm{c}}^{\text{NS}}$, but there exists a non-signaling quantum behavior  $\mathrm{P}_{\mathrm{q}}^{\text{NS}}$ such that $f(\mathrm{P}_{\mathrm{q}}^{\text{NS}}) > \eta$. An important example is the CHSH inequality~\cite{CHSH}, given by $(f_{\mathrm{CHSH}},2)$ where
\begin{equation}
    f_{\text{CHSH}}(\mathrm{P}) = \langle A_{0}(B_{0} + B_{1})\rangle + \langle A_{1}(B_{0} - B_{1})\rangle. \label{eq:chsh}
\end{equation}
Here, $\langle A_{x}B_{y} \rangle = \sum_{a,b\in \{0,1\}} (-1)^{a+b}\p(a,b|x,y)$ are the correlators\footnote{The notation $\langle A_{x}(B_{0} \pm B_{1})\rangle$ represents $\langle A_{x}B_{0}\rangle \pm \langle A_{x}B_{1} \rangle$.}. The CHSH inequality is the only non-trivial facet of the local polytope in the minimal Bell scenario up to relabelings (see, e.g., Ref.~\cite{Brunner_review}), and it can be violated up to a value of $2\sqrt{2}$ by a maximally entangled pair of qubits and a pair of maximally incompatible qubit measurements per party. Up to local isometries, this is the only quantum strategy that can achieve $f_{\mathrm{CHSH}}(\mathrm{P}_{\text{q}}^{\text{NS}}) = 2\sqrt{2}$, and we say $f_{\mathrm{CHSH}}$ \textit{self-tests} this quantum state and measurements~\cite{MayersYao,McKagueSinglet,YangSelfTest} (see Ref.~\cite{SupicSelfTest} for an overview).  

As noted in \eqref{eq:chsh}, it will be convenient to work with the conditional behavior $\p(a,b|x,y) = \p(a,b,x,y)/\p(x,y)$, which we represent as a matrix
\begin{equation}
\mathrm{P} :=\left(\begin{array}{cc|cc}
\p(00|00) & \p(01|00) & \p(00|01) & \p(01|01)\\
\p(10|00) & \p(11|00) & \p(10|01) & \p(11|01)\\
\hline
\p(00|10) &\p(01|10) & \p(00|11) & \p(01|11) \\
\p(10|10) & \p(11|10) & \p(10|11) & \p(11|11)\\
\end{array}\right). \label{eq:behav}
\end{equation}
Parameter independence can be equivalently stated in terms of non-signaling constraints, which are satisfied if the scalar products of P with the following matrices are zero: 

\begin{equation*}
    \begin{aligned}
        \mathrm{NS}_{0}^{ \A \to \B} \coloneqq \left(
\begin{array}{cc|cc}
 1 & 0 & 0 & 0 \\
 1 & 0 & 0 & 0 \\\hline
 -1 & 0 & 0 & 0 \\
 -1 & 0 & 0 & 0 \\
\end{array}
\right), \mathrm{NS}_{1}^{\A \to \B} \coloneqq \left(
\begin{array}{cc|cc}
 0 & 0 & 1 & 0 \\
 0 & 0 & 1 & 0 \\\hline
 0 & 0 & -1 & 0 \\
 0 & 0 & -1 & 0 \\
\end{array}
\right), \\
  \mathrm{NS}_{0}^{\B \to \A} \coloneqq \left(
\begin{array}{cc|cc}
 1 & 1 & -1 & -1 \\
 0 & 0 & 0 & 0 \\\hline
 0 & 0 & 0 & 0 \\
 0 & 0 & 0 & 0 \\
\end{array}
\right),  \mathrm{NS}_{1}^{\B \to \A} \coloneqq \left(
\begin{array}{cc|cc}
 0 & 0 & 0 & 0 \\
 0 & 0 & 0 & 0 \\\hline
 1 & 1 & -1 & -1 \\
 0 & 0 & 0 & 0 \\
\end{array}
\right) .
\end{aligned}
\label{Equation::NSConstraints}
\end{equation*}
A behavior  $\mathrm{P}$ satisfying $\left<\mathrm{NS}_{y}^{\A \to \B},\mathrm{P}\right>=0$ is said to be non-signaling from Alice to Bob with respect to Bob's measurement setting $Y=y$. Similarly, if $\left<\mathrm{NS}_{x}^{\B \to \A},\mathrm{P}\right>=0$, $\mathrm{P}$ is non-signaling from Bob to Alice, with respect to Alice's measurement setting $X=x$.

Conversely, a non-zero scalar product between $\mathrm{P}$ and at least one of the non-signaling vectors implies that parameter independence is violated. If the scalar product with either $\mathrm{NS}_{y}^{\A \to \B}$ or $\mathrm{NS}_{x}^{\B \to \A}$ is non-zero for some $x,y \in \{0,1\}$, but not both, either Alice's input $x$ influences Bob's output $b$ or Bob's input $y$ influences Alice's output $a$. We call behaviors with this property \textit{one-way signaling}. If for some $x,y\in \{0,1\}$ the scalar product is non-zero with both $\mathrm{NS}_{y}^{\A \to \B}$ and $\mathrm{NS}_{x}^{\B \to \A}$, Alice's input $x$ influences Bob's output $b$ and simultaneously Bob's input $y$ influences Alice's output $a$. We call behaviors with this property \textit{two-way signaling}.  We denote the largest scalar product of $\mathrm{P}$ with $\mathrm{NS}_{y}^{ \A \to \B}$ over $y\in \{0,1\}$ by $\epsilon^{\A \to \B}$, and similarly denote the largest scalar product with $\mathrm{NS}_{x}^{ \B \to \A}$ over $x\in \{0,1\}$ by $\epsilon^{\B \to \A}$. Therefore, any behavior  $\mathrm{P}$ satisfies, 

\begin{align}\begin{split}
    &\left| \sum_{a} \p(a,b|x,y) - \sum_{a} \p(a,b|x',y)\right| \leq \epsilon^{\A \to \B}  \text{ and }\\
     &\left| \sum_{b} \p(a,b|x,y) - \sum_{b} \p(a,b|x,y')\right| \leq \epsilon^{\B \to \A} \\
    \end{split}
    \label{Def::GenPD}
\end{align}
for all $a,b,x,y,x'$ and $y'$. We denote by $\mathbf{S}_{(\epsilon^{\A \to \B},\epsilon^{\B \to \A})}$ the set of measurement and outcome independent LHV behaviors of the form
\begin{equation*}
     \p(a,b|x,y) = \sum_{\lambda}\p(\lambda)\p(a|x,y,\lambda)\p(b|x,y,\lambda)
\end{equation*}
 admitting~\eqref{Def::GenPD}. Furthermore, if we replace $\p(\lambda)$ with $\p(\lambda|x,y) = \p(\lambda)\p(x,y|\lambda) / \p(x,y)$ where $l \leq \p(x,y|\lambda)\leq 1/4$ and $0 \leq l \leq 1/4$, the LHV model will also be measurement dependent, where the amount of dependence is quantified by $l$~\cite{Putz_2016}. We denote by $\mathbf{S}_{(\xi^{\A \to \B},\xi^{\B \to \A})}^l$ the set of all measurement dependent LHV behaviors of this form satisfying
 \begin{equation}
 \begin{aligned}
     \big | \p(b|x,y,\lambda) - \p(b|x',y,\lambda) \big | &\leq \xi^{\A \to \B} \ \text{and} \\
    \big | \p(a|x,y,\lambda) - \p(a|x,y',\lambda) \big | &\leq \xi^{\B \to \A} \label{eq:lam_dep}
\end{aligned}
 \end{equation}
 for all $\lambda,a,b,x,y,x'$ and $y'$. Finally, in the expressions above, note that conditioned on the variables $X,Y$ and $\Lambda$, the joint probability of Alice and Bob factorizes. In particular, $\p(a,b|x,y,\lambda) = \p(a|x,y,\lambda)\p(b|x,y,\lambda)$. This results from the LHV model admitting outcome independence. We will hold on to this assumption throughout the paper.

\section{Binary channel model for partial signaling} \label{sec:sig_sec}

To introduce parameter dependence in the minimal Bell scenario, we consider a noisy communication channel between Alice and Bob that leaks information about one party's input to the other before any measurements take place, as illustrated for one-way signaling in \cref{fig:scenario}. In the general case, Alice's input is a pair of bits $(X,\tilde{Y})$, where $\tilde{Y}$ is a noisy copy of Bob's input $Y$, whose distribution is determined by the channel: $\{\text{p}_{\tilde{Y}|Y}(\tilde{y}|y) \ : \ \tilde{y} \in \{0,1\}\}$. Similarly, $\{\text{p}_{\tilde{X}|X}(\tilde{x}|x) \ : \ \tilde{x} \in \{0,1\}\}$ determines Bob's additional information $\tilde{X}$ about Alice's input\footnote{We include subscripts on $\p$, e.g., $\text{p}_{\tilde{X}}$, to indicate the random variable over which $\p$ is a distribution; they will be omitted when it is clear from context.}. In general, these channels need not be symmetric, allowing for a different crossover probability for each symbol and for each party. 

\subsection{Vertices of the LHV polytope}\label{sec:verts}

If Alice and Bob share a local hidden variable $\Lambda$, and output deterministic functions of their inputs and $\Lambda$, they realize a \textit{signaling local behavior }:
\begin{multline}
    \p(a,b,x,y) = \p(x,y)\sum_{\lambda}\p_{\Lambda}(\lambda) \sum_{\tilde{y},\tilde{x}}\p_{\tilde{Y}|Y}(\tilde{y}|y) \, \p_{\tilde{X}|X}(\tilde{x}|x) \\ \cdot \p_{A|X\tilde{Y}\Lambda}(a|x,\tilde{y},\lambda) \, \p_{B|\tilde{X}Y\Lambda}(b|\tilde{x},y,\lambda),  \label{eq:p_sig_model}
\end{multline}
where $\p_{A|X\tilde{Y}\Lambda}$ and $\p_{B|\tilde{X}Y\Lambda}$ are Alice's and Bob's local response functions, respectively. Notice that if $\p_{A|X\tilde{Y}\Lambda}$ and $\p_{B|\tilde{X}Y\Lambda}$ are independent of $\tilde{Y}$ and $\tilde{X}$, i.e., the parties ignore their additional input, \eqref{eq:p_sig_model} reduces to \eqref{eq:p_classical}. The same happens if $\p_{\tilde{Y}|Y}(\tilde{y}|y) =\p_{\tilde{X}|X}(\tilde{x}|x) = 1/2$ for all $\tilde{y}, \, y, \, \tilde{x}, \, x \in \{0,1\}$, i.e., the channel transmits no information about the other party's input. 

We now observe that the set of behaviors \eqref{eq:p_sig_model} is a polytope. Each response function  can be decomposed into a convex sum of deterministic functions, 
\begin{equation*}
    \p_{A|X\tilde{Y}\Lambda}(a|x,\tilde{y},\lambda) = \sum_{k} \mu(k|\lambda) \, \delta_{a,f_{k}(x,\tilde{y})},
\end{equation*}
where $k$ indexes one of finitely many deterministic functions $f_{k}:\{0,1\}^2 \to \{0,1\}$, and $\{\mu(k|\lambda)\}_{k}$ is a distribution over this index that may in general depend on the local hidden variable $\Lambda=\lambda$. A similar expression holds for Bob's response functions, and we label the associated distributions $\{\nu(k|\lambda)\}_{k}$. We can therefore rewrite \eqref{eq:p_sig_model}:
\begin{multline}
    \p(a,b,x,y) = \p(x,y)\sum_{k,l}q(k,l) \\ \cdot \Bigg(\sum_{\tilde{y}}\p_{\tilde{Y}|Y}(\tilde{y}|y) \, \delta_{a,f_{k}(x,\tilde{y})}\Bigg)\Bigg(\sum_{\tilde{x}}\p_{\tilde{X}|X}(\tilde{x}|x) \, \delta_{b,f_{l}(\tilde{x},y)}\Bigg) \label{eq:p_sig_model_new}
\end{multline}
where $q(k,l) = \sum_{\lambda}\p_{\Lambda}(\lambda)\mu(k|\lambda)\nu(l|\lambda)$. Notice that \eqref{eq:p_sig_model_new} is a convex combination of deterministic local response functions that are randomized over possible outputs of the noisy communication channel. These form a finite set of vertices, and hence for every choice of channel, the set of behaviors \eqref{eq:p_sig_model} forms a polytope, which we label $\mathbb{S}$. 

To illustrate this, we return to \cref{fig:scenario}, where a binary symmetric channel with crossover probability $1-p$ sends a noisy copy of Alice's input to Bob. We therefore have 
\begin{equation*}
    \begin{gathered}
    \p_{\tilde{Y}|Y}(\tilde{y}|y) = 1/2 \ \ \forall \tilde{y},y\in \{0,1\}  \ \ \text{and}\\
    \p_{\tilde{X}|X}(\tilde{x}|x) = p \, \delta_{\tilde{x},x} + (1-p)\, \delta_{\tilde{x} \oplus 1,x},
    \end{gathered}
\end{equation*}
i.e., Alice receives no information about Bob's input and Bob receives a copy of Alice's input $X$ with probability $p \in [1/2,1)$, and the flipped bit $X \oplus 1$ with probability $1-p$. Now consider the choice of deterministic functions 
\begin{equation*}
    \begin{aligned}
        f_{A}(x,\tilde{y}) = x, \ \ \ \text{and} \ \ \ f_{B}(\tilde{x},y) = \tilde{x} \oplus y,
    \end{aligned}
\end{equation*}
that is, Alice outputs $A=X$, while Bob outputs $B = \tilde{X} \oplus Y$. These are equivalently described by the following tables, respectively:
\begin{equation}
    \begin{array}{c|c||c}
    X & \tilde{Y} & A  \\ \hline
    0 & 0 & 0 \\
    0 & 1 & 0  \\
    1 & 0 & 1  \\
    1 & 1 & 1  \\
\end{array}  \quad \text{and} \quad
\begin{array}{c|c||c}
    \tilde{X} & Y & B  \\ \hline
    0 & 0 & 0 \\
    0 & 1 & 1  \\
    1 & 0 & 1  \\
    1 & 1 & 0  \\
\end{array} . \label{eq:fun_tab}
\end{equation}
Following \eqref{eq:p_sig_model_new}, the induced conditional behavior takes the form
$$
\text{P} = \left(
\begin{array}{cc|cc}
 p & 1-p & 1-p & p \\
 0 & 0 & 0 & 0 \\ \hline
 0 & 0 & 0 & 0 \\
 1-p & p & p & 1-p \\
\end{array}
\right).
$$
We can see that $\epsilon^{\B \to \A} = 0$, i.e., there is no signaling from Bob to Alice, while $\epsilon^{\A \to \B} = 2p-1$, reflecting the fact that Alice can partially signal to Bob constrained by the parameter $p$. The remaining vertices of this polytope are generated similarly by choosing different deterministic functions constructed from truth tables such as \eqref{eq:fun_tab}.

The procedure used to derive the behavior $\mathrm P$ above can generate local strategies under the most general form of two-way signaling, where, for example, the communication fidelity of the symbol $X=0$ need not be the same as $X=1$. In this work, we consider the fidelities of $X=0$, $X=1$, $Y=0$ and $Y=1$ to be $p,q,r$ and $s$, respectively, with $p,q,r,s \in [0,1]$. Specifically,
\begin{equation*}
\begin{aligned}
    &\p_{\tilde{X}|X}(0|0) = p, \ \ \ \p_{\tilde{X}|X}(1|1) = q \\
    &\p_{\tilde{Y}|Y}(0|0) = r \ \ \text{and} \ \ \p_{\tilde{Y}|Y}(1|1) = s.
\end{aligned}
\end{equation*}
We denote by $\mathbb{S}_{(p,q),(r,s)}$ the corresponding polytope of all distributions obtained by this channel model. In Appendix~\ref{Appendix::FullSigVerts}, we present the vertices of this polytope up to equivalence of local relabelings for both one-way and two-way signaling scenarios. 

\subsection{Facets of the LHV polytope}

In \cref{app:all_facets}, we present the facets enumerated for certain families of the parameters $p,q,r$ and $s$. Note that since the vertices are parameterised by continuous variables $(p,q,r,s)$, one cannot use facet enumeration software to generate the facets directly. In this work, we have developed an interpolation technique for doing so. First, for certain fixed values of the parameters, we used the facet enumeration software PANDA~\cite{LORWALD2015297} to generate the facets. Then, for each fixed set of values for these parameters, we grouped them into equivalent relabeling classes and ordered them according to the set of extremal local strategies that saturate them. This helps us to track how a facet inequality changes with changing values of the parameters. From this, we guess the analytic form of these facets and ensure that the analytic versions of the set of local extremal strategies saturate them. By convexity, every local strategy must also satisfy the inequality. Finally, since in every case the number of local extremals that saturate each facet inequality is greater than or equal to the dimension of the probability space, we ensure that we have the correct form of the facet inequality. In \cref{Table::FacetClassification}, we provide a quantitative summary of the facets derived for each signaling polytope considered. 

\begin{widetext}

        \begin{table}[h]
    \centering 
    \begin{tabular}{|c|c|c|c|c|c|} \hline
        Model & Description & $\#$ Classes  & $\#$ Interesting Classes & Reference \\ \hline \hline
     \vphantom{\Big(}  $\mathbb{S}_{(p,p),(1/2,1/2)}$  \vphantom{\Big(}& $\substack{\A \to \B \\ \text{Equal weights}}$    & 13  & 2 & \vphantom{\Big(}  Appendix~\ref{Appendix::one_way}  \vphantom{\Big(} \\ \hline
        \vphantom{\Big(}  $\mathbb{S}_{(p,p),(p,p)} $  \vphantom{\Big(} & $\substack{\A \to \B \text{ and } \B \to \A \\ \text{Equal weights}}$    &  15 &  3 & \vphantom{\Big(} Appendix~\ref{Appendix::G_One_Way} \vphantom{\Big(} \\ \hline
           $\mathbb{S}_{(p,1),(r,1)}$ & $\substack{\A \to \B \text{ and } \B \to \A \\ \text{Varying weights} \\ \text{Fixed bit per party perfectly sent}}$ &  29  & 20 & Appendix~\ref{Appendix::one_symbol_perfect}\\ \hline
           $\mathbb{S}_{(p,1)^*,(p,1)^*}$  & $\substack{\A \to \B \text{ and } \B \to \A \\ \text{Varying weights} \\ \text{Random bit per party perfectly sent}}$  &  1043  &  5  & Appendix~\ref{Appendix::random_symbol_perfect}\\ \hline
    \end{tabular}
    \caption{$\mathbb{S}_{(p,q),(r,s)}$ denotes the set of behaviors obtained using local strategies and a noisy channel model where the symbols $X=0$, $X=1$, $Y=0$ and $Y=1$ are sent with fidelities $p,q,r$ and $s$ respectively. $\#$ Classes denote the number of equivalence classes up to local relabelings and $\#$ Interesting Classes denote the number of equivalence classes for which there exists numerical evidence that non-signaling quantum correlations violate them.}
    \label{Table::FacetClassification}
\end{table}

\end{widetext}

\subsection{Channel model corresponding to general parameter independence} \label{sec:strong_sig}

A particular variation of the noisy channel model of interest is when one of Alice's symbols, say $X=1$, is always perfectly transmitted to Bob, but the other symbol $X=0$ is transmitted with a fidelity $p\in[0,1)$. Similarly, Bob's symbol $Y=0$ is transmitted to Alice with fidelity $r\in[0,1)$, while the symbol $Y=1$ is transmitted perfectly. The local polytope in this scenario is denoted $\mathbb{S}_{(p,1),(r,1)}$, and we present its facets in \cref{Appendix::one_symbol_perfect}. Note that choosing $p = r = 0$ recovers the non-signaling scenario, since a constant symbol ($X=1$ for Alice and $Y=1$ for Bob) is transmitted to the other party, regardless of the actual input. 

In a more general version of this signaling scenario, one could consider mixing different combinations of $\mathbb{S}_{(p,1), (r,1)}, \, \mathbb{S}_{(p,1),(1,r)}, \, \mathbb{S}_{(1,p),(r,1)}$ and $\mathbb{S}_{(1,p),(1,r)}$, i.e., mixing which of the two symbols will be perfectly transmitted to the other party, while the remaining symbol is transmitted with fidelity $p$ and $r$ for Alice and Bob, respectively. We denote this set by

\begin{multline}
    \mathbb{S}_{(p,1)^*,(r,1)^*} = \mathrm{ConvHull}\Big[ \mathbb{S}_{(p,1), (r,1)} \\ \cup \mathbb{S}_{(p,1),(1,r)} \cup  \mathbb{S}_{(1,p),(r,1)} \cup \mathbb{S}_{(1,p),(1,r)}\Big]. \label{eq:poly_star}
\end{multline}
$\mathbb{S}_{(p,1)^*,(r,1)^*}$ is a polytope whose vertices are equal to the union of the vertices of each polytope in the convex hull. In \cref{Appendix::random_symbol_perfect} we provide two of its facet inequalities for which we have numerical evidence that they are violated by a non-signaling quantum behavior for all $p,r<1$ (cf. \cref{sec:2way_exp}). 

We now prove that the polytope $\mathbb{S}_{(p,1)^*,(r,1)^*}$ is equal to the set of LHV behaviors achievable under a general relaxation of parameter independence, $\mathbf{S}_{(\epsilon^{\A \to \B},\epsilon^{\B \to \A})}$ (see \cref{Def::GenPD}), when setting $p=\epsilon^{\B \to \A}$ and $r=\epsilon^{\A \to \B}$. This is formally stated below. 

\begin{lemma} Let $p,r\in[0,1)$, and $\mathbb{S}_{(p,1)^*,(r,1)^*}$ be the set of LHV behaviors obtained under the asymmetric signaling channel, as defined in \cref{eq:poly_star}. Let $\mathbf{S}_{(p,r)}$ denote the set of LHV behaviors obtained under a general relaxation of parameter independence, as defined below \cref{Def::GenPD}. Then
    \begin{equation*}
        \mathbb{S}_{(p,1)^*,(r,1)^*} = \mathbf{S}_{(p,r)}. 
    \end{equation*}
\label{Theorem:modelconnection}
\end{lemma}

\noindent See \cref{Proof::Theorem_Model_connection} for proof.

\subsection{Connection to joint relaxations of measurement and parameter independence}
So far, we have focused on relaxations of parameter independence. However, given outcome independence holds, our techniques apply to relaxations of measurement independence or a joint relaxation of both measurement and parameter independence. This is due to the fact that the set of behaviors $\p(A,B|X,Y)$ obtained from a measurement dependent but parameter and outcome independent LHV model can also be realized by a LHV model admitting parameter dependence but measurement and outcome independence. Since the set of such behaviors is a convex polytope and the fact that the extremal behaviors in the set $\mathbf{S}_{(\epsilon^{\A \to \B},\epsilon^{\B\to\A})}$ are either non-signaling and local, or signaling and non-local, the extremal behaviors obtained upon relaxing measurement independence only must also have this property. A stronger statement can also be made for a joint relaxation of both measurement and parameter independence, as stated below.

\begin{lemma} \label{Lemma::MD_PD_and_PD}
    Let $A,B,X, Y$ and $\Lambda$ be random variables where $A,B,X$ and $Y$ are binary. Then for any behavior $\p(A,B|X,Y)$ admitting
    \begin{equation*}
        \p(a,b|x,y) = \sum_{\lambda} \p_{\Lambda}(\lambda|x,y)\p_{A}(a|x,y,\lambda)\p_{B}(b|x,y,\lambda),
    \end{equation*}
    where $\p_{\Lambda},\,\p_{A}$ and $\p_{B}$ are distributions, there exists a random variable $\Lambda'$ and tuple of distributions $(\p_{\Lambda'}',\p_{A}',\p_{B}')$ such that 
    \begin{equation*}
        \p(a,b|x,y) = \sum_{\lambda'} \p'_{\Lambda'}(\lambda')\p'_{A}(a|x,y,\lambda')\p'_{B}(b|x,y,\lambda').
    \end{equation*}
\end{lemma}

\noindent See Appendix~\ref{Appendix::MD_PD_and_PD} for proof. 

Lemma~\ref{Lemma::MD_PD_and_PD} shows that the set of behaviors upon joint relaxations of measurement and parameter independence coincides with another set of behaviors where only parameter independence is relaxed. It is then natural to ask what would be the effective signaling in this set as a function of the parameters for measurement and  parameter independence relaxations. We provide an analytical formula for this next.

\begin{lemma} \label{Lemma::jointbound}    
    Let $0\leq l \leq 1/4$ and $0 \leq \epsilon \leq 1$. Let $\mathbf{S}^{l}_{(\epsilon,\epsilon)}$ be the set of behaviors over binary random variables $X,Y,A$ and $B$ with relaxed measurement and parameter dependence parameters $l$ and $\xi^{\A \to \B}=\xi^{\B \to \A}=\epsilon$, respectively, as defined above Eq~\eqref{eq:lam_dep}. Then $\mathbf{S}^{l}_{(\epsilon,\epsilon)} \subseteq \mathbf{S}_{(\kappa,\kappa)}$, with 
\begin{equation*}
    \kappa = \epsilon +(1-\epsilon)\frac{1-4l}{1-2l}.
\end{equation*}
\end{lemma}

\noindent See Appendix~\ref{Appendix::MD_PD_and_PD} for proof. 

Note the following consistencies: when either $l= 0$ or $\epsilon =1$, i.e., either measurement or parameter independence is perfectly relaxed, $\kappa=1$ in the effective model, leading to the set of all behaviors. On the other hand, when $l = 1/4$ and $\epsilon=0$, $\kappa=0$, meaning that one recovers the non-signaling LHV set when neither measurement or parameter independence is relaxed.

\subsection{Non-signaling subspaces} \label{sec:NS_sub}
Since the set of non-signaling quantum behaviors is a strict subset of the set of all non-signaling behaviors, it might be relevant to understand whether there are non-signaling behaviors that can be realized by LHV models in a two-way signaling scenario, but cannot be realized with one-way signaling. Conversely, if a one-way signaling scenario can realize the same set of non-signaling behaviors as a two-way scenario, it would show that this particular two-way scenario is redundant when it comes to simulating quantum behaviors. In this sub-section, we claim this to be the case for the two-way scenario $\mathbb{S}_{(p,p),(p,p)}$. We do so by providing an algorithm to calculate the largest non-signaling set realizable within a signaling polytope $\mathbb{S}$, which we now describe.

In the minimal Bell scenario, all non-signaling distributions lie on the intersection of the  hyperplanes defined by the vectors $ \mathrm{NS}_{0}^{ \A \to \B}$, $ \mathrm{NS}_{1}^{ \A \to \B}$, $ \mathrm{NS}_{0}^{ \B \to \A}$ and  $ \mathrm{NS}_{1}^{ \B \to \A}$. Now, for each behavior  $\mathrm{P}$, there are 16 probabilities, knowing all of which is not necessary. The normalization condition implies that within each block (that is, a pair of inputs for Alice and Bob, and corresponding quadrant of the matrix \eqref{eq:behav}), the probabilities sum up to 1. This reduces the dimension by 4. Now, if a probability distribution is non-signaling, then the sum of the first two probabilities in every row and column is the same as the sum of the last two probabilities. This further reduces the dimension by 4. Therefore, a signaling probability distribution must have true dimension between 9 and 12. In particular, the set containing this distribution can be seen as an embedding in $\mathbb{R}^d$ where $9 \leq d \leq 12$. Given such a set, one can look at its projection on the $d-1$ hyperplane described by one of the non-signaling constraints above. If $d-1 \neq 8$, one can then further project it down by another non-signaling constraint, and continue until the dimension reaches 8~\footnote{One might lose distributions if in the first step one projects down to $d-2$.}. This intuition can be formalized into the following algorithm.

\begin{itemize}
    \item \textbf{Step 1:} Take the set of vertices, $\mathrm{Vert}\left[\mathbb{S}\right]$, of the set $\mathbb{S}$ and pick a non-signaling constraint vector, say $ \mathrm{NS}_{0}^{ \A \to \B}$. 
    \item \textbf{Step 2:} Take  $\mathbb{S}^*  \subseteq \mathrm{Vert}\left[\mathbb{S}\right]$ such that for every $\mathrm{P} \in \mathbb{S}^*$, $\left<\mathrm{NS}_{0}^{ \A \to \B},\mathrm{P}\right> \neq 0$.
    \item \textbf{Step 3:} Collect all pairs $(\mathrm{P}_j,\mathrm{P}_k)$, such that $\mathrm{P}_{j},\mathrm{P}_{k} \in \mathbb{S}^*$. 
    \item \textbf{Step 4:} For every pair $(\mathrm{P}_j,\mathrm{P}_k)$, define a line segment $\tilde{\mathrm{P}}_{j,k,\alpha} \coloneqq \alpha \mathrm{P}_i + (1-\alpha) \mathrm{P}_j$ with $0 \leq \alpha \leq 1$ and solve for $\alpha'$ such that $\left<\mathrm{NS}_{0}^{ \A \to \B},\tilde{\mathrm{P}}_{j,k,\alpha'}\right> = 0$. Define the set of all such intersection points as $\Pi_{\mathrm{NS}_{0}^{ \A \to \B}}\left[\mathbb{S}^*\right]$.  Ignore cases when no solutions exist. 
    \item \textbf{Step 5:} Calculate the set 
    $$
    \big(\mathrm{Vert}\left[\mathbb{S}\right] \setminus \mathbb{S}^*\big) \bigcup \Pi_{\mathrm{NS}_{0}^{ \A \to \B}}\left[\mathbb{S}^*\right].
    $$
    \item \textbf{Step 6:} If this set is signaling, repeat Step 1 through Step 4 starting with this set and a different non-signaling constraint vector.
\end{itemize}

\noindent We denote by $\Pi_{\mathrm{NS}}[\mathbb{S}]$ the non-signaling subset of $\mathbb{S}$. We now claim that there exists signaling channel models where two-way signaling does not provide leverage in simulating non-signaling correlations over one-way signaling. 

\begin{claim}     
    Let $p,q,r,s\in [1/2,1)$ and $\mathbb{S}_{(p,q),(r,s)}$ be the set of LHV behaviors obtained under a pair of signaling channels as defined in \cref{sec:verts}, and let $$\mathbb{S}_{\mathrm{1-W}}^p \coloneqq \mathrm{ConvHull}\left[\mathbb{S}_{(p,p),(1/2,1/2)} \  \cup \ \mathbb{S}_{(1/2,1/2),(p,p)} \right].$$
    Then $ \Pi_{\mathrm{NS}}\left[\mathbb{S}_{\mathrm{1-W}}^p\right] = \Pi_{\mathrm{NS}}\left[\mathbb{S}_{(p,p),(p,p)}  \right].$   
    \label{theorem:nsprojection}
\end{claim}
\noindent See Appendix~\ref{Appendix::NS_Subspace} for a proof sketch.

An implication of Claim \ref{theorem:nsprojection} is that if there exists a non-signaling quantum strategy that lies outside the one-way signaling polytope $\mathbb{S}_{\mathrm{1-W}}^p$, it must also lie outside the two-way signaling polytope $\mathbb{S}_{(p,p),(p,p)}$. Furthermore, the existence of an analogous connection for different signaling models can also be investigated using the algorithm provided, which we leave for future investigation. 

\section{Exposed non-signaling quantum correlations} \label{sec:exposed}

In this section, we focus on two different signaling scenarios from \cref{sec:sig_sec}, and show that for both, there are non-signaling quantum correlations that cannot be simulated locally.  

\subsection{Symmetric one-way signaling}

Recall the symmetric one-way signaling scenario from Alice to Bob in \cref{fig:scenario}, with crossover probability $1-p$. Specifically, the LHV polytope of interest in this setting is denoted $\mathbb{S}_{(p,p),(1/2,1/2)}$ and its facets can be found in \cref{Appendix::one_way}. One facet of particular interest (denoted $\text{F}_{13}$ in \cref{Appendix::one_way}) is given by $\langle \mathrm{P}, F_{p}\rangle \leq 1$, where
$$
F_{p} = \left(
\begin{array}{cc|cc}
 0 & 0 & 1 & 0 \\
 0 & 1 & 0 & 0 \\ \hline
 0 & -\frac{1-p}{p} & 0 & \frac{1-p}{p} \\
 0 & 0 & 0 & 0 \\
\end{array}
\right).
$$
Let us define the following correlators:
\begin{equation*}
    \begin{aligned}
        \langle A_{x,y} \rangle &= \sum_{a,b\in \{0,1\}} (-1)^{a}\p(a,b|x,y) \\
        \langle B_{x,y} \rangle &= \sum_{a,b\in \{0,1\}} (-1)^{b}\p(a,b|x,y)\\
        \langle A_{x}B_{y} \rangle &= \sum_{a,b\in \{0,1\}} (-1)^{a+b}\p(a,b|x,y).
    \end{aligned}
\end{equation*}
Note that for $x,y\in \{0,1\}$, the constraints $\langle A_{x,0} \rangle = \langle A_{x,1} \rangle \equiv \langle A_{x}\rangle $ and $\langle B_{0,y} \rangle = \langle B_{1,y} \rangle \equiv \langle B_{y} \rangle$ are equivalent to the non-signaling constraints. For non-signaling quantum behaviors, we define the observables $A_{x} = \sum_{a}(-1)^{a}M_{a|x}$ and $B_{y} = \sum_{b}(-1)^b N_{b|y}$, which yields $\langle A_{x}\rangle = \tr[(A_{x}\otimes \id_{Q_{B}})\rho]$, $\langle B_{y}\rangle = \tr[(\id_{Q_{B}}\otimes B_{y})\rho]$ and $\langle A_{x}B_{y}\rangle = \tr[(A_{x} \otimes B_{y})\rho]$. The behavior  can be expressed in terms of correlators via the relation
\begin{multline*}
    \p(a,b|x,y) = \frac{1}{4}\Big( 1 + (-1)^{a}\langle A_{x,y} \rangle + (-1)^{b}\langle B_{x,y} \rangle \\+ (-1)^{a+b}\langle A_{x}B_{y} \rangle\Big).
\end{multline*}
A direct substitution shows that $\langle \mathrm{P},F_{p}\rangle \leq 1$ is equivalent to 
\begin{multline*}
      \langle A_{01} - A_{00}\rangle + \frac{1-p}{p} \langle A_{11} - A_{10}\rangle \\ + \langle B_{01} - B_{00} \rangle + \frac{1-p}{p}\langle B_{10} - B_{11} \rangle \\ + \langle A_{0}(B_{0} + B_{1})\rangle + \frac{1-p}{p} \langle A_{1}(B_{0} - B_{1})\rangle \leq 2.
\end{multline*}
By restricting to behaviors that are non-signaling from Bob to Alice, i.e., those that satisfy $\langle A_{x,0} - A_{x,1}\rangle = 0$, we obtain the following expression, along with some properties.   
\begin{proposition} \label{lem:self-test}
    Let $p \in (2/5,1)$ and $T_{p}$ be the Bell functional
    \begin{multline}
        T_{p}(\mathrm{P}) = \langle B_{01} - B_{00} \rangle + \frac{1-p}{p}\langle B_{10} - B_{11} \rangle \\ + \langle A_{0}(B_{0} +B_{1}) \rangle + \frac{1-p}{p} \langle A_{1}(B_{0}  - B_{1} )\rangle. \label{eq:tilt_chsh}
    \end{multline}
    Then the following properties hold:
    \begin{enumerate}
        \item For all non-signaling local behaviors $\mathrm{P}_{\mathrm{L}}^{\mathrm{NS}}$, $T_{p}(\mathrm{P}_{\mathrm{L}}^{\mathrm{NS}}) \leq 2(|2p-1|+1-p)/p$.
        \item For all one-way signaling local behaviors $\mathrm{P}_{\mathrm{L}}^{\mathrm{S}} \in \mathbb{S}_{(p,p),(1/2,1/2)}$, $T_{p}(\mathrm{P}_{\mathrm{L}}^{\mathrm{S}}) \leq 2\max\Big\{p , \Big| -p + 2\frac{(1-p)^2}{p}\Big|\Big\} + 2(1-p)$.
        \item For all non-signaling quantum behaviors $\mathrm{P}_{\mathrm{Q}}^{\mathrm{NS}}$, $T_{p}(\mathrm{P}_{\mathrm{Q}}^{\mathrm{NS}}) \leq 2\sqrt{2}\sqrt{\frac{p}{3p-1}}$.
        \item Up to local isometries, there is a unique non-signaling quantum strategy that achieves $T_{p}(\mathrm{P}_{\mathrm{Q}}^{\mathrm{NS}}) = 2\sqrt{2}\sqrt{\frac{p}{3p-1}}$, given by
        \begin{equation}
        \begin{gathered}
            \rho = \ketbra{\psi_{\theta}}{\psi_{\theta}}, \ \ket{\psi_{\theta}} = \cos(\theta) \ket{00} + \sin(\theta) \ket{11},\\
            A_{0} = \sigma_{X}, \ \ \ A_{1} = \sigma_{Z}, \\
            B_{0} = \cos(\phi) \, \sigma_{Z} + \sin(\phi) \, \sigma_{X}, \\
            B_{1} = -\cos(\phi) \, \sigma_{Z} + \sin(\phi) \, \sigma_{X},
        \end{gathered} \label{eq:qubit_strat}
        \end{equation}
        where
        \begin{equation*}
            \begin{gathered}
                \theta = \frac{1}{2} \mathrm{arctan}\big[f_{1}/g_{1}\big] - \pi/2,\ \
                \phi = \mathrm{arctan}\big[f_{2}/g_{2}\big], \\
                f_{1} = \frac{\sqrt{p(-2 + 5p)}}{1 - 3p}, \ \ g_{1} = \frac{1-2p}{-1+3p},\\
                f_{2} = -\sqrt{\frac{-2+5p}{-2+6p}}, \ \ g_{2} = \sqrt{\frac{p}{-2+6p}}
            \end{gathered}
        \end{equation*}
        when $p \in (2/5,1/2) \cup (1/2,1)$ and $\theta = \phi= -\pi/4$ when $p = 1/2$.
    \end{enumerate}
\end{proposition}
\noindent See \cref{app:proof1} for proof. Note that when $p\in [1/2,1)$, the maximum local value of $T_{p}$ is equal to 2 for both non-signaling and one-way signaling local behaviors. 

When evaluated on fully non-signaling strategies $\mathrm{P}^{\text{NS}}$, we have $\langle B_{x,y} \rangle = \langle B_{y} \rangle$, and therefore
\begin{multline}
    T_{p}(\mathrm{P}^{\text{NS}}) = \frac{1-2p}{p}\langle B_{0} - B_{1} \rangle \\ + \langle A_{0}(B_{0} +B_{1}) \rangle + \frac{1-p}{p} \langle A_{1}(B_{0}  - B_{1} )\rangle \leq 2, \label{eq:new_ineq_ns}
\end{multline}
which reduces to the CHSH inequality when $p = 1/2$. Thus, since \cref{eq:tilt_chsh} is a facet, we can view it as a generalization of the CHSH inequality to the symmetric one-way signaling scenario. Furthermore, \eqref{eq:new_ineq_ns} is a member of a family of tilted-CHSH inequalities derived in Ref.~\cite[Section 3.2.1]{Barizien_2024}. It was shown in Ref.~\cite[Appendix B.3]{Barizien_2024} that this family self-tests the partially entangled strategy in \cref{eq:qubit_strat}, and this is summarized in points 3 and 4 of \cref{lem:self-test}. Notably, the maximum value of $T_{p}$ under non-signaling quantum behaviors (achieved by the strategy \eqref{eq:qubit_strat}) is strictly larger than $2$ for all $p \in [1/2,1)$, and converges to $2$ as $p$ converges to 1. We can therefore conclude that in the presence of arbitrarily strong (but not perfect) input signaling from Alice to Bob via a binary symmetric channel, there exist non-signaling quantum correlations that cannot be produced classically.     

\subsection{Asymmetric two-way signaling} \label{sec:2way_exp}
We now consider the two-way noisy channel model considered in \cref{sec:strong_sig}, which induces the local polytope $\mathbb{S}_{(p,1),(r,1)}$. Recall that here, $X=1$ is transmitted perfectly to Bob, $X=0$ is transmitted with a fidelity $p<1$, $Y=0$ is transmitted to Alice with fidelity $r<1$ and $Y=1$ is transmitted perfectly. We now select a particular facet (5$^{\text{th}}$ on the list in \cref{Appendix::one_symbol_perfect}), given by 
$$
G_{p,r} = \left(
\begin{array}{cc|cc}
 1 & 1 & 0 & r-1 \\
 1 & 0 & 0 & 0 \\ \hline
 0 & 0 & 0 & 0 \\
 p-1 & 0 & 0 & (1-p)(1-r) \\
\end{array}
\right).
$$
We now show that this facet is in fact a witness for quantum behaviors. That is, for all $p,r<1$, there exists a quantum non-signaling behavior that cannot be realized within the two-way signaling LHV model $\mathbb{S}_{(p,1),(r,1)}$.

\begin{proposition} \label{lem:2_ways}
    Let $p,r \in [0,1)$ and $W_{p,r}$ be the Bell functional $W_{p,r}(\mathrm{P}) = \langle G_{p,r}, \mathrm{P}\rangle$, given by
    \begin{multline}
    W_{p,r}(\mathrm{P}) = 1 - \p(11|00) - (1-r)\p(01|01) \\ - (1-p)\p(10|10) + (1-p)(1-r)\p(11|11). \label{eq:tilt_chsh_2}
    \end{multline}
    Then the following properties hold:
    \begin{enumerate}
        \item For all two-way signaling local behaviors $\mathrm{P}_{\mathrm{L}}^{\mathrm{S}} \in \mathbb{S}_{(p,1),(r,1)}$, $W_{p,r}(\mathrm{P}_{\mathrm{L}}^{\mathrm{S}}) \leq 1$.
        \item There exists a non-signaling quantum behavior $P_{\mathrm{Q}}^{\mathrm{NS}}$ such that $W_{p,r}(\mathrm{P}_{\mathrm{Q}}^{\mathrm{NS}}) > 1$. 
    \end{enumerate}
\end{proposition}
\noindent See \cref{app:proof1} for proof. 

Note that when evaluated on non-signaling behaviors and $p = r = 0$, $W_{p,r}(\mathrm{P}^{\text{NS}})$ reduces to the CHSH inequality up to constants and relabeling:
\begin{equation*}
    W_{0,0}(\mathrm{P}^{\text{NS}}) = \frac{1}{2} + \frac{1}{4}\langle A_{0}(-B_{0} + B_{1})\rangle + \frac{1}{4}\langle A_{1}(B_{0} + B_{1})\rangle \leq 1.
\end{equation*}
In the same spirit as \cref{eq:tilt_chsh}, we can view \cref{eq:tilt_chsh_2} as another generalization of the CHSH facet to signaling scenarios, this time in the presence of (the more general) two-way signaling. While we were not able to analytically characterize the optimal non-signaling quantum strategy in this case, \cref{lem:2_ways} shows that such exposed points do exist even when the amount of two-way signaling is arbitrarily high. 

We conjecture that the same is true for the more general signaling scenario $\mathbb{S}_{(p,1)^*,(r,1)^*}$, which is equal to the convex hull of $\mathbb{S}_{(p,1),(r,1)}, \, \mathbb{S}_{(p,1),(1,r)}, \, \mathbb{S}_{(1,p),(r,1)}$ and $\mathbb{S}_{(1,p),(1,p)}$. 

\begin{conjecture}
    \label{conj:sig} 
    Let $\mathbb{S}_{(p,1)^*,(r,1)^*}$ denote the LHV polytope defined in \cref{eq:poly_star}. For all $p,r < 1$, there exists a non-signaling quantum behavior $P_{\mathrm{Q}}^{\mathrm{NS}}$ such that $P_{\mathrm{Q}}^{\mathrm{NS}} \notin \mathbb{S}_{(p,1)^*,(r,1)^*}$ can be witnessed by one the Bell inequalities reported in \cref{Appendix::random_symbol_perfect}. 
\end{conjecture}

\noindent Our conjecture is based on numerically searching over two-qubit strategies that violate one of the facets provided in \cref{Appendix::random_symbol_perfect}, which are expected to be optimal in the standard (non-signaling) minimal Bell scenario following Jordan's Lemma~\cite{Jordan,Barizien2025}.  

Proving that $\mathbb{S}_{(p,1)^*,(r,1)^*}$ cannot simulate all non-signaling quantum correlations would imply that the same is true for the most general model of parameter independence $\mathbf{S}_{(p,r)}$, a consequence of the equivalence between the two established in \cref{Theorem:modelconnection}. Furthermore, following \cref{Lemma::MD_PD_and_PD,Lemma::jointbound}, the equivalence between a joint relaxation of measurement and parameter independence, and an effective relaxation of parameter independence only, suggests that the set $\mathbf{S}_{(p,r)}^{l}$ does not contain the set of all non-signaling quantum correlations whenever $p,r<1$ and $l>0$. We summarize this in the following proposition. 

\begin{proposition}
    Suppose Conjecture \ref{conj:sig} holds. Then for every $p,r\in [0,1)$ and $l \in (0,1/4]$, there exists a non-signaling quantum behavior $P_{\mathrm{Q}}^{\mathrm{NS}}$ such that $P_{\mathrm{Q}}^{\mathrm{NS}} \notin \mathbf{S}_{p,r}^{l}$. That is, for any general relaxation of parameter independence and measurement independence (see \cref{Def::GenPD}), there exists a non-signaling quantum behavior that cannot be simulated using LHV models.   \label{prop:conj_prob}
\end{proposition}

\cref{prop:conj_prob} is consistent with recent reports by Refs.~\cite{Vieira_2025,ramanathan2025} who made the same conclusion using different approaches. In the former, the authors consider a Bell scenario with binary inputs per party and infinitely many outcomes per input. In the latter, the authors report a similar result with binary outcomes per input. However, neither assume an explicit signaling mechanism. The equivalences established in \cref{Theorem:modelconnection} and \cref{Lemma::MD_PD_and_PD,Lemma::jointbound} show that it is in fact sufficient to consider LHV models arising from binary signaling channels. Moreover, a positive resolution to Conjecture \ref{conj:sig} would provide an explicit Bell inequality for this LHV model that witnesses a non-signaling quantum behavior.       

\section{DI randomness certification} \label{sec:rand}

In the previous section, we showed that certain quantum correlations cannot be simulated by classical models even with near perfect input signaling. We now explore the impact of this on the certification of DI randomness. 

\subsection{DI randomness without signaling} 

To begin, we recall the cryptographic setting without signaling. We call the tuple $\mathcal{H}_{Q_{A}}$, $\mathcal{H}_{Q_{B}}$, $\{M_{a|x}\}$, $\{N_{b|y}\}$, $\rho_{Q_{A}Q_{B}}$ that realizes a non-signaling quantum behavior  a \textit{non-signaling quantum strategy}. In the cryptographic context, we must allow for an adversary Eve, who holds a system $E$ that purifies the bipartite state $\rho_{Q_{A}Q_{B}}$~\footnote{The measurement operators $M_{a|x}$ and $N_{b|y}$ act trivially on the system $E$.}. We denote the tripartite pure state by $\ket{\psi}_{Q_{A}Q_{B}E}$, and note that without loss of generality we can take Alice's and Bob's measurements to be projective according to Naimark’s dilation theorem~\cite{paulsen_2003,BhavsarDI}. Furthermore, we allow Eve to prepare the strategy however she likes, provided the induced behavior  $\mathrm{P}$ is compatible with the set of observations made by Alice and Bob, e.g., the violation of a Bell inequality, expressed as a linear functional $f(\mathrm{P}) \geq \omega > \eta$, where $\eta$ is the local bound and $\omega$ is the observed value. Eve's goal is to select the compatible strategy that allows her to learn the most information about the outcome of one of Alice's measurements, in this case the setting $X=0$. This corresponds to the security of a spot-checking protocol~\cite{MS2}, where Alice mostly chooses $X=0$ to generate randomness, and sometimes chooses a random input to test for a Bell violation with Bob (see, e.g., \cite[Section XIV.B.2]{Pirandola_2020}). The asymptotic rate of randomness generation $r_{\text{NS}}$, in bits per interaction, is given by
\begin{equation}
    r_{\text{NS}}(f,\omega) = \inf_{\substack{\mathcal{H}_{Q_{A}},\mathcal{H}_{Q_{B}}, \mathcal{H}_{E} \\ \ket{\psi}_{Q_{A}Q_{B}E}, \{M_{a|x}\},\{N_{b|y}\} \\ \text{compatible with} \  f(\mathrm{P}) \geq \omega}} H(A|X=0,E)_{\rho}
\end{equation}
where $H(A|B)_{\rho} = H(AB)_{\rho} - H(B)_{\rho}$ is the conditional von Neumann entropy with $H(A)_{\rho} = - \tr[\rho \log(\rho)]$, which is evaluated on the post-measurement state after Alice receives the input setting $X=0$, 
\begin{multline}
    \rho_{AE|X=0} = \sum_{a \in \{0,1\}} \ketbra{a}{a}_{A} \\ \otimes \tr_{Q_{A}Q_{B}}\big[(M_{a|0} \otimes \id_{Q_{B}E})\ketbra{\psi}{\psi}\big]. \label{eq:pms}
\end{multline}
Asymptotic rates can be also used as a basis for rates with finite statistics~\cite{DFR,MetgerGEAT,ADFRV}.

\subsection{DI randomness with signaling} 

To allow for noisy input signaling between parties, we slightly modify the previous formulation. For ease of presentation, we consider the setting of one-way signaling of Alice's input to Bob via a symmetric binary channel with crossover probability $1-p \in (0,1/2]$ (as illustrated in \cref{fig:scenario}). This can be straightforwardly generalized to the asymmetric two-way case. As before, we consider a tripartite state $\ket{\psi}_{Q_{A}Q_{B}E}$, and Alice's measurement on $Q_{A}$ is modeled as before, i.e., as one of two binary outcome POVMs selected by her input $X$ only, $\{M_{a|x}\}$. Bob's device however has four possible input choices; his binary input $Y$ and the bit received from the noisy channel $\tilde{X}$. We therefore consider four POVMs for Bob, $\{N_{b|\tilde{x},y}\}$. However, they should capture the fact that $\tilde{X}$ is a noisy copy of Alice's input $X$. To do this, we define Bob's noisy POVM elements that depend on $X$ rather than $\tilde{X}$ according to the crossover probability,
\begin{equation}
    \widetilde{N}_{b|x,y}^{(p)} := p \, N_{b|x,y} + (1-p) \, N_{b|x\oplus 1,y}. 
\end{equation}
When $p = 1$, $\tilde{X} = X$ and Bob's measurement can depend arbitrarily on both inputs $(X,Y)$. When $p = 1/2$ (or $N_{b|\tilde{x},y}$ is independent of $\tilde{x}$), $\widetilde{N}_{b|0,y}^{(p)} = \widetilde{N}_{b|1,y}^{(p)}$ and Bob's measurement is independent of Alice's input $X$, recovering the non-signaling setting. Intermediate values imply that $X$ can influence Bob's measurement choice, however his device could incorrectly decode $X = 0$ when $X = 1$, or vice-versa, preventing perfect signaling from taking place. We call the induced behavior a \textit{signaling quantum behavior}, which is given by
\begin{equation}
    \p(a,b|x,y) = \tr\big[(M_{a|x} \otimes \widetilde{N}_{b|x,y}^{(p)})\rho_{Q_{A}Q_{B}} \big]. \label{eq:sig_p}
\end{equation}

Note that signaling could be detected in the observed behavior  $\mathrm{P}$. Specifically, the parties can check the overlap between $\mathrm{P}$ and the non-signaling matrices, which is equal to 
\begin{equation*}
\begin{aligned}
    \left<\mathrm{NS}_{y}^{\A \to \B},\mathrm{P}\right> &= \tr\big[ \widetilde{N}_{0|0,y}^{(p)}\rho_{Q_{B}} \big]  - \tr\big[\widetilde{N}_{0|1,y}^{(p)}\rho_{Q_{B}} \big] \\ 
    &= (2p-1) \tr\big[(N_{0|0,y} - N_{0|1,y})\rho_{Q_{B}} \big].
\end{aligned}
\end{equation*}
If any of these values are non-zero, the honest parties can conclude that the non-signaling assumption is violated and abort the protocol\footnote{When the number of rounds is finite, small deviations from zero of the non-signaling inner products are an inevitable consequence of statistical fluctuations. In this case, the honest parties can permit a small amount of signaling consistent with the expected fluctuations as a function of the number of rounds. If the observed value exceeds this, they can abort the protocol.}. We therefore include the constraints 
\begin{equation*}
    \left<\mathrm{NS}_{y}^{\A \to \B},\mathrm{P}\right> = 0, \ \forall y \in \{0,1\}.
\end{equation*}
If $p > 1/2$, these constraints imply
\begin{equation*}
    \tr\big[(N_{0|0,y} - N_{0|1,y})\rho_{Q_{B}} \big] =0,
\end{equation*}
which is a weaker condition than $N_{0|0,y} = N_{0|1,y}$. Thus, Eve may still prepare Bob's device such that it exploits the additional input $\tilde{X}$, but she must do so in such a way that it is not noticeable in the observed statistics.    

With this in mind, we define the asymptotic rate of randomness generation in the presence of noisy input signaling,
\begin{equation}
    r_{\mathrm{S},p}(f,\omega) = \inf_{\substack{\mathcal{H}_{Q_{A}},\mathcal{H}_{Q_{B}}, \mathcal{H}_{E} \\ \ket{\psi}_{Q_{A}Q_{B}E}, \{M_{a|x}\},\{N_{b|\tilde{x},y}\} \\ \text{compatible with} \  f(\mathrm{P}) \geq \omega \\ \text{and} \ \langle \text{NS}_{y}^{\mathbf{A}\to \mathbf{B}}, \mathrm{P} \rangle = 0 }} H(A|X=0,E)_{\rho}, \label{eq:sig_ent}
\end{equation}
where the entropy is evaluated on the state \eqref{eq:pms}, and the dependence on the parameter $p \in [1/2,1)$ is implicit in the definition of $\mathrm{P}$ given in \eqref{eq:sig_p}. We then have $r_{\mathrm{S},1/2}(f,\omega) = r_{\mathrm{NS}}(f,\omega)$. 

The setup we have described is one in which each device learns its own input and partial information about the input of the other device. Eve, on the other hand, learns no additional information about any inputs or outputs once the protocol commences. In particular, the noisy channel that leaks information between the devices does not leak anything to Eve. This is a consequence of the ``closed lab assumption'', which states that no information can enter or leave the lab where the devices operate once the protocol commences\footnote{Eve does however know which input Alice will use to generate randomness}. Without this, secret information could be arbitrarily leaked to Eve and security would never be possible. Since we are considering protocols for DI randomness generation, we assume that both devices are held inside the same secure lab\footnote{This would not be in case, however, for DI quantum key distribution protocols.}. We keep the closed lab assumption along with other standard assumptions in our analysis (see~\cite[Section XIV.B.1]{Pirandola_2020} for a summary). The assumption we are relaxing is the non-signaling assumption between the devices themselves; this could, for example, be a consequence of performing a Bell test in a tight space, though still within a secure laboratory.

\subsection{Choice of Bell inequalities} 

To compare the amount of randomness certified by the violation of different Bell inequalities in this scenario, we consider the statistics generated from a two-qubit non-signaling quantum strategy subject to depolarizing noise. Specifically, we consider a strategy of the form 
\begin{equation}
    \begin{gathered}
        \rho = (1-\mu) \ketbra{\psi_{\theta}}{\psi_{\theta}} + \mu \frac{\id_{4}}{4}, \\ 
        A_{x} = \cos(a_{x})\, \sigma_{Z} + \sin(a_{x}) \, \sigma_{X}, \\
        B_{y} = \cos(b_{y}) \, \sigma_{Z} + \sin(b_{y}) \, \sigma_{X},
    \end{gathered} \label{eq:noisy_strat}
\end{equation}
where $\mu \in [0,1]$ is the depolarizing noise parameter, and $(\theta,a_{0},a_{1},b_{0},b_{1})$ correspond to the angles that achieve the maximum quantum violation of the given Bell inequality, represented by the function $f$ in \cref{eq:sig_ent}. For example, when testing the CHSH inequality $f=f_{\text{CHSH}}$, we choose the values $(\pi/4,0,\pi/2,\pi/4,-\pi/4)$, corresponding to maximally anti-commuting measurements on a maximally entangled state. Then, for a fixed $\mu$, we evaluate the functional $f$ on the behavior  induced by the the noisy optimal strategy in \cref{eq:noisy_strat}. This gives a Bell value $\omega_{f}(\mu)$ (depending on the choice of Bell functional $f$ and noise $\mu$), for which we compute the randomness $r_{\mathrm{S},p}(f,\omega_{f}(\mu))$ for a fixed crossover probability $1-p$. This corresponds to the setting where every state that could be prepared by an experimentalist is subject to the same depolarizing noise $\mu$, and their setup is susceptible to noisy input communication from Alice to Bob with crossover probability $1-p$. Both $\mu$ and $p$ are assumed to be known to the experimentalist, and they are free to optimize over the choice of Bell functional $f$ (which determines the chosen angles $(\theta,a_{0},a_{1},b_{0},b_{1})$) that maximizes the rate $r_{\mathrm{S},p}(f,\omega_{f}(\mu))$. 

For our example, we compare three different choices of Bell functional $f$. The first is the CHSH functional, $f=f_{\text{CHSH}}$ defined in \cref{eq:chsh}, which corresponds to the only non-trivial facet (up to relabelings) of the non-signaling local polytope. The second is the tilted functional defined in \cref{eq:tilt_chsh}, where the tilting parameter is set equal to the crossover probability $1-p$, $f=T_{p}$. This corresponds to a facet of the LHV polytope with symmetric one-way signaling of Alice's input to Bob, and in particular $T_{1/2} = f_{\text{CHSH}}$. Finally, we consider the optimal choice of tilted functional, $f = T_{q}$, where we maximize over the value $q \in (2/5,1)$. We denote the optimal choice by $T_{q_{\text{opt}}}$. Since this family includes both $f_{\text{CHSH}}$ and $T_{p}$ as special cases, 
\begin{multline*}
     r_{\mathrm{S},p}(T_{q_{\text{opt}}},\omega_{T_{q_{\text{opt}}}}(\mu)) := \sup_{q \in (2/5,1)}r_{\mathrm{S},p}(T_{q},\omega_{T_{q}}(\mu)) \\\geq \max\{r_{\mathrm{S},p}(T_{p},\omega_{T_{p}}(\mu)),r_{\mathrm{S},p}(f_{\text{CHSH}},\omega_{f_{\text{CHSH}}}(\mu))\}    
\end{multline*}   
for all choices of $p$ and $\mu$. Note that in all cases, we only consider Bell values achievable by non-signaling quantum behaviors, i.e., those induced by \eqref{eq:noisy_strat}. 

We also deliberately certify the local randomness from Alice's first measurement setting ($X=0$) in \cref{eq:sig_ent} to align with the quantum strategy \eqref{eq:qubit_strat}. Here, Alice's device measures the observable $\sigma_{X}$ on a partially entangled state $\cos(\theta)\ket{00} + \sin(\theta)\ket{11}$, resulting in a uniformly distributed outcome, $\p(a|x=0) = 1/2$ $\forall a\in\{0,1\}$. This implementation therefore generates 1 bit of randomness, which is the maximum that can be achieved from the output of one party in the minimal Bell setting. If, using the same strategy, Alice's second setting ($X=1$) was used, a measurement in the computational basis would result in a biased outcome depending on $\theta$, and thus less entropy. This would also happen if randomness was generated from one of Bob's measurements. We therefore consider Alice's first measurement setting, while noting that a larger entropy in the noiseless setting does not necessarily imply a larger conditional entropy in the presence of an eavesdropper. It can however be a good indication in certain regimes.      

\subsection{Results}

Having defined the constraints, a lower bound on \cref{eq:sig_ent} can be computed using the numerical technique of Ref.~\cite{BrownDeviceIndependent2}. The results are shown in \cref{fig:rand_1,fig:rand_2}. In \cref{fig:rand_1}, we see that when $p=1/2$ (blue lines), $T_{1/2} = f_{\text{CHSH}}$, and optimizing over the family $T_{q}$ results in a slight improvement at intermediate noise values. For $p = 0.55$ (red lines), we find that the CHSH inequality outperforms $T_{p}$ for intermediate values of $\mu$ until the high noise regime, where both $T_{p}$ and $T_{q_{\text{opt}}}$ provide greater robustness. As $p$ increases further, both $T_{p}$ and $T_{q_{\text{opt}}}$ outperform the CHSH inequality, exhibiting higher rates at intermediate noise values and remaining non-zero at larger noise values. We find that for all values of $p$, the optimal choice $q_{\text{opt}}$ tends to $p$ as the noise becomes high. This is reflected in the fact that both the solid and dotted curves go to zero at the same noise value in \cref{fig:rand_1}. This suggests that the facet inequalities $T_{p}$ are the most robust in the high noise noise regime, replicating similar observations made when comparing tilted Bell inequalities to the CHSH inequality in the standard non-signaling scenario~\cite{WBC}. This is to be expected, since facet inequalities become the only inequalities that can be violated with high enough noise. We plot the optimal parameter $q_{\text{opt}}$ in \cref{app:plots} to illustrate this behavior .    

In \cref{fig:rand_2}, we focus on the case where there is no depolarizing noise ($\mu = 0$), i.e., the maximum quantum non-signaling value of each Bell functional is achieved in the presence of noisy one-way signaling. Here we also see the facet inequalities $T_{p}$ outperforming the CHSH inequality; the latter exhibits a positive rate when the probability of Bob receiving Alice's input is less than roughly $p\approx 0.7$, while the former certifies a non-zero rate for $p$ close to 1. Furthermore, the optimal choice $T_{q_{\text{opt}}}$ certifies a near maximum amount of randomness for almost all values of $p$ that are less than 1.  

\begin{figure}
    \centering
    \includegraphics[width=0.5\textwidth]{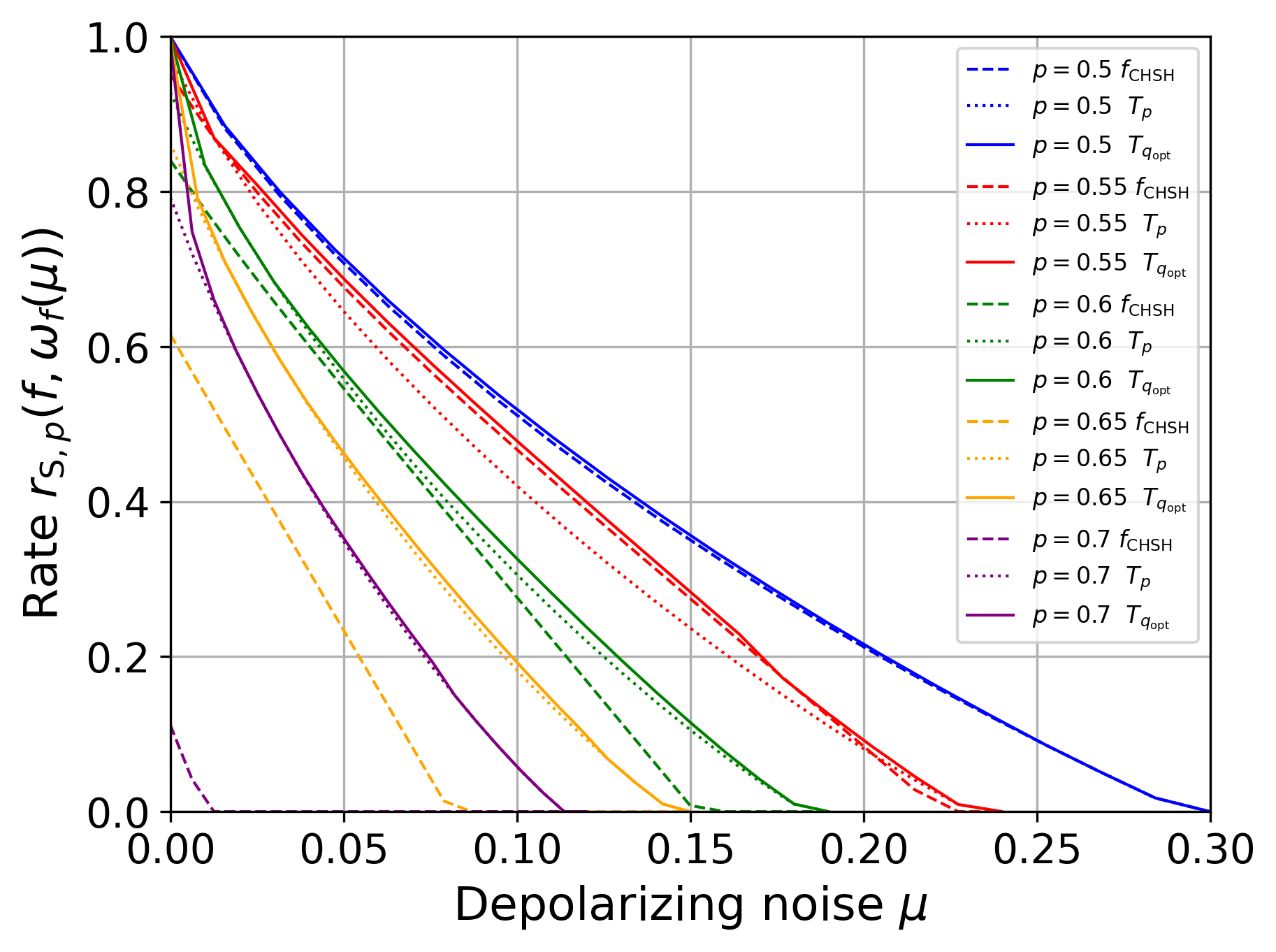}
    \caption{Comparison between different Bell expression when certifying DI randomness in the presence of one-way signaling and depolarizing noise. We consider a binary symmetric channel where the probability that Bob correctly receives a copy of Alice's input is $p \in [1/2,1)$. $f_{\mathrm{CHSH}}$ denotes the CHSH functional, $T_{p}$ denotes the new facet derived in this work \eqref{eq:tilt_chsh}, and $T_{q_{\text{opt}}}$ denotes the optimal choice from the family of expressions $T_{q}$ for $q\in (2/5,1)$.}
    \label{fig:rand_1}
\end{figure}

\begin{figure}
    \centering
    \includegraphics[width=0.5\textwidth]{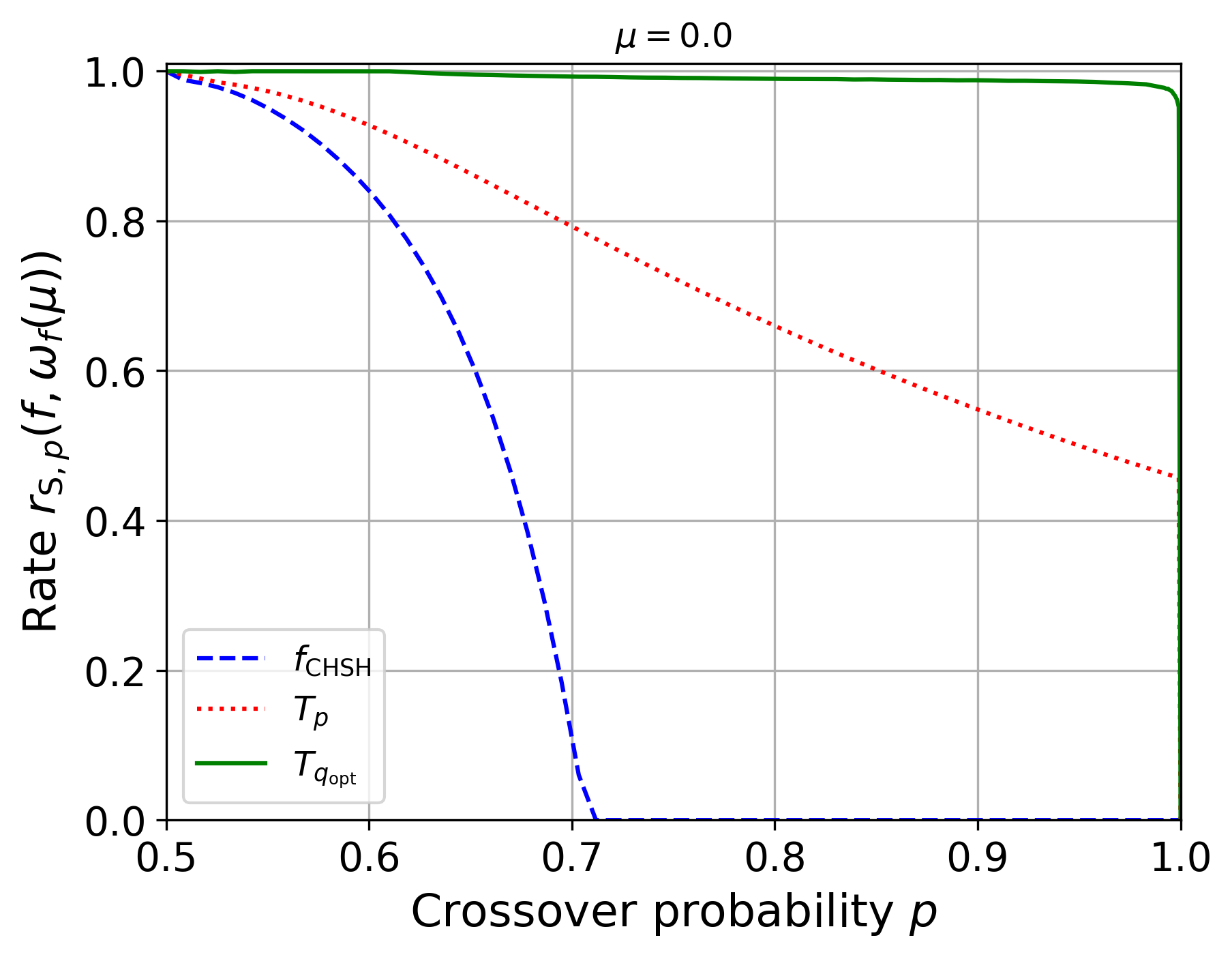}
    \caption{Comparison between different Bell expression when certifying DI randomness in the presence of one-way signaling only. The expressions $f_{\text{CHSH}},\, T_{p}$ and $T_{q_{\text{opt}}}$ are as described in \cref{fig:rand_1}. For each Bell expression, the observed value is set equal to the maximum non-signaling quantum value, i.e., the value is generated from the strategy \eqref{eq:noisy_strat} with $\mu = 0$ and the optimal measurement angles for the given Bell functional.}
    \label{fig:rand_2}
\end{figure}

\section{Discussion} \label{sec:disc}
Our results show that nonlocal quantum correlations are robust to noisy input signaling in the minimal Bell scenario. By considering a binary channel through which the signaling takes place, we were able to completely characterize the new local polytope for both one-way and two-way signaling scenarios. We then explored some properties of these polytopes, and identified new Bell inequalities that generalize the CHSH inequality to signaling scenarios. We showed that in the case of noisy one-way and two-way signaling, these inequalities can always be violated by non-signaling quantum behaviors, and in the one-way case, they certify more DI randomness than the CHSH inequality in the presence of depolarizing noise. We now discuss possible directions for future work.   

We presented a complete list of facets for the local polytope in three signaling scenarios: symmetric one-way, symmetric two-way, and asymmetric two-way with one symbol perfectly transmitted. For the first and third scenario, we identified a facet that can always be violated by a non-signaling quantum strategy, providing an explicit construction that, in the one-way case, is a self-test. It would be interesting to obtain similar results for other inequalities reported in \cref{app:all_facets}. In particular, analytically characterizing the optimal quantum strategy for the Bell expression $W_{p,r}$ studied in \cref{lem:2_ways} is left open, as is the case for the more general scenario in Conjecture \ref{conj:sig}. We made preliminary progress in this direction by numerically optimizing over qubit strategies, however finding a clean analytical solution, if one exists, is a more challenging task. 

In the one-way signaling scenario, we found improvements over the CHSH inequality when certifying DI randomness. In particular, \cref{fig:rand_2} shows that non-signaling quantum behaviors can obtain a rate bounded away from 0 for almost perfect one-way signaling, i.e., $p$ very close to 1 (cf. dotted red line in \cref{fig:rand_2}). It would be interesting to obtain the analytical form of this curve and understand if there is a discontinuity at $p=1$ (where no randomness can be certified), or if the rate rapidly drops to zero as $p$ approaches 1. Based on our numerics, we conjecture the former to be the case. Proving this would require analytically solving the rate function \eqref{eq:sig_ent}; in the non-signaling scenario, tools such as Jordan's lemma~\cite{Jordan,ABGMPS,BhavsarDI} and self-testing~\cite{WBC} have been very useful in achieving this goal, and whether similar steps can be taken in the signaling scenario is an interesting question. Furthermore, the optimal choice of Bell expression from the family \eqref{eq:tilt_chsh} certifies almost 1 bit of randomness even when $p$ approaches 1 (cf. solid green line in \cref{fig:rand_2}). The maximum amount of local randomness that can be certified in this Bell scenario is 1 bit, and analytically understanding how close one can get to this maximum in signaling scenarios is an appealing direction.   

We studied the randomness of one party's output that can be certified in a symmetric one-way signaling scenario. The numerical tools of Ref.~\cite{BrownDeviceIndependent2} used to achieve this can be extended to two-way signaling scenarios, and to certifying the randomness from both party's outcomes. Accounting for two-way signaling could be more relevant, since the assumption that one party's input leaks to the other but not the other way around could be difficult to justify in practice. Towards this end, we identified interesting facets in the symmetric two-way signaling scenario, and their application to randomness certification is yet to be investigated. We also leave the study of other noise models, such as inefficient detectors, to future work. 

Finally, it would be interesting to explore similar effects in Bell scenarios with more inputs and outputs. While there could be multiple relevant signaling channels, the resulting LHV model will still form a polytope and a similar analysis to the one here could be performed.   

\vspace{0.2cm}

\noindent \textit{Note added:} During the writing up of this work we became aware of a related work~\cite{zapatero2026} that also showed to possibility of certifying DI randomness in the presence of input signaling.

\section*{Acknowledgments}
We are grateful to Roger Colbeck for fruitful discussions. KS was supported by the Departmental Studentship from the Department of Mathematics, University of York, l’Agence Nationale de la Recherche (ANR) project ANR-22-CE47-001 and l’Agence Nationale de la Recherche (Plan France 2030) grant number ANR-22-PETQ-0009 . LW was supported by the European
Union Horizon Europe research and innovation program
under the project “Quantum Secure Networks Partnership” (QSNP, Grant Agreement No. 101114043), and the Engineering and Physical Sciences Research Council (EPSRC, Grant No. EP/SO23607/1). 

%

\onecolumngrid

\appendix

\section{Vertices of the local polytope with general signaling models}
\label{Appendix::FullSigVerts}

We consider here the most general setting with binary asymmetric channels between Alice and Bob. There two channels in total, the first from Alice to Bob and the second from Bob to Alice. Alice's setting $X=0$ is sent perfectly with a probability $p$ and her setting $X=1$ is sent perfectly with probability $q$. These parameters are $r$ and $s$, respectively, for Bob. We present the vertices up to local relabeling operations, calculated using the procedure described in \cref{sec:verts}. The numbers in the subscript denote the number of elements in each class. Note that the local deterministic non-signaling distributions remain extremal.

$$
\left(
\begin{array}{cc|cc}
 1 & 0 & 1 & 0 \\
 0 & 0 & 0 & 0 \\ \hline
 1 & 0 & 1 & 0 \\
 0 & 0 & 0 & 0 \\
\end{array}
\right)_{16}, \left(
\begin{array}{cc|cc}
 0 & 0 & 0 & 0 \\
 0 & 1 & 1-p & p \\ \hline
 0 & 0 & 0 & 0 \\
 0 & 1 & q & 1-q \\
\end{array}
\right)_{32},
\left(
\begin{array}{cc|cc}
 0 & 0 & 0 & 0 \\
 1-p & p & 1-p & p \\ \hline
 0 & 0 & 0 & 0 \\
 q & 1-q & q & 1-q \\
\end{array}
\right)_{16}, \left(
\begin{array}{cc|cc}
 0 & 0 & 0 & 0 \\
 0 & 1 & 0 & 1 \\ \hline
 0 & 1-r & 0 & s \\
 0 & r & 0 & 1-s \\
\end{array}
\right)_{32}
$$
$$
\left(
\begin{array}{cc|cc}
 0 & 0 & 0 & 0 \\
 0 & 1 & 1-p & p \\ \hline
 0 & 1-r & q s & s-q s \\
 0 & r & q-q s & (q-1) (s-1) \\
\end{array}
\right)_{16},
\left(
\begin{array}{cc|cc}
 0 & 0 & 0 & 0 \\
 1-p & p & 0 & 1 \\ \hline
 q r & r-q r & 0 & 1-s \\
 q-q r & (q-1) (r-1) & 0 & s \\
\end{array}
\right)_{16},$$
$$
\left(
\begin{array}{cc|cc}
 0 & 0 & 0 & 0 \\
 1-p & p & 1-p & p \\ \hline
 q r & r-q r & q-q s & (q-1) (s-1) \\
 q-q r & (q-1) (r-1) & q s & s-q s \\
\end{array}
\right)_{16},
\left(
\begin{array}{cc|cc}
 0 & 1-r & p s & s-p s \\
 0 & r & p-p s & (p-1) (s-1) \\ \hline
 0 & 0 & 0 & 0 \\
 0 & 1 & 1-q & q \\
\end{array}
\right)_{16},$$
$$
\left(
\begin{array}{cc|cc}
 (p-1) (r-1) & p-p r & 0 & s \\
 r-p r & p r & 0 & 1-s \\ \hline
 0 & 0 & 0 & 0 \\
 q & 1-q & 0 & 1 \\
\end{array}
\right)_{16},\left(
\begin{array}{cc|cc}
 (p-1) (r-1) & p-p r & p s & s-p s \\
 r-p r & p r & p-p s & (p-1) (s-1) \\ \hline
 0 & 0 & 0 & 0 \\
 q & 1-q & 1-q & q \\
\end{array}
\right)_{16},
$$
$$
\left(
\begin{array}{cc|cc}
 0 & 1-r & 0 & s \\
 0 & r & 0 & 1-s \\ \hline
 0 & 1-r & 0 & s \\
 0 & r & 0 & 1-s \\
\end{array}
\right)_{16},\left(
\begin{array}{cc|cc}
 0 & 1-r & p s & s-p s \\
 0 & r & p-p s & (p-1) (s-1) \\ \hline
 0 & 1-r & s-q s & q s \\
 0 & r & (q-1) (s-1) & q-q s \\
\end{array}
\right)_{16},\left(
\begin{array}{cc|cc}
 (p-1) (r-1) & p-p r & 0 & s \\
 r-p r & p r & 0 & 1-s \\ \hline
 q r & r-q r & 0 & 1-s \\
 q-q r & (q-1) (r-1) & 0 & s \\
\end{array}
\right)_{16},
$$
$$
\left(
\begin{array}{cc|cc}
 (p-1) (r-1) & p-p r & p s & s-p s \\
 r-p r & p r & p-p s & (p-1) (s-1) \\ \hline
 q r & r-q r & (q-1) (s-1) & q-q s \\
 q-q r & (q-1) (r-1) & s-q s & q s \\
\end{array}
\right)_{16};
$$

\section{Facets of signaling polytopes from noisy channel models} \label{app:all_facets}

\subsection{Facets of \texorpdfstring{$\mathbb{S}_{(p,p),(1/2,1/2)}$}{}}
\label{Appendix::one_way}

$$\mathrm{F}_1 = \left(
\begin{array}{cc|cc}
 1 & 0 & 0 & 0 \\
 0 & 0 & 0 & 0 \\ \hline
 0 & 0 & 0 & 0 \\
 0 & 0 & 0 & 0 \\
\end{array}
\right), \quad \mathrm{F}_2 = \left(
\begin{array}{cc|cc}
 1 & -\frac{1-p}{2 p-1} & 0 & 0 \\
 1 & -\frac{1-p}{2 p-1} & 0 & 0 \\ \hline
 0 & -\frac{p-1}{2 p-1} & 0 & 0 \\
 0 & -\frac{p-1}{2 p-1} & 0 & 0 \\
\end{array}
\right),$$
$$
\mathrm{F}_3 =\left(
\begin{array}{cc|cc}
 1 & -\frac{-3 p^3+4 p^2-3 p+1}{4 p^3-4 p^2+3 p-1} & 0 & -\frac{p^3}{4 p^3-4 p^2+3 p-1} \\
 -\frac{-5 p^3+5 p^2-3 p+1}{4 p^3-4 p^2+3 p-1} & -\frac{-4 p^3+3 p^2-2 p+1}{4 p^3-4 p^2+3 p-1} & \frac{(p-1)^2 p}{4 p^3-4 p^2+3 p-1} & 0 \\ \hline
 0 & -\frac{(p-1)^3}{4 p^3-4 p^2+3 p-1} & 0 & -\frac{p^3-p^2}{4 p^3-4 p^2+3 p-1} \\
 0 & -\frac{p^3-p^2}{4 p^3-4 p^2+3 p-1} & 0 & -\frac{(p-1)^3}{4 p^3-4 p^2+3 p-1} \\
\end{array}
\right),
$$
$$
\mathrm{F}_4 =\left(
\begin{array}{cc|cc}
 -\frac{-3 p^2+4 p-2}{2 p^2-3 p+2} & \frac{(p-1) \left(5 p^2-6 p+2\right)}{(2 p-1) \left(2 p^2-3 p+2\right)} & 0 & -\frac{(p-1)^2 p}{(2 p-1)
   \left(2 p^2-3 p+2\right)} \\
 -\frac{-3 p^2+4 p-2}{2 p^2-3 p+2} & -\frac{-5 p^3+9 p^2-7 p+2}{(2 p-1) \left(2 p^2-3 p+2\right)} & 0 & -\frac{p^3}{(2 p-1) \left(2 p^2-3
   p+2\right)} \\ \hline
 0 & -\frac{(p-1) p^2}{(2 p-1) \left(2 p^2-3 p+2\right)} & 0 & -\frac{(p-1) p^2}{(2 p-1) \left(2 p^2-3 p+2\right)} \\
 0 & -\frac{(p-1)^3}{(2 p-1) \left(2 p^2-3 p+2\right)} & -\frac{p^2-p}{2 p^2-3 p+2} & -\frac{(p-1) \left(3 p^2-3 p+1\right)}{(2 p-1) \left(2
   p^2-3 p+2\right)} \\
\end{array}
\right),
$$
$$\mathrm{F}_5 = \left(
\begin{array}{cc|cc}
 1 & -\frac{-3 p^3+4 p^2-3 p+1}{4 p^3-4 p^2+3 p-1} & 0 & -\frac{(p-1)^3}{4 p^3-4 p^2+3 p-1} \\
 -\frac{-3 p^2+2 p-1}{2 p^2-p+1} & -\frac{-5 p^3+5 p^2-3 p+1}{4 p^3-4 p^2+3 p-1} & 0 & -\frac{p^3-p^2}{4 p^3-4 p^2+3 p-1} \\ \hline
 0 & -\frac{(p-1)^3}{4 p^3-4 p^2+3 p-1} & 0 & -\frac{(p-1)^2 p}{4 p^3-4 p^2+3 p-1} \\
 0 & -\frac{p^3-p^2}{4 p^3-4 p^2+3 p-1} & 0 & -\frac{p^3}{4 p^3-4 p^2+3 p-1} \\
\end{array}
\right),$$
$$\mathrm{F}_6 = \left(
\begin{array}{cc|cc}
 1 & -\frac{-3 p^3+4 p^2-3 p+1}{4 p^3-4 p^2+3 p-1} & 0 & -\frac{(p-1)^3}{4 p^3-4 p^2+3 p-1} \\
 -\frac{-3 p^2+2 p-1}{2 p^2-p+1} & -\frac{-5 p^3+5 p^2-3 p+1}{4 p^3-4 p^2+3 p-1} & 0 & -\frac{p^3-p^2}{4 p^3-4 p^2+3 p-1} \\ \hline
 0 & -\frac{(p-1)^3}{4 p^3-4 p^2+3 p-1} & 0 & -\frac{(p-1)^2 p}{4 p^3-4 p^2+3 p-1} \\
 0 & -\frac{p^3-p^2}{4 p^3-4 p^2+3 p-1} & 0 & -\frac{p^3}{4 p^3-4 p^2+3 p-1} \\
\end{array}
\right),$$
$$\mathrm{F}_7 = \left(
\begin{array}{cc|cc}
 -\frac{-4 p^2+3 p-1}{3 p^2-2 p+1} & -\frac{-7 p^3+10 p^2-5 p+1}{(2 p-1) \left(3 p^2-2 p+1\right)} & 0 & -\frac{p^2-p}{3 p^2-2 p+1} \\
 -\frac{-4 p^2+3 p-1}{3 p^2-2 p+1} & -\frac{-7 p^3+8 p^2-4 p+1}{(2 p-1) \left(3 p^2-2 p+1\right)} & 0 & 0 \\ \hline
 0 & -\frac{(p-1) p^2}{(2 p-1) \left(3 p^2-2 p+1\right)} & 0 & 0 \\
 0 & -\frac{(p-1)^3}{(2 p-1) \left(3 p^2-2 p+1\right)} & -\frac{p^2-p}{3 p^2-2 p+1} & -\frac{p (2 p-1)}{3 p^2-2 p+1} \\
\end{array}
\right),$$
$$
\mathrm{F}_8=\left(
\begin{array}{cc|cc}
 1 & -\frac{-5 p^3+7 p^2-4 p+1}{(2 p-1) \left(3 p^2-2 p+1\right)} & 0 & -\frac{p^2}{3 p^2-2 p+1} \\
 -\frac{-4 p^2+3 p-1}{3 p^2-2 p+1} & \frac{(p-1) p^2}{(2 p-1) \left(3 p^2-2 p+1\right)}+1 & 0 & 0 \\ \hline
 0 & -\frac{(p-1)^3}{(2 p-1) \left(3 p^2-2 p+1\right)} & 0 & -\frac{(p-1) p}{3 p^2-2 p+1} \\
 0 & -\frac{(p-1) p^2}{(2 p-1) \left(3 p^2-2 p+1\right)} & 0 & 0 \\
\end{array}
\right),
$$
$$\mathrm{F}_{9}=\left(
\begin{array}{cc|cc}
 1 & -\frac{1-p}{2 p-1} & 0 & 0 \\
 \frac{(p-1)^2}{p (2 p-1)} & 0 & -\frac{2-3 p}{2 p-1} & -\frac{2-3 p}{2 p-1} \\ \hline
 0 & -\frac{p-1}{2 p-1} & 0 & 0 \\
 0 & -\frac{p-1}{2 p-1} & 0 & 0 \\
\end{array}
\right), \quad \mathrm{F}_{10}=\left(
\begin{array}{cc|cc}
 1 & \frac{(p-1)^3}{2 p^3-3 p^2+3 p-1} & 0 & 0 \\
 1 & -\frac{-p^3+p^2-2 p+1}{2 p^3-3 p^2+3 p-1} & -\frac{p^2}{p^2-p+1} & 0 \\ \hline
 0 & -\frac{(p-1)^3}{2 p^3-3 p^2+3 p-1} & 0 & 0 \\
 0 & -\frac{(p-1) p^2}{2 p^3-3 p^2+3 p-1} & 0 & \frac{(p-1) p}{p^2-p+1} \\
\end{array}
\right),$$
$$\mathrm{F}_{11}=\left(
\begin{array}{cc|cc}
 1 & -\frac{1-2 p}{3 p-1} & 0 & -\frac{p}{3 p-1} \\
 \frac{2 (2 p-1)}{3 p-1} & \frac{2 (2 p-1)}{3 p-1} & 0 & 0 \\ \hline
 0 & -\frac{p-1}{3 p-1} & 0 & 0 \\
 0 & 0 & 0 & -\frac{p-1}{3 p-1} \\
\end{array}
\right), \quad \mathrm{F}_{12}=\left(
\begin{array}{cc|cc}
 1 & -\frac{1-2 p}{3 p-1} & 0 & 0 \\
 1 & 1 & -\frac{1-p}{3 p-1} & 0 \\ \hline
 0 & 0 & 0 & 0 \\
 0 & -\frac{p-1}{3 p-1} & 0 & -\frac{p}{3 p-1} \\
\end{array}
\right)$$
$$\mathrm{F}_{13}= \left(
\begin{array}{cc|cc}
 0 & 0 & 1 & 0 \\
 0 & 1 & 0 & 0 \\ \hline
 0 & -\frac{1-p}{p} & 0 & -\frac{p-1}{p} \\
 0 & 0 & 0 & 0 \\
\end{array}
\right)$$

\subsection{Facets of \texorpdfstring{$\mathbb{S}_{(p,p),(p,p)} $}{}}
\label{Appendix::G_One_Way}

$$
\mathrm{F}_1'=\left(
\begin{array}{cc|cc}
 1 & 0 & 0 & 0 \\
 0 & 0 & 0 & 0 \\ \hline
 0 & 0 & 0 & 0 \\
 0 & 0 & 0 & 0 \\
\end{array}
\right),
$$
$$
\mathrm{F}_2'=\left(
\begin{array}{cc|cc}
 \frac{p}{2 p-1} & \frac{p}{2 p-1} & \frac{p-1}{2 p-1} & \frac{p-1}{2 p-1} \\
 0 & 0 & 0 & 0 \\ \hline
 0 & 0 & 0 & 0 \\
 0 & 0 & 0 & 0 \\ 
\end{array}
\right),
$$
$$
\mathrm{F}_3'=\left(
\begin{array}{cc|cc}
 1 & -\frac{-3 p^3+2 p^2-2 p+1}{4 p^3-4 p^2+3 p-1} & 0 & \frac{(p-1) p^2}{4 p^3-4 p^2+3 p-1} \\
 -\frac{-3 p^3+2 p^2-2 p+1}{4 p^3-4 p^2+3 p-1} & -\frac{-2 p^3+2 p^2-2 p+1}{4 p^3-4 p^2+3 p-1} & -\frac{(p-1)^3}{4 p^3-4 p^2+3 p-1} & 0 \\ \hline
 0 & -\frac{(p-1)^3}{4 p^3-4 p^2+3 p-1} & 0 & -\frac{-p^3+2 p^2-p}{4 p^3-4 p^2+3 p-1} \\
 0 & -\frac{(p-1) p^2}{4 p^3-4 p^2+3 p-1} & -\frac{p}{2 p^2-p+1} & -\frac{-p^3+2 p^2-p}{4 p^3-4 p^2+3 p-1} \\
\end{array}
\right),
$$
$$
\mathrm{F}_4'=\left(
\begin{array}{cc|cc}
 1 & -\frac{-p^2-p+1}{2 p-1} & 0 & 0 \\
 1 & -\frac{-p^2-p+1}{2 p-1} & 0 & 0 \\ \hline
 0 & \frac{(p-1)^2}{2 p-1} & 0 & 0 \\
 -\frac{p^2}{2 p-1} & 0 & -\frac{p^2-p}{2 p-1} & -\frac{p^2-p}{2 p-1} \\
\end{array}
\right),
$$
$$
\mathrm{F}_5'=\left(
\begin{array}{cc|cc}
 1 & 1 & 0 & 0 \\
 -\frac{1-2 p}{3 p-1} & 1 & 0 & -\frac{p-1}{3 p-1} \\ \hline
 0 & -\frac{1-p}{3 p-1} & 0 & 0 \\
 0 & 0 & 0 & -\frac{p}{3 p-1} \\
\end{array}
\right),
$$

\subsection{Facets of \texorpdfstring{$\mathbb{S}_{(p,1),(r,1)}$}{}}
\label{Appendix::one_symbol_perfect}
 In the subscript, $s$ refers to a facet for which numerical evidence suggests no quantum violation is possible using two-qubit strategies. On the other hand, $v$ refers to a facet for which there is numerical evidence that a non-signaling two-qubit strategy violates the corresponding facet inequality.
$$
 \left(
\begin{array}{cc|cc}
 1 & 0 & 0 & 0 \\
 0 & 0 & 0 & 0 \\ \hline
 0 & 0 & 0 & 0 \\
 0 & 0 & 0 & 0 \\
\end{array}
\right)_{16,s},
 \left(
\begin{array}{cc|cc}
 1 & 1 & 0 & 0 \\
 1-\frac{1}{r} & 1-\frac{1}{r} & \frac{1}{r}-1 & \frac{1}{r}-1 \\ \hline
 0 & 0 & 0 & 0 \\
 0 & 0 & 0 & 0 \\
\end{array}
\right)_{4,s},
\left(
\begin{array}{cc|cc}
 1 & -\frac{1-p}{p} & 0 & 0 \\
 1 & -\frac{1-p}{p} & 0 & 0 \\ \hline
 0 & -\frac{p-1}{p} & 0 & 0 \\
 0 & -\frac{p-1}{p} & 0 & 0 \\
\end{array}
\right)_{4,s},
\left(
\begin{array}{cc|cc}
 1 & 1 & 0 & -\frac{r-1}{p r-r-1} \\
 1 & 1-\frac{1}{-p r+r+1} & 0 & 0 \\ \hline
 0 & -\frac{1}{p r-r-1}-1 & 0 & 0 \\
 0 & -\frac{2-p}{p r-r-1}-1 & -\frac{p (-r)+p+r-1}{p r-r-1} & 0 \\
\end{array}
\right)_{16,v}
 $$
$$
\left(
\begin{array}{cc|cc}
 1 & 1 & 0 & r-1 \\
 1 & 0 & 0 & 0 \\ \hline
 0 & 0 & 0 & 0 \\
 p-1 & 0 & 0 & p r-p-r+1 \\
\end{array}
\right)_{16,v*},\left(
\begin{array}{cc|cc}
 1 & 1 & 0 & -\frac{1}{p} \\
 1 & 1 & 1-\frac{1}{p} & 1-\frac{1}{p} \\ \hline
 0 & 0 & 0 & -\frac{p-1}{p} \\
 0 & 0 & 0 & -\frac{p-1}{p} \\
\end{array}
\right)_{8,s},\left(
\begin{array}{cc|cc}
 1 & 1 & 0 & 0 \\
 1 & -\frac{-2 p r+p+r}{2 p r-p-r-1} & -\frac{p r-p}{2 p r-p-r-1} & -\frac{p r-p-r+1}{2 p r-p-r-1} \\ \hline
 0 & -\frac{p r-r}{2 p r-p-r-1} & 0 & -\frac{p (-r)+p+r-1}{2 p r-p-r-1} \\
 0 & -\frac{p r-p-r+1}{2 p r-p-r-1} & -\frac{p (-r)+p+r-1}{2 p r-p-r-1} & -\frac{p (-r)+p+r-1}{2 p r-p-r-1} \\
\end{array}
\right)_{16,v}
$$
$$
\left(
\begin{array}{cc|cc}
 p r-r+1 & p r-r+1 & 0 & r-1 \\
 p r-p-r+2 & p r-p-r+1 & p (-r)+p+r-1 & p (-r)+p+r-1 \\ \hline
 0 & r-p r & 0 & 0 \\
 p-1 & r-p r & 0 & p r-p-r+1 \\
\end{array}
\right)_{16,v},$$
$$\left(
\begin{array}{cc|cc}
 1 & 1 & 0 & 0 \\
 1 & 1-\frac{1}{p (-r)+p+1} & -\frac{1}{p r-p-1}-1 & -\frac{2-r}{p r-p-1}-1 \\ \hline
 0 & 0 & 0 & -\frac{p (-r)+p+r-1}{p r-p-1} \\
 -\frac{p-1}{p r-p-1} & 0 & 0 & 0 \\
\end{array}
\right)_{16,v}
$$
$$
\left(
\begin{array}{cc|cc}
 -\frac{-2 p r+2 p+2 r-1}{p r-p-r} & -\frac{-2 p r+2 p+2 r-1}{p r-p-r} & 0 & 0 \\
 -\frac{-2 p r+2 p+2 r-1}{p r-p-r} & \frac{2 (p r-p-r+1)}{p r-p-r} & -\frac{p r-p-r+1}{p r-p-r} & -\frac{p r-p-r+1}{p r-p-r} \\ \hline
 0 & -\frac{p r-p-r+1}{p r-p-r} & 0 & 0 \\
 0 & -\frac{p r-p-r+1}{p r-p-r} & 0 & 0 \\
\end{array}
\right)_{4,s}, \left(
\begin{array}{cc|cc}
 1 & 1 & 0 & 0 \\
 1 & 1 & 0 & 0 \\ \hline
 0 & 1-\frac{1}{r} & 0 & 0 \\
 -\frac{1}{r} & 1-\frac{1}{r} & -\frac{r-1}{r} & -\frac{r-1}{r} \\
\end{array}
\right)_{8,s}
$$
$$
\left(
\begin{array}{cc|cc}
 p r-p+1 & p r-p-r+2 & 0 & r-1 \\
 p r-p+1 & p r-p-r+1 & p-p r & p-p r \\ \hline
 0 & p (-r)+p+r-1 & 0 & 0 \\
 p-1 & p (-r)+p+r-1 & 0 & p r-p-r+1 \\
\end{array}
\right)_{16,v}
$$
$$
\left(
\begin{array}{cc|cc}
 -\frac{-3 p^2 r+2 p^2+2 p r-2 p+1}{2 p^2 r-p^2-p r+p-1} & -\frac{-3 p^2 r+2 p^2+3 p r-3 p-r+2}{2 p^2 r-p^2-p r+p-1} & 0 & -\frac{r-1}{2 p^2 r-p^2-p r+p-1} \\
 -\frac{-3 p^2 r^2+3 p^2 r-p^2+2 p r^2-3 p r+p+r}{r \left(2 p^2 r-p^2-p r+p-1\right)} & -\frac{-3 p^2 r^2+3 p^2 r-p^2+3 p r^2-3 p r+p-r^2+r}{r \left(2 p^2 r-p^2-p
   r+p-1\right)} & -\frac{\left(p^2-p\right) (r-1)^2}{r \left(2 p^2 r-p^2-p r+p-1\right)} & -\frac{\left(p^2-p\right) (r-1)^2}{r \left(2 p^2 r-p^2-p r+p-1\right)} \\ \hline
 0 & -\frac{(p-1)^2 (r-1)}{2 p^2 r-p^2-p r+p-1} & 0 & 0 \\
 -\frac{p-1}{2 p^2 r-p^2-p r+p-1} & -\frac{(p-1)^2 (r-1)}{2 p^2 r-p^2-p r+p-1} & 0 & -\frac{p r-p-r+1}{2 p^2 r-p^2-p r+p-1} \\
\end{array}
\right)_{16,v}
$$
$$
\left(
\begin{array}{cc|cc}
 p r-r+1 & p r-r+1 & 0 & r-1 \\
 p r-p-r+2 & p r-p-r+1 & p (-r)+p+r-1 & p (-r)+p+r-1 \\ \hline
 0 & 0 & 0 & 0 \\
 0 & 1-p & p (-r)+p+r-1 & 0 \\
\end{array}
\right)_{16,v},\left(
\begin{array}{cc|cc}
 -\frac{-p+r+1}{p r-r-1} & -\frac{2-p}{p r-r-1} & 0 & -\frac{r-1}{p r-r-1} \\
 -\frac{-p+r+1}{p r-r-1} & -\frac{1-p}{p r-r-1} & -\frac{p-p r}{p r-r-1} & -\frac{p-p r}{p r-r-1} \\ \hline
 0 & -\frac{p-1}{p r-r-1} & 0 & 0 \\
 0 & 0 & \frac{(p-1) (r-1)}{p r-r-1} & 0 \\
\end{array}
\right)_{16,v},$$
$$\left(
\begin{array}{cc|cc}
 -\frac{p-r+1}{p r-p-1} & -\frac{p-r+1}{p r-p-1} & 0 & 0 \\
 -\frac{2-r}{p r-p-1} & -\frac{1-r}{p r-p-1} & -\frac{r-1}{p r-p-1} & 0 \\ \hline
 0 & -\frac{r-p r}{p r-p-1} & 0 & \frac{(p-1) (r-1)}{p r-p-1} \\
 -\frac{p-1}{p r-p-1} & -\frac{r-p r}{p r-p-1} & 0 & 0 \\
\end{array}
\right)_{16,v}
$$
$$
\left(
\begin{array}{cc|cc}
 -\frac{-2 p r+2 r-1}{p r-r+1} & -\frac{-2 p r+2 r-1}{p r-r+1} & 0 & 0 \\
 -\frac{-2 p r+p+2 r-2}{p r-r+1} & -\frac{-2 p r+p+2 r-1}{p r-r+1} & -\frac{p r-p-r+1}{p r-r+1} & -\frac{p r-p-r+1}{p r-r+1} \\ \hline
 0 & -\frac{p r-r}{p r-r+1} & 0 & 0 \\
 0 & -\frac{p r-r}{p r-r+1} & 0 & 0 \\
\end{array}
\right)_{4,s},
$$
$$
\left(
\begin{array}{cc|cc}
 -\frac{-3 p^2 r+p^2+3 p r-p-r+1}{2 p^2 r-p^2-p r+p-1} & -\frac{-3 p^2 r+p^2+4 p r-2 p-2 r+2}{2 p^2 r-p^2-p r+p-1} & 0 & -\frac{r-1}{2 p^2 r-p^2-p r+p-1} \\
 -\frac{-3 p^2 r+2 p^2+3 p r-3 p-r+2}{2 p^2 r-p^2-p r+p-1} & -\frac{-3 p^2 r+2 p^2+4 p r-3 p-2 r+2}{2 p^2 r-p^2-p r+p-1} & -\frac{p^2 r-p^2-2 p r+2 p+r-1}{2 p^2 r-p^2-p
   r+p-1} & -\frac{p^2 r-p^2-2 p r+2 p+r-1}{2 p^2 r-p^2-p r+p-1} \\ \hline
 0 & -\frac{(p-1) (p r-2 r+1)}{2 p^2 r-p^2-p r+p-1} & 0 & 0 \\
 -\frac{p-1}{2 p^2 r-p^2-p r+p-1} & -\frac{(p-1) (p r-2 r+1)}{2 p^2 r-p^2-p r+p-1} & 0 & -\frac{(p-1) (r-1)}{2 p^2 r-p^2-p r+p-1} \\
\end{array}
\right)_{16,v}
$$
$$
\left(
\begin{array}{cc|cc}
 -\frac{1-p r}{p r-r-1} & 1 & 0 & -\frac{r-1}{p r-r-1} \\ 
 -\frac{2-p r}{p r-r-1} & 1 & -\frac{r-1}{p r-r-1} & -\frac{r-1}{p r-r-1} \\ \hline
 0 & 0 & 0 & 0 \\
 -\frac{p-1}{p r-r-1} & 0 & 0 & -\frac{(p-1) (r-1)}{p r-r-1} \\
\end{array}
\right)_{16,v}
$$
$$
\left(
\begin{array}{cc|cc}
 -\frac{-2 p^2 r+2 p^2-2 p+r+1}{p^2 r-p^2+p r+p-r-1} & -\frac{-2 p^2 r+2 p^2+p r-3 p+2}{p^2 r-p^2+p r+p-r-1} & 0 & -\frac{r-1}{p^2 r-p^2+p r+p-r-1} \\
 -\frac{-2 p^2 r^2+3 p^2 r-p^2-3 p r+p+r^2+r}{r \left(p^2 r-p^2+p r+p-r-1\right)} & -\frac{-2 p^2 r^2+3 p^2 r-p^2+p r^2-3 p r+p+r}{r \left(p^2 r-p^2+p r+p-r-1\right)} &
   -\frac{p^2 r^2-2 p^2 r+p^2-p r^2+2 p r-p}{r \left(p^2 r-p^2+p r+p-r-1\right)} & -\frac{p^2 r^2-2 p^2 r+p^2-p r^2+2 p r-p}{r \left(p^2 r-p^2+p r+p-r-1\right)} \\ \hline
 0 & -\frac{(p-1) (p r-p+1)}{p^2 r-p^2+p r+p-r-1} & 0 & 0 \\
 0 & -\frac{\left(p^2-p\right) (r-1)}{p^2 r-p^2+p r+p-r-1} & \frac{(p-1) (r-1)}{p^2 r-p^2+p r+p-r-1} & 0 \\
\end{array}
\right)_{16,v}
$$
$$
\left(
\begin{array}{cc|cc}
 -\frac{-2 p^2 r+p^2+p r-p+1}{p^2 r-1} & -\frac{-2 p^2 r+p^2+2 p r-2 p-r+2}{p^2 r-1} & 0 & 0 \\
 -\frac{-2 p^2 r^2+2 p^2 r-p^2+p r^2-2 p r+p+r}{r \left(p^2 r-1\right)} & -\frac{-2 p^2 r^2+2 p^2 r-p^2+2 p r^2-2 p r+p-r^2+r}{r \left(p^2 r-1\right)} & -\frac{(r-1)
   \left(p^2 r-p^2+p\right)}{r \left(p^2 r-1\right)} & -\frac{(p-1) (r-1) (p r-p+r)}{r \left(p^2 r-1\right)} \\ \hline
 0 & -\frac{p^2 r-p^2-2 p r+2 p+r-1}{p^2 r-1} & 0 & \frac{(p-1) (r-1)}{p^2 r-1} \\
 -\frac{p-1}{p^2 r-1} & -\frac{p^2 r-p^2-2 p r+2 p+r-1}{p^2 r-1} & 0 & 0 \\
\end{array}
\right)_{16,v}
$$
$$
\left(
\begin{array}{cc|cc}
 -\frac{1-p r}{p^2 r-1} & -\frac{-2 p r+r+1}{p^2 r-1} & 0 & -\frac{r-1}{p^2 r-1} \\
 -\frac{-p^2-p r+p+1}{p^2 r-1} & -\frac{-p^2-2 p r+2 p+r}{p^2 r-1} & \frac{(p-1) (p r-p)}{p^2 r-1} & \frac{(p-1) (p r-p)}{p^2 r-1} \\ \hline
 0 & \frac{-2 p r+r+1}{p^2 r-1}+1 & 0 & 0 \\
 -\frac{p-1}{p^2 r-1} & \frac{-2 p r+r+1}{p^2 r-1}+1 & 0 & -\frac{(p-1) (r-1)}{p^2 r-1} \\
\end{array}
\right)_{16,v},\left(
\begin{array}{cc|cc}
 -\frac{p-r+1}{p r-p-1} & -\frac{p-r+1}{p r-p-1} & 0 & 0 \\
 -\frac{2-r}{p r-p-1} & -\frac{1-r}{p r-p-1} & -\frac{r-1}{p r-p-1} & 0 \\ \hline
 0 & 0 & 0 & \frac{(p-1) (r-1)}{p r-p-1} \\
 0 & -\frac{1-p}{p r-p-1} & \frac{(p-1) (r-1)}{p r-p-1} & \frac{(p-1) (r-1)}{p r-p-1} \\
\end{array}
\right)_{16,v},
$$
$$
\left(
\begin{array}{cc|cc}
 -\frac{-2 p^2 r+2 p r-r+1}{p^2 r-1} & -\frac{-2 p^2 r+3 p r-p-2 r+2}{p^2 r-1} & 0 & -\frac{r-1}{p^2 r-1} \\
 -\frac{-2 p^2 r+p^2+2 p r-2 p-r+2}{p^2 r-1} & -\frac{-2 p^2 r+p^2+3 p r-2 p-2 r+2}{p^2 r-1} & -\frac{p^2 r-p^2-2 p r+2 p+r-1}{p^2 r-1} & -\frac{p^2 r-p^2-2 p r+2 p+r-1}{p^2
   r-1} \\ \hline
 0 & -\frac{(p-1) (p r-r+1)}{p^2 r-1} & 0 & 0 \\
 0 & -\frac{\left(p^2-2 p+1\right) r}{p^2 r-1} & \frac{(p-1) (r-1)}{p^2 r-1} & 0 \\
\end{array}
\right)_{16,v}
$$
$$
\left(
\begin{array}{cc|cc}
 1-r & 1 & 0 & r-1 \\
 2-r & 1 & r-1 & r-1 \\ \hline
 0 & p r-r & 0 & 0 \\
 0 & p r-p-r+1 & p (-r)+p+r-1 & 0 \\
\end{array}
\right)_{16,v}, \left(
\begin{array}{cc|cc}
 -\frac{-2 p r+2 p-1}{p r-p+1} & -\frac{-2 p r+2 p+r-2}{p r-p+1} & 0 & 0 \\
 -\frac{-2 p r+2 p-1}{p r-p+1} & -\frac{-2 p r+2 p+r-1}{p r-p+1} & -\frac{p r-p}{p r-p+1} & -\frac{p r-p}{p r-p+1} \\ \hline
 0 & -\frac{p r-p-r+1}{p r-p+1} & 0 & 0 \\
 0 & -\frac{p r-p-r+1}{p r-p+1} & 0 & 0 \\
\end{array}
\right)_{4,s}
$$
$$
\left(
\begin{array}{cc|cc}
 -\frac{p (-r)+p+r-1}{2 p r-p-r+1} & -\frac{p (-r)+p-1}{2 p r-p-r+1} & 0 & 0 \\
 -\frac{-p r+r-1}{2 p r-p-r+1} & \frac{p r}{2 p r-p-r+1} & \frac{p (r-1)}{2 p r-p-r+1} & \frac{p (r-1)}{2 p r-p-r+1} \\ \hline
 0 & \frac{(p-1) r}{2 p r-p-r+1} & 0 & 0 \\
 0 & \frac{(p-1) r}{2 p r-p-r+1} & 0 & 0 \\
\end{array}
\right)_{4,s}
$$
$$
\left(
\begin{array}{cc|cc}
 -\frac{-2 p^2 r+2 p^2-2 p+r+1}{p^2 r-p^2+p r+p-r-1} & -\frac{-2 p^2 r+2 p^2+p r-3 p+2}{p^2 r-p^2+p r+p-r-1} & 0 & 0 \\
 -\frac{-2 p^2 r^2+3 p^2 r-p^2-3 p r+p+r^2+r}{r \left(p^2 r-p^2+p r+p-r-1\right)} & -\frac{-2 p^2 r^2+3 p^2 r-p^2+p r^2-3 p r+p+r}{r \left(p^2 r-p^2+p r+p-r-1\right)} &
   -\frac{(p r-p) (p r-p+1)}{r \left(p^2 r-p^2+p r+p-r-1\right)} & -\frac{(p r-p-r+1) (p r-p+r)}{r \left(p^2 r-p^2+p r+p-r-1\right)} \\ \hline
 0 & -\frac{(p-1) (p r-p+1)}{p^2 r-p^2+p r+p-r-1} & 0 & -\frac{p (-r)+p+r-1}{p^2 r-p^2+p r+p-r-1} \\
 0 & -\frac{p^2 r-p^2-p r+p}{p^2 r-p^2+p r+p-r-1} & -\frac{p (-r)+p+r-1}{p^2 r-p^2+p r+p-r-1} & -\frac{p (-r)+p+r-1}{p^2 r-p^2+p r+p-r-1} \\
\end{array}
\right)_{16,v}
$$
$$
\left(
\begin{array}{cc|cc}
 -\frac{-2 p^2 r+2 p r-r+1}{p^2 r-1} & -\frac{-2 p^2 r+3 p r-p-2 r+2}{p^2 r-1} & 0 & 0 \\
 -\frac{-2 p^2 r+p^2+2 p r-2 p-r+2}{p^2 r-1} & -\frac{-2 p^2 r+p^2+3 p r-2 p-2 r+2}{p^2 r-1} & -\frac{p^2 r-p^2-p r+p+r-1}{p^2 r-1} & -\frac{(p-1) (p r-p)}{p^2 r-1} \\ \hline
 0 & -\frac{p^2 r-3 p r+p+2 r-1}{p^2 r-1} & 0 & \frac{(p-1) (r-1)}{p^2 r-1} \\
 -\frac{p-1}{p^2 r-1} & -\frac{p^2 r-3 p r+p+2 r-1}{p^2 r-1} & 0 & 0 \\
\end{array}
\right)_{16,v}
$$

\subsection{Examples of interesting facets of \texorpdfstring{$\mathbb{S}_{(p,1)^*,(r=p,1)^*}$}{}}
\label{Appendix::random_symbol_perfect}

Note that for any $p,r\in[0,1)$, $\mathbb{S}_{(p,1)^*,(r,1)^*} \subseteq \mathbb{S}_{(p=p',1)^*,(r=p',1)^*}$, where $p'\coloneqq \max\{p,r\}$. Therefore, for simplicity here we take $r=p$. We numerically found 5 interesting facets of this polytope, i.e., ones that indicate a quantum violation following a numerical search. We managed to derive the following two of them analytically.

$$G1=\left(
\begin{array}{cc|cc}
 1 & \frac{p (p (p ((p-3) p-1)+10)-10)+2}{p (p (p ((p-3) p-1)+9)-8)} & 0 & 0 \\
 1 & \frac{p (p (p ((p-3) p-1)+10)-10)+1}{p (p (p ((p-3) p-1)+9)-8)} & \frac{1}{p (p (p ((p-3) p-1)+9)-8)} & 0 \\ \hline
 0 & \frac{(p-1) ((p-3) p+1)}{p (p (p ((p-3) p-1)+9)-8)} & 0 & 0 \\
 \frac{(p-2) (p-1)}{p (p ((p-3) p-1)+9)-8} & \frac{(p-2) (p-1)^2}{p (p (p ((p-3) p-1)+9)-8)} & 0 & \frac{1-p}{p (p (p ((p-3) p-1)+9)-8)} \\
\end{array}
\right)$$

$$
G2=\left(
\begin{array}{cc|cc}
 \frac{(p-2) (p-1)}{p (p ((p-3) p-1)+9)-8}+1 & \frac{p \left(p^4-3 p^3+6 p-6\right)+1}{p (p (p ((p-3) p-1)+9)-8)} & 0 & \frac{1-p}{p (p (p ((p-3) p-1)+9)-8)} \\
 \frac{p \left(p^4-3 p^3+5 p-4\right)-1}{p (p (p ((p-3) p-1)+9)-8)} & \frac{p \left(p^4-3 p^3+5 p-4\right)-1}{p (p (p ((p-3) p-1)+9)-8)} & \frac{(p-2) (p-1)}{p (p (p ((p-3)
   p-1)+9)-8)} & -\frac{(p-2) (p-1)^2}{p (p (p ((p-3) p-1)+9)-8)} \\ \hline
 0 & 0 & 0 & \frac{p-1}{p (p (p ((p-3) p-1)+9)-8)} \\
 \frac{1}{p (p (p ((p-3) p-1)+9)-8)} & 0 & \frac{p-1}{p (p (p ((p-3) p-1)+9)-8)} & \frac{p-1}{p (p (p ((p-3) p-1)+9)-8)} \\
\end{array}
\right)
$$

\section{Proof of \cref{Theorem:modelconnection}} \label{Proof::Theorem_Model_connection}

\noindent \textbf{Lemma 1.} \textit{ Let $p,r\in[0,1)$, and $\mathbb{S}_{(p,1)^*,(r,1)^*}$ be the set of LHV behaviors obtained under the asymmetric signaling channel, as defined in \cref{eq:poly_star}. Let $\mathbf{S}_{(p,r)}$ denote the set of LHV behaviors obtained under a general relaxation of parameter independence, as defined below \cref{Def::GenPD}. Then
    \begin{equation*}
        \mathbb{S}_{(p,1)^*,(r,1)^*} = \mathbf{S}_{(p,r)}. 
    \end{equation*}}

\begin{proof}
We consider the case $p=r$ for simplicity, and note that the same proof holds for the general case. Both sets $\mathbf{S}_{p,p}$ and $\mathbb{S}_{(p,1)^*,(p,1)^*}$ are convex, so it suffices to show that every extremal point of $\mathbf{S}_{p,p}$ is contained in $\mathbb{S}_{(p,1)^*,(p,1)^*}$. Any extremal point of $\mathbf{S}_{p,p}$ is of the form 
\begin{equation*}
    \p(a,b|x,y) = \p_{A}(a|x,y)\p_{B}(b|x,y), 
\end{equation*}
where the marginals satisfy
\begin{equation} 
    | \p_{A}(a|x,0) - \p_{A}(a|x,1) | = p \ \ \ \text{and} \ \ \ | \p_{B}(b|0,y) - \p_{B}(b|1,y) | = p. \label{eq:e3}
\end{equation}
An extremal point of $\mathbb{S}_{(p,1),(p,1)}$ is of the form
\begin{equation*}
    \p(a,b|x,y) = \p_{A}(a|x,y)\p_{B}(b|x,y)
\end{equation*}
where the marginals can be expressed as the output of an asymmetric signaling channel,
\begin{equation}
\begin{aligned}
    \p_{A}(a|x,0) &=  p \, \tilde{\p}_{A}(a|x,0) + (1-p) \, \tilde{\p}_{A}(a|x,1), \\
    \p_{A}(a|x,1) &=  \tilde{\p}_{A}(a|x,1), \\
    \p_{B}(b|0,y) &=  p  \, \tilde{\p}_{B}(b|0,y) + (1-p) \, \tilde{\p}_{B}(b|1,y) \ \ \
\text{and} \\
    \p_{B}(b|1,y) &=  \tilde{\p}_{B}(b|1,y).
\end{aligned} \label{eq:p-resp}
\end{equation}
Note that a direct calculation shows that these marginals satisfy \cref{eq:e3}, i.e., $\mathbb{S}_{(p,1),(p,1)} \subseteq \mathbf{S}_{p,p}$. A similar calculation shows that the same holds for $\mathbb{S}_{(1,p),(p,1)}, \, \mathbb{S}_{(p,1),(1,p)}$ and $\mathbb{S}_{(1,p),(1,p)}$, where the symbols sent perfectly are changed.

To prove that every point in $\mathbf{S}_{p,p}$, is contained in the convex hull of $\mathbb{S}_{(p,1),(1,p)}, \mathbb{S}_{(1,p),(p,1)}, \, \mathbb{S}_{(p,1),(1,p)}$ and $\mathbb{S}_{(1,p),(1,p)}$, consider any extremal point of $\mathbf{S}_{p,p}$ that satisfies \cref{eq:e3}. We know that 
\begin{equation*}
    (\p_{A}(a|x,0) - \p_{A}(a|x,1))(-1)^t = p , \ \ \ \text{and} \ \ \ (\p_{B}(b|0,y) - \p_{B}(b|1,y))(-1)^u = p
\end{equation*}
for some $t,u \in \{0,1\}$. 

\vspace{0.2cm}

\noindent \underline{Case 1:} Suppose $t = u = 0$. Define the following response functions by inverting \cref{eq:p-resp} (this is always possible since $p > 0$):
\begin{equation}
\begin{aligned}
    \tilde{\p}_{A}(a|x,0) &=  \frac{1}{p}  \p_{A}(a|x,0) - \frac{1-p}{p} \p_{A}(a|x,1), \\
    \tilde{\p}_{A}(a|x,1) &=  \p_{A}(a|x,1), \\
    \tilde{\p}_{B}(b|0,y) &=  \frac{1}{p} \tilde{\p}_{B}(b|0,y) - \frac{1-p}{p} \tilde{\p}_{B}(b|1,y) \ \ \
\text{and} \\
    \tilde{\p}_{B}(b|1,y) &=  \p_{B}(b|1,y).
\end{aligned} \label{eq:resp_case1}
\end{equation}
We now verify that these are valid response functions. Firstly, they are normalized since $\p_{A}$ and $\p_{B}$ are normalized. Both $\tilde{\p}_{A}(a|x,1)$ and $\tilde{\p}_{B}(b|1,y)$ are positive because $\p_{A}$ and $\p_{B}$ are positive. For $\tilde{\p}_{A}(a|x,0)$, note that because $\p$ belongs to $\mathbf{S}_{p,p}$, its marginals satisfy $\p_{A}(a|x,0) = \p_{A}(a|x,1) + p$ and $\p_{B}(b|0,y) = \p_{B}(b|1,y) + p$, and therefore 
\begin{equation*}
    \p_{A}(a|x,0) = \p_{A}(a|x,1) + p = \p_{A}(a|x,1)(1-p) + p (1+\p_{A}(a|x,1)) \geq \p_{A}(a|x,1)(1-p),
\end{equation*}
which implies $\tilde{\p}_{A}(a|x,0) \geq 0$. We can similarly show $\tilde{\p}_{B}(b|0,y) \geq 0$. We therefore have that the behavior belongs to $\mathbb{S}_{(p,1),(1,p)}$, i.e., is of the form \cref{eq:p-resp} where the response functions are given by \cref{eq:resp_case1}. 

\vspace{0.2cm}

\noindent \underline{Case 2:} Suppose $t = u = 1$. Define the following response functions:
\begin{equation}
\begin{aligned}
    \tilde{\p}_{A}(a|x,0) &=  \p_{A}(a|x,0), \\
    \tilde{\p}_{A}(a|x,1) &=  \frac{1}{p}  \p_{A}(a|x,1) - \frac{1-p}{p} \p_{A}(a|x,0), \\
    \tilde{\p}_{B}(b|0,y) &=  \p_{B}(b|0,y) \ \ \
\text{and} \\
    \tilde{\p}_{B}(b|1,y) &=  \frac{1}{p}  \p_{B}(b|1,y) - \frac{1-p}{p} \p_{B}(b|0,y).
\end{aligned} \label{eq:resp_case2}
\end{equation}  
We again have normalization and positivity of the first and third line. For the second line, 
\begin{equation*}
    \p_{A}(a|x,1) = \p_{A}(a|x,0) + p = (1-p)\p_{A}(a|x,0) + p(1 + \p_{A}(a|x,0)) \geq (1-p)\p_{A}(a|x,0),
\end{equation*}
implying $\tilde{\p}_{A}(a|x,1) \geq 0$. A similar argument establishes that $\tilde{\p}_{B}(b|1,y) \geq 0$. Inverting \cref{eq:resp_case2}, we can therefore write $\p$ as the output of a channel that sends $X=0$ and $Y = 0$ perfectly, and $X=1$ and $Y=1$ with fidelity $p$, and thus it is a member of $\mathbb{S}_{(1,p),(1,p)}$.

\vspace{0.2cm}

\noindent \underline{Cases 3 and 4:} We can proceed in the same way as before. When $t = 0$ and $u = 1$, Alice's response functions are those in \cref{eq:resp_case1}, and Bob's are those in \cref{eq:resp_case2}, and we find the behavior belongs to $\mathbb{S}_{(p,1),(1,p)}$. When $t = 1$ and $u = 0$, we get the opposite. Alice's response functions are those in \cref{eq:resp_case2}, and Bob's are those in \cref{eq:resp_case1}, and we find the behavior belongs to $\mathbb{S}_{(1,p),(p,1)}$. This completes the proof.  

\end{proof}

\section{Proof of \cref{Lemma::MD_PD_and_PD,Lemma::jointbound}} \label{Appendix::MD_PD_and_PD}

\noindent \textbf{Lemma 2.} \textit{Let $A,B,X, Y$ and $\Lambda$ be random variables where $A,B,X$ and $Y$ are binary. Then for any behavior $\p(A,B|X,Y)$ admitting
    \begin{equation*}
        \p(a,b|x,y) = \sum_{\lambda} \p_{\Lambda}(\lambda|x,y)\p_{A}(a|x,y,\lambda)\p_{B}(b|x,y,\lambda),
    \end{equation*}
    where $\p_{\Lambda},\,\p_{A}$ and $\p_{B}$ are distributions, there exists a random variable $\Lambda'$ and tuple of distributions $(\p_{\Lambda'}',\p_{A}',\p_{B}')$ such that 
    \begin{equation*}
        \p(a,b|x,y) = \sum_{\lambda'} \p'_{\Lambda'}(\lambda')\p'_{A}(a|x,y,\lambda')\p'_{B}(b|x,y,\lambda').
    \end{equation*}}
\begin{proof}
    Since the starting model admits measurement dependence, for any fixed $x$ and $y$,
\begin{equation*}
    \p_{\Lambda} (\lambda|x,y) = \frac{\p(\lambda)\p(x,y|\lambda)}{\p(x,y)} = 4 \p(\lambda) p (x,y|\lambda),
\end{equation*}
where the last equality follows from the fact that $\p(x,y)=\sum_\lambda \p(x,y|\lambda)\p(\lambda) = 1/4$. Note, that $\sum_\lambda \p_{\Lambda}(\lambda|x,y) = 1$ for every choice of inputs $x$ and $y$. Let us define four random variables $\{\Lambda_{x,y}\}_{x,y}$ that take the same values as $\Lambda$, such that $\p(\Lambda_{x,y}) = \p_{\Lambda}(\Lambda|x,y)$ for every $x$ and $y$. With this we define $\Lambda'$ as a random variable such that
\begin{equation*}
    \p'_{\Lambda'}(\Lambda')  = \p(\Lambda_{00})\p(\Lambda_{01})\p(\Lambda_{10})\p(\Lambda_{11}).
\end{equation*}
In other words, $\Lambda'$ is the tuple $(\Lambda_{00},\Lambda_{01},\Lambda_{10},\Lambda_{11})$. With this one gets that $\p'_{\Lambda'}(\Lambda'|x,y) = \p'_{\Lambda'}(\Lambda')$ for any choice of $x$ and $y$. Now since, Alice and Bob have access to both $X$ and $Y$, when $\Lambda'$ is provided to both of them, they can choose their shared randomness to be $\Lambda_{x,y}$. In particular, we set

\begin{equation*}
    \p'_{A}(a|x,y,\lambda') = \p_{A}(a|x,y,\lambda_{x,y})  \quad \text{and} \quad  \p'_{B}(b|x,y,\lambda') = \p_{B}(b|x,y,\lambda_{x,y}).
\end{equation*}
Therefore, 
\begin{equation*}
    \begin{split}
        &\sum_{\lambda'}\p'_{\Lambda'}(\lambda')\p'_{A}(a|x,y,\lambda')\p'_{B}(b|x,y,\lambda')\\
        &=\sum_{(\lambda_{\tilde x, \tilde y})_{\tilde x, \tilde y}} \left(\prod_{\tilde x,\tilde y} \p(\lambda_{\tilde x, \tilde y})  \right) \p_{A}(a|x,y,\lambda_{x,y}) \p_{B}(b|x,y,\lambda_{x,y})\\
        &= \sum_{\lambda_{x,y}} \p(\lambda_{x,y}) \p_{A}(a|x,y,\lambda_{x,y}) \p_{B}(b|x,y,\lambda_{x,y}) \left[ \prod_{\tilde x\neq x, \tilde y \neq y}  \left( \sum_{\tilde x, \tilde y} \p(\lambda_{\tilde x, \tilde y}) \right)\right] \\
        &=\sum_{\lambda_{x,y}} \p(\lambda_{x,y}) \p_{A}(a|x,y,\lambda_{x,y}) \p_{B}(b|x,y,\lambda_{x,y}) \\
        &=\sum_{\lambda} \p_{\Lambda}(\lambda | x,y) \p_{A}(a|x,y,\lambda) \p_{B}(b|x,y,\lambda) \\
        &= \p(a,b|x,y).
    \end{split} 
\end{equation*}
Here the third equality follows from the fact that each sum in the product is 1 and the fourth equality follows from noticing that since all terms depend on $x$ and $y$, one can rename $\lambda_{x,y}$ as $\lambda$.    
\end{proof}

We now turn our attention to \cref{Lemma::jointbound}. The proof follows from two supporting lemmas that we state first.

\begin{lemma}
\label{Lemma::DiffBound}
    Let $\Lambda, X$ and  $Y$ be three random variables, where $X$ and $Y$ are binary. Consider a distribution $q(X,Y,\Lambda)$ over these such that $q(x,y|\lambda) \geq l$ and $0 \leq l \leq 1/4$ for any $\lambda,x$ and $y$, and $q(x,y)=1/4$. Then for any function $g : \Lambda \to [0,1]$ ,
    \begin{align*}
        &\left| \sum_{\lambda} q(\lambda|x,y)g(\lambda)  -  \sum_{\lambda} q(\lambda|x,y')g(\lambda) \right| \leq \frac{1-4l}{1-2l} \ \forall\  x, \\
        &\left| \sum_{\lambda} q(\lambda|x,y)g(\lambda)  -  \sum_{\lambda} q(\lambda|x',y)g(\lambda) \right| \leq \frac{1-4l}{1-2l} \ \forall\  y. 
    \end{align*}    
\end{lemma}
\begin{proof}
    We prove the first inequality, and the second follows similarly. Since, $q(x,y|\lambda) \geq l$, we get using normalization $q(x,y|\lambda) \leq 1-3l$. Let $\mu = q(\lambda | x,y) / q(\lambda|x,y') $. Then using Bayes' rule and the bounds on $q(x,y|\lambda)$, we get
    \begin{equation}
   \frac{l}{1-3l} \leq\frac{q(x,y|\lambda)}{q(x,y'|\lambda)} = \mu \leq \frac{1-3l}{l} \label{eq:mu_ineqs}
    \end{equation}
Now, for any function $g : \Lambda \to [0,1]$, let 
\begin{align*}
    \Gamma &\coloneqq \left|  \sum_{\lambda} q(\lambda|x,y)g(\lambda)  -  \sum_{\lambda} q(\lambda|x,y')g(\lambda) \right| \\
    &=\left|  \sum_{\lambda} \mu q(\lambda|x,y')g(\lambda)  -  \sum_{\lambda} q(\lambda|x,y')g(\lambda) \right| \\
    &= \left| \sum_{\lambda}q(\lambda|x,y')(\mu -1) g(\lambda)  \right|
\end{align*}
Note that $\sum_{\lambda} q(\lambda|x,y')(\mu-1) =0$. Since $\mu \in \mathbb{R}$, this implies that 
\begin{equation*}
   \sum_{\lambda| \mu > 1} q(\lambda|x,y')(\mu-1) = \sum_{\lambda| \mu \leq 1} q(\lambda|x,y')(1-\mu) =: \gamma.
\end{equation*}
By splitting the sum we can rewrite $\Gamma$ as
\begin{equation*}
    \Gamma = \left|  \sum_{\lambda| \mu > 1} q(\lambda|x,y')(\mu-1)g(\lambda)  -\sum_{\lambda| \mu \leq 1} q(\lambda|x,y')(1-\mu)g(\lambda) \right|
\end{equation*}
Now, since $g(\lambda) \in [0,1]$, $\Gamma = |a-b|$ for $a,b\in [0,\gamma]$, and therefore $\Gamma \leq \gamma$.

Next we find an upper bound on $\gamma$. First, by \cref{eq:mu_ineqs} $\mu - 1 \leq (1-4l)/l $ and we get that
\begin{equation*}
    \gamma \leq \frac{1-4l}{l}\sum_{\lambda| \mu > 1} q(\lambda|x,y').
\end{equation*}
But \cref{eq:mu_ineqs} also implies that $1-\mu\leq  (1-4l)/(1-3l)$, and we get
\begin{equation*}
    \gamma \leq  \frac{1-4l}{1-3l} \left(\sum_{\lambda| \mu \leq 1} q(\lambda|x,y') \right) = \frac{1-4l}{1-3l} \left(1-\sum_{\lambda| \mu > 1} q(\lambda|x,y') \right).
\end{equation*}
From these two bounds, it follows that 
\begin{equation*}
    \gamma \leq \min \left\{\frac{1-4l}{l}\sum_{\lambda| \mu > 1} q(\lambda|x,y') ,\frac{1-4l}{1-3l} \left(1-\sum_{\lambda| \mu > 1} q(\lambda|x,y') \right)\right\}
\end{equation*}
Note that here the minimum depends on the sum inside the expressions, which may vary for different hidden variable models. However, we need a bound that holds for all hidden variable models. Note that $(1-4l)/l \geq 0$ and $(1-4l)/(1-3l)\geq 0$ for $l\in[0,1/4]$. Letting $ w= \sum_{\lambda| \mu > 1} q(\lambda|x,y') \in [0,1]$, we can find a universal bound by taking the worst case over $w\in [0,1]$, i.e.,
\begin{equation*}
    \gamma \leq \max_{w\in[0,1]} \min\Big\{ \frac{1-4l}{l} w , \frac{1-4l}{1-3l} (1-w)\Big\}.
\end{equation*}
Note that $\frac{1-4l}{l}w$ is linearly increasing in $w$, while $\frac{1-4l}{1-3l}(1-w)$ is linearly decreasing in $w$, hence the maximum occurs when both arguments of the inner minimization are equal, i.e., $\frac{1-4l}{l} w = \frac{1-4l}{1-3l}(1-w)$. Solving this, we find $w = l/(1-2l)$, and the objective function evaluates to $\gamma \leq \frac{1-4l}{1-2l}$. This completes the proof.
\end{proof}

The next lemma shows that any signaling distribution can be expressed as a convex combination of its signaling and non-signaling fractions.

\begin{lemma}
\label{Lemma::LocalPart}
    Suppose $\left|\p(a|x,y,\lambda) - \p(a|x,y',\lambda)\right| \leq \epsilon$ for all $a,x,y,y'$ and $\lambda$ and some $\epsilon \in [0,1]$. Then there exists probability distributions $u(A|X,\Lambda)$, $v(A|X,Y,\Lambda)$ such that
    \begin{equation*}
        \p(a|x,y,\lambda) = \epsilon v(a|x,y,\lambda) + (1-\epsilon) u(a|x,\lambda) 
    \end{equation*}
\end{lemma}

\begin{proof}
    For any fixed $x$ and $\lambda$ let 
\begin{equation}
\label{eq::marginals}
    \alpha_0 := p(a=0|x,y=0,\lambda) \quad \text{ and } \quad \alpha_1 := p(a=0|x,y=1,\lambda).
\end{equation}
With this, Alice's marginals are given by
\begin{equation*}
\begin{split}
    \big(p(a=0|x,y=0,\lambda),\,p(a=1|x,y=0,\lambda)\big) = (\alpha_0,1-\alpha_0) \\
    \big(p(a=0|x,y=1,\lambda), \, p(a=1|x,y=1,\lambda)\big) = (\alpha_1,1-\alpha_1) \\
    \end{split}
\end{equation*}
Since $\left|p(a|x,y,\lambda) - p(a|x,y',\lambda)\right| \leq \epsilon$, for an arbitrary distribution admitting this relation, we get $\nu := |\alpha_0 - \alpha_1|  \leq \epsilon$. We now consider two cases, (i) $\alpha_0 \geq \alpha_1$ and (ii) $\alpha_0 \leq \alpha_1$. 

\vspace{0.2cm}

\underline{Case (i); $\alpha_0 - \alpha_1 = \nu$:}

\vspace{0.2cm}

We define the following maximally signalling distribution: $\zeta(a|x,y,\lambda) = \delta_{a,y}$, i.e., $A$ is a copy of $Y$. We then define the non-signalling distribution $\varphi(a|x,\lambda)$ by
\begin{equation*}
    \big( \varphi(0|x,\lambda),  \, \varphi(1|x,\lambda)\big) = \Big( \frac{\alpha_1}{1-\nu} , \frac{1-\alpha_0}{1-\nu}\Big).
\end{equation*}
Note that $\varphi(0|x,\lambda) + \varphi(1|x,\lambda) = \frac{1-(\alpha_0 - \alpha_1)}{1-\nu} =1$, i.e., $\varphi$ is normalized. Furthermore, $\alpha_{1}/(1-\nu) \geq 0$, and hence $\varphi$ is a valid distribution. Furthermore, it satisfies 
\begin{equation*}
    \begin{aligned}
        \nu \, \zeta(0|x,y=0,\lambda) + (1-\nu)\varphi(0|x,\lambda) &= \nu + \alpha_1 = \alpha_0 =  \p(a=0|x,y=0,\lambda)\\
        \nu \, \zeta(0|x,y=1,\lambda) + (1-\nu)\varphi(0|x,\lambda) &= \alpha_1 = \p(a=0|x,y=0,\lambda),
    \end{aligned}
\end{equation*}
i.e., $\p(a|x,y,\lambda) = \nu \, \zeta(a|x,y,\lambda) + (1-\nu)\varphi(0|x,\lambda)$.

\vspace{0.2cm}

\underline{Case (ii); $\alpha_1 - \alpha_0 =\nu$:}

\vspace{0.2cm}

This case follows similarly as the one above by switching $\alpha_0$ and $\alpha_1$ and interchanging the deterministic signaling responses. 

\vspace{0.4cm}

\noindent Combining the results of these two cases, we have 
\begin{equation*}
    \p(a|x,y,\lambda) = \nu \zeta(a|x,y,\lambda) + (1-\nu) \varphi(a|x,\lambda).
\end{equation*}
Upon rearranging;
\begin{equation*}
    \begin{split}
        &\nu \zeta(a|x,y,\lambda) + (1-\nu) \varphi(a|x,\lambda)\\
        &= \nu \zeta(a|x,y,\lambda) + ((1-\epsilon)+(\epsilon-\nu))\varphi(a|x,\lambda) \\
        &=\nu \zeta(a|x,y,\lambda) + (\epsilon-\nu)\varphi(a|x,\lambda) + (1-\epsilon) \varphi (a|x,\lambda) \\
        &= \epsilon \left[\frac{\nu}{\epsilon} \zeta(a|x,y,\lambda) + \frac{\epsilon - \nu}{\epsilon}\varphi(a|x,\lambda)   \right] + (1-\epsilon) \varphi (a|x,\lambda).
    \end{split}
\end{equation*}
Now we set 
\begin{equation*}
    u(a|x,\lambda) = \varphi (a|x,\lambda) \text{ and } v(a|x,y,\lambda) =  \frac{\nu}{\epsilon} \zeta(a|x,y,\lambda) + \frac{\epsilon - \nu}{\epsilon}\varphi(a|x,\lambda).
\end{equation*}
Note that $v$ is indeed a probability distribution since it is the convex combination of two probability distributions, where the convexity follows from the fact that $\epsilon - \nu \geq 0$. This completes the proof.
\end{proof}

    \noindent \textbf{Lemma 3.} \textit{Let $0\leq l \leq 1/4$ and $0 \leq \epsilon \leq 1$. Let $\mathbf{S}^{l}_{(\epsilon,\epsilon)}$ be the set of behaviors over binary random variables $X,Y,A$ and $B$ with relaxed measurement and parameter dependence parameters $l$ and $\xi^{\A \to \B}=\xi^{\B \to \A}=\epsilon$, respectively, as defined above Eq~\eqref{eq:lam_dep}. Then $\mathbf{S}^{l}_{(\epsilon,\epsilon)} \subseteq \mathbf{S}_{(\kappa,\kappa)}$, with 
\begin{equation*}
    \kappa = \epsilon +(1-\epsilon)\frac{1-4l}{1-2l}.
\end{equation*}}

\begin{proof} We first use the previous two Lemmas to find a bound on the amount of total signalling present in any behavior from $\mathbf{S}_{(\epsilon,\epsilon)}^{l}$: 
 \begin{align*}
    & \left|\p(a|x,y) - \p(a|x,y')\right| = \left|\sum_{\lambda} q(\lambda|x,y)\p(a|x,y,\lambda) - \sum_{\lambda} q(\lambda|x,y') \p(a|x,y',\lambda) \right|  \\
     &\stackrel{\text{Lemma~\ref{Lemma::LocalPart}}}{=}\left|(1-\epsilon) \left( \sum_{\lambda} q(\lambda|x,y) u(a|x,\lambda) - \sum_{\lambda}q(\lambda|x,y')u(a|x,\lambda) \right) + \epsilon\left(  \sum_{\lambda} q(\lambda|x,y) v(a|x,y\lambda) - \sum_{\lambda}q(\lambda|x,y')v(a|x,y'\lambda) \right)  \right| \\
     &\leq (1-\epsilon)\left| \sum_{\lambda} q(\lambda|x,y) u(a|x,\lambda) - \sum_{\lambda}q(\lambda|x,y')u(a|x,\lambda)\right| + \epsilon \left|  \sum_{\lambda} q(\lambda|x,y) v(a|x,y\lambda) - \sum_{\lambda}q(\lambda|x,y')v(a|x,y'\lambda) \right| \\
     &\stackrel{\text{Lemma~\ref{Lemma::DiffBound}}}{\leq} (1-\epsilon) \frac{1-4l}{1-2l} +\epsilon =: \kappa,
 \end{align*}
 where for the first inequality we applied the triangle inequality. A similar calculation follows for Bob's marginals. We have therefore established that any behavior in $\mathbb{S}_{(\epsilon,\epsilon)}^{l}$ is of the form
 \begin{equation}
     \p(a,b|x,y) = \sum_{\lambda} q(\lambda|xy) \p_{A}(a|x,y,\lambda) \p_{B}(b|x,y,\lambda) \label{eq:MI_decomp}
 \end{equation}
 where
 \begin{equation}
     \Big| \sum_{a} \p(a,b|x,y) - \sum_{a} \p(a,b|x',y) \Big| \leq \kappa, \ \text{and} \ \Big| \sum_{b} \p(a,b|x,y) - \sum_{b} \p(a,b|x,y') \Big| \leq \kappa \label{eq:margs}
 \end{equation}
 for all $a,b,x,x',y$ and $y'$. Furthermore, by \cref{Lemma::MD_PD_and_PD}, there exists a hidden variable $\Lambda'$ and distributions $\p_{\Lambda'}'$, $\p_{A}'$ and $\p_{B}'$ such that
 \begin{equation}
     \p(a,b|x,y) = \sum_{\lambda'} \p_{\Lambda'}(\lambda') \p_{A}'(a|x,y,\lambda') \p_{B}'(b|x,y,\lambda') \label{eq:PI_decomp}.
 \end{equation}
 Since \cref{eq:MI_decomp,eq:PI_decomp} are decompositions of the same behavior, their marginals must be equal, implying that the marginals of \cref{eq:PI_decomp} satisfy \cref{eq:margs}. We have therefore found a representation of $\p(a,b|x,y) \in \mathbf{S}_{(\epsilon,\epsilon)}^{l}$ that is measurement independent and has a measurement dependence quantified by $\kappa$ in \cref{eq:margs}. It is therefore a member of $\mathbf{S}_{(\kappa,\kappa)}$, completing the proof.
 \end{proof}

 \section{Proof sketch for Claim \ref{theorem:nsprojection}}
\label{Appendix::NS_Subspace}

\noindent \textbf{Claim 1.} \textit{Let $p,q,r,s\in [1/2,1)$ and $\mathbb{S}_{(p,q),(r,s)}$ be the set of LHV behaviors obtained under a pair of signaling channels as defined in \cref{sec:verts}, and let $$\mathbb{S}_{\mathrm{1-W}}^p \coloneqq \mathrm{ConvHull}\left[\mathbb{S}_{(p,p),(1/2,1/2)} \  \cup \ \mathbb{S}_{(1/2,1/2),(p,p)} \right].$$
    Then $ \Pi_{\mathrm{NS}}\left[\mathbb{S}_{\mathrm{1-W}}^p\right] = \Pi_{\mathrm{NS}}\left[\mathbb{S}_{(p,p),(p,p)}  \right].$ }

\vspace{0.2cm}

\noindent \textit{Proof sketch.} To prove the claim, we notice that for $\mathbb{S}_{\mathrm{1-W}}^p$, the set of local strategies is the convex hull of the union of the extremal local strategies where Alice signals to Bob and the extremal local strategies where Bob signals to Alice. Without loss of generality, we can independently take the one-way signaling model where only Alice signals to Bob or Bob signals to Alice, find the respective non-signaling subspaces and then take the convex hull of the union. 

Below, we provide an example application of the algorithm outlined in \cref{sec:NS_sub} to find the non-signaling subspace when signaling is allowed from Alice to Bob, $\mathbb{S}_{(p,p),(1/2,1/2)}$. The other two cases, $\mathbb{S}_{(1/2,1/2),(p,p)}$ and $\mathbb{S}_{(p,p),(p,p)}$, can be derived analogously to prove the claim. We omit their explicit derivation here for brevity.

In the model $\mathbb{S}_{(p,p),(1/2,1/2)}$, Bob never signals to Alice, hence all behavior lie on the hyperplanes  $\left<\mathrm{NS}_1,\mathbf{x}\right>=0$ and $\left<\mathrm{NS}_2,\mathbf{x}\right>=0$. Therefore, we only need to consider restrictions imposed by $\mathrm{NS}_3$ and $\mathrm{NS}_4$. We run the algorithm from \cref{sec:NS_sub} and calculate the set of intersection points on each of the two hyperplanes. In each case, we classify the behavior up to equivalence of relabeling operations. We first present a representative of the classes that we claim to be extremal and then present the representatives of the remaining classes as convex decompositions of the claimed extremals. 

\vspace{0.2cm}

\noindent \underline{Restriction of $\mathbb{S}_{(p,p),(1/2,1/2)}$ by $\mathrm{NS}_3$:}

\vspace{0.2cm}

We claim that there are three extremal behaviors up to equivalence of relabeling symmetries; 
$$
\left(
\begin{array}{cc|cc}
 0 & 0 & 0 & 0 \\
 0 & 1 & 1-p & p\\ \hline
 0 & 0 & 0 & 0 \\
 0 & 1 & p & 1-p \\
\end{array}
\right),  \frac{1}{2}\left(
\begin{array}{cc|cc}
 1-p & p & 0 & 1 \\
 p & 1-p & 1-p & p \\ \hline
 1-p & p & p & 1-p \\
 p & 1-p & 0 & 1 \\
\end{array}
\right), \frac{1}{2}\left(
\begin{array}{cc|cc}
 1-p & p & 1 & 0 \\
 p & 1-p & 0 & 1 \\ \hline
 1-p & p & 0 & 1 \\
 p & 1-p & 1 & 0 \\
\end{array}
\right)
$$
There are 16 elements in the first class, 32 in the second class and 8 in the third class. Elements of the rest of the classes are not extremal as they can be written as convex combinations of the distributions from the three classes above and the 16 non-signaling local deterministic distributions, as shown below. 
$$
\left(
\begin{array}{cc|cc}
 \frac{1}{2} & \frac{1}{2} & p & 1-p \\
 0 & 0 & 0 & 0 \\ \hline
 \frac{1}{2} & \frac{1}{2} & 1-p &p \\
 0 & 0 & 0 & 0 \\
\end{array}
\right)_{8} = \frac{1}{2}\left(
\begin{array}{cc|cc}
 0 & 1 & p & 1-p \\
 0 & 0 & 0 & 0 \\ \hline
 0 & 1 & 1-p & p \\
 0 & 0 & 0 & 0 \\
\end{array}
\right) + \frac{1}{2} \left(
\begin{array}{cc|cc}
 1 & 0 & p & 1-p \\
 0 & 0 & 0 & 0 \\ \hline
 1 & 0 & 1 & p \\
 0 & 0 & 0 & 0 \\
\end{array}
\right)
$$

$$
\left(
\begin{array}{cc|cc}
 \frac{1}{2} & \frac{1}{2} & 1 & 0 \\
 0 & 0 & 0 & 0 \\ \hline
 \frac{1}{2} & \frac{1}{2} & 1 & 0 \\
 0 & 0 & 0 & 0 \\
\end{array}
\right)_8 = \frac{1}{2} \left(
\begin{array}{cc|cc}
 0 & 1 & 1 & 0 \\
 0 & 0 & 0 & 0 \\ \hline
 0 & 1 & 1 & 0 \\
 0 & 0 & 0 & 0 \\
\end{array}
\right)  +\frac{1}{2} \left(
\begin{array}{cc|cc}
 1 & 0 & 1 & 0 \\
 0 & 0 & 0 & 0 \\ \hline
 1 & 0 & 1 & 0 \\
 0 & 0 & 0 & 0 \\
\end{array}
\right)
$$
$$
\left(
\begin{array}{cc|cc}
 \frac{1}{2} & \frac{1}{2} & \frac{5}{6} & \frac{1}{6} \\
 0 & 0 & 0 & 0 \\ \hline
 \frac{1}{2} & \frac{1}{2} & \frac{2}{3} & \frac{1}{3} \\
 0 & 0 & 0 & 0 \\
\end{array}
\right)_{16} = \frac{1}{2}\left(
\begin{array}{cc|cc}
 0 & 1 & \frac{2}{3} & \frac{1}{3} \\
 0 & 0 & 0 & 0 \\ \hline
 0 & 1 & \frac{1}{3} & \frac{2}{3} \\
 0 & 0 & 0 & 0 \\
\end{array}
\right) +\frac{1}{2}\left(
\begin{array}{cc|cc}
 1 & 0 & 1 & 0 \\
 0 & 0 & 0 & 0 \\ \hline
 1 & 0 & 1 & 0 \\
 0 & 0 & 0 & 0 \\
\end{array}
\right)
 $$
 $$
 \left(
\begin{array}{cc|cc}
 \frac{1}{2} & \frac{1}{2} & \frac{1}{2} & \frac{1}{2} \\
 0 & 0 & 0 & 0 \\ \hline
 \frac{1}{2} & \frac{1}{2} & \frac{1}{2} & \frac{1}{2} \\
 0 & 0 & 0 & 0 \\
\end{array}
\right)_4 = \frac{1}{2} \left(
\begin{array}{cc|cc}
 0 & 1 & 1 & 0 \\
 0 & 0 & 0 & 0 \\ \hline
 0 & 1 & 1 & 0 \\
 0 & 0 & 0 & 0 \\
\end{array}
\right) + \frac{1}{2} \left(
\begin{array}{cc|cc}
 1 & 0 & 0 & 1 \\
 0 & 0 & 0 & 0 \\ \hline
 1 & 0 & 0 & 1 \\
 0 & 0 & 0 & 0 \\
\end{array}
\right)
  $$
  $$
  \left(
\begin{array}{cc|cc}
 \frac{1}{2} & \frac{1}{2} & 1 & 0 \\
 0 & 0 & 0 & 0 \\ \hline
 \frac{1-p}{2} & \frac{p}{2} & \frac{1}{2} & 0 \\
 \frac{p}{2} & \frac{1-p}{2} & \frac{1}{2} & 0 \\
\end{array}
\right)_{16} =\frac{1-p}{2}\left(
\begin{array}{cc|cc}
 0 & 1 & 1 & 0 \\
 0 & 0 & 0 & 0 \\ \hline
 0 & 0 & 0 & 0 \\
 0 & 1 & 1 & 0 \\
\end{array}
\right) + \frac{1-p}{2} \left(
\begin{array}{cc|cc}
 1 & 0 & 1 & 0 \\
 0 & 0 & 0 & 0 \\ \hline
 1 & 0 & 1 & 0 \\
 0 & 0 & 0 & 0 \\
\end{array}
\right) $$

$$+  \frac{p}{2} \left(
\begin{array}{cc|cc}
 0 & 1 & 1 & 0 \\
 0 & 0 & 0 & 0 \\ \hline
 0 & 1 & 1 & 0 \\
 0 & 0 & 0 & 0 \\
\end{array}
\right) +  \frac{p}{2} \left(
\begin{array}{cc|cc}
 1 & 0 & 1 & 0 \\
 0 & 0 & 0 & 0 \\ \hline
 0 & 0 & 0 & 0 \\
 1 & 0 & 1 & 0 \\
\end{array}
\right)
  $$
  $$
  \left(
\begin{array}{cc|cc}
 \frac{1}{2} & \frac{1}{2} & \frac{1+p}{2} & \frac{1-p}{2} \\
 0 & 0 & 0 & 0 \\ \hline
 \frac{1-p}{2} & \frac{p}{2} & \frac{1}{2} & 0 \\
 \frac{p}{2} & \frac{1-p}{2} & \frac{1-p}{2} & \frac{p}{2} \\
\end{array}
\right)_{32} = \frac{1-p}{2}\left(
\begin{array}{cc|cc}
 0 & 1 & p & 1-p \\
 0 & 0 & 0 & 0 \\ \hline
 0 & 0 & 0 & 0 \\
 0 & 1 & 1-p & p \\
\end{array}
\right) +  \frac{1-p}{2}\left(
\begin{array}{cc|cc}
 1 & 0 & 1 & 0 \\
 0 & 0 & 0 & 0 \\ \hline
 1 & 0 & 1 & 0 \\
 0 & 0 & 0 & 0 \\
\end{array}
\right) $$
$$+ \frac{p}{2}\left(
\begin{array}{cc|cc}
 0 & 1 & 1 & 0 \\
 0 & 0 & 0 & 0 \\ \hline
 0 & 1 & 1 & 0 \\
 0 & 0 & 0 & 0 \\
\end{array}
\right) + 
\frac{p}{2}  \left(
\begin{array}{cc|cc}
 1 & 0 & p & 1-p \\
 0 & 0 & 0 & 0 \\ \hline
 0 & 0 & 0 & 0 \\
 1 & 0 & 1-p & p \\
\end{array}
\right)
  $$
  $$
  \left(
\begin{array}{cc|cc}
 \frac{1}{2} & \frac{1}{2} & \frac{2-p}{2} & \frac{p}{2} \\
 0 & 0 & 0 & 0 \\ \hline
 \frac{1-p}{2} & \frac{p}{2} & \frac{1}{2} & 0 \\
 \frac{p}{2} & \frac{1-p}{2} & \frac{p}{2} & \frac{1-p}{2} \\
\end{array}
\right)_{32} = \frac{1-p}{2}\left(
\begin{array}{cc|cc}
 0 & 1 & 1-p & p \\
 0 & 0 & 0 & 0 \\ \hline
 0 & 0 & 0 & 0 \\
 0 & 1 & p & 1-p \\
\end{array}
\right) +  \frac{1-p}{2}\left(
\begin{array}{cc|cc}
 1 & 0 & 1 & 0 \\
 0 & 0 & 0 & 0 \\ \hline
 1 & 0 & 1 & 0 \\
 0 & 0 & 0 & 0 \\
\end{array}
\right)$$
$$+ \frac{p}{2}\left(
\begin{array}{cc|cc}
 0 & 1 & 1 & 0 \\
 0 & 0 & 0 & 0 \\ \hline
 0 & 1 & 1 & 0 \\
 0 & 0 & 0 & 0 \\
\end{array}
\right) + \frac{p}{2}  \left(
\begin{array}{cc|cc}
 1 & 0 & 1-p & p \\
 0 & 0 & 0 & 0 \\ \hline
 0 & 0 & 0 & 0 \\
 1 & 0 & p & 1-p \\
\end{array}
\right)
  $$
  $$
  \left(
\begin{array}{cc|cc}
 \frac{1}{2} & \frac{1}{2} & \frac{1}{2} & \frac{1}{2} \\
 0 & 0 & 0 & 0 \\ \hline
 \frac{1-p}{2} & \frac{p}{2} & \frac{1}{2} & 0 \\
 \frac{p}{2} & \frac{1-p}{2} & 0 & \frac{1}{2} \\
\end{array}
\right)_{16} = \frac{1-p}{2}\left(
\begin{array}{cc|cc}
 0 & 1 & 0 & 1 \\
 0 & 0 & 0 & 0 \\ \hline
 0 & 0 & 0 & 0 \\
 0 & 1 & 0 & 1 \\
\end{array}
\right) + \frac{1-p}{2}\left(
\begin{array}{cc|cc}
 1 & 0 & 1 & 0 \\
 0 & 0 & 0 & 0 \\ \hline
 1 & 0 & 1 & 0 \\
 0 & 0 & 0 & 0 \\
\end{array}
\right) $$
$$+ \frac{p}{2}\left(
\begin{array}{cc|cc}
 0 & 1 & 1 & 0 \\
 0 & 0 & 0 & 0 \\ \hline
 0 & 1 & 1 & 0 \\
 0 & 0 & 0 & 0 \\
\end{array}
\right) + \frac{p}{2} \left(
\begin{array}{cc|cc}
 1 & 0 & 0 & 1 \\
 0 & 0 & 0 & 0 \\ \hline
 0 & 0 & 0 & 0 \\
 1 & 0 & 0 & 1 \\
\end{array}
\right)
  $$
  $$
  \left(
\begin{array}{cc|cc}
 \frac{p}{2} & \frac{1-p}{2} & \frac{1}{2} & 0 \\
 \frac{1-p}{2} & \frac{p}{2} & \frac{1}{2} & 0 \\ \hline
 \frac{1-p}{2} & \frac{p}{2} & \frac{1}{2} & 0 \\
 \frac{p}{2} & \frac{1-p}{2} & \frac{1}{2} & 0 \\
\end{array}
\right)_{24} = \frac{1-p}{2}\left(
\begin{array}{cc|cc}
 0 & 0 & 0 & 0 \\
 1 & 0 & 1 & 0 \\ \hline
 1 & 0 & 1 & 0 \\
 0 & 0 & 0 & 0 \\
\end{array}
\right) + \frac{1-p}{2}\left(
\begin{array}{cc|cc}
 0 & 1 & 1 & 0 \\
 0 & 0 & 0 & 0 \\ \hline
 0 & 0 & 0 & 0 \\
 0 & 1 & 1 & 0 \\
\end{array}
\right) $$
$$+ \frac{p}{2} \left(
\begin{array}{cc|cc}
 0 & 0 & 0 & 0 \\
 0 & 1 & 1 & 0 \\ \hline
 0 & 1 & 1 & 0 \\
 0 & 0 & 0 & 0 \\
\end{array}
\right)+ \frac{p}{2} \left(
\begin{array}{cc|cc}
 1 & 0 & 1 & 0 \\
 0 & 0 & 0 & 0 \\ \hline
 0 & 0 & 0 & 0 \\
 1 & 0 & 1 & 0 \\
\end{array}
\right)
  $$
  $$
\left(
\begin{array}{cc|cc}
 \frac{1}{2} & \frac{1}{2} & p & 1-p \\
 0 & 0 & 0 & 0 \\ \hline
 \frac{1-p}{2} & \frac{p}{2} & \frac{1-p}{2} & \frac{p}{2} \\
 \frac{p}{2} & \frac{1-p}{2} & \frac{1-p}{2} & \frac{p}{2} \\
\end{array}
\right)_{16} = \frac{1-p}{2}\left(
\begin{array}{cc|cc}
 0 & 1 & p & 1-p \\
 0 & 0 & 0 & 0 \\ \hline
 0 & 0 & 0 & 0 \\
 0 & 1 & 1-p & p \\
\end{array}
\right)+ \frac{1-p}{2}\left(
\begin{array}{cc|cc}
 1 & 0 & p & 1-p \\
 0 & 0 & 0 & 0 \\ \hline
 1 & 0 & 1-p & p \\
 0 & 0 & 0 & 0 \\
\end{array}
\right)$$
$$+ \frac{p}{2} \left(
\begin{array}{cc|cc}
 0 & 1 & p & 1-p \\
 0 & 0 & 0 & 0 \\ \hline
 0 & 1 & 1-p & p \\
 0 & 0 & 0 & 0 \\
\end{array}
\right) + \frac{p}{2}   \left(
\begin{array}{cc|cc}
 1 & 0 & p & 1-p \\
 0 & 0 & 0 & 0 \\ \hline
 0 & 0 & 0 & 0 \\
 1 & 0 & 1-p & p \\
\end{array}
\right)  
  $$
  $$
  \left(
\begin{array}{cc|cc}
 \frac{p}{2} & \frac{1-p}{2} & \frac{p}{2} & \frac{1-p}{2} \\
 \frac{1-p}{2} & \frac{p}{2} & \frac{p}{2} & \frac{1-p}{2} \\ \hline
 \frac{1-p}{2} & \frac{p}{2} & \frac{1-p}{2} & \frac{p}{2} \\
 \frac{p}{2} & \frac{1-p}{2} & \frac{1-p}{2} & \frac{p}{2} \\
\end{array}
\right)_{8} = \frac{1-p}{2}\left(
\begin{array}{cc|cc}
 0 & 0 & 0 & 0 \\
 1 & 0 & p & 1-p \\ \hline
 1 & 0 & 1-p & p \\
 0 & 0 & 0 & 0 \\
\end{array}
\right) + \frac{1-p}{2}\left(
\begin{array}{cc|cc}
 0 & 1 & p & 1-p \\
 0 & 0 & 0 & 0 \\ \hline
 0 & 0 & 0 & 0 \\
 0 & 1 & 1-p & p \\
\end{array}
\right) $$
$$+ \frac{p}{2}\left(
\begin{array}{cc|cc}
 0 & 0 & 0 & 0 \\
 0 & 1 & p & 1-p \\ \hline
 0 & 1 & 1-p & p \\
 0 & 0 & 0 & 0 \\
\end{array}
\right)  + \frac{p}{2} \left(
\begin{array}{cc|cc}
 1 & 0 & p & 1-p \\
 0 & 0 & 0 & 0 \\ \hline
 0 & 0 & 0 & 0 \\
 1 & 0 & 1-p & p \\
\end{array}
\right)
  $$
  $$
  \left(
\begin{array}{cc|cc}
 \frac{p}{2} & \frac{1-p}{2} & \frac{p}{2} & \frac{1-p}{2} \\
 \frac{1-p}{2} & \frac{p}{2} & \frac{1-p}{2} & \frac{p}{2} \\ \hline
 \frac{1-p}{2} & \frac{p}{2} & \frac{1-p}{2} & \frac{p}{2} \\
 \frac{p}{2} & \frac{1-p}{2} & \frac{p}{2} & \frac{1-p}{2} \\
\end{array}
\right)_{8} = \frac{1-p}{2}\left(
\begin{array}{cc|cc}
 0 & 0 & 0 & 0 \\
 0 & 1 & 0 & 1 \\ \hline
 0 & 1 & 0 & 1 \\
 0 & 0 & 0 & 0 \\
\end{array}
\right) + \frac{1-p}{2} \left(
\begin{array}{cc|cc}
 0 & 0 & 0 & 0 \\
 0 & 1 & 1 & 0 \\ \hline
 0 & 1 & 1 & 0 \\
 0 & 0 & 0 & 0 \\
\end{array}
\right) $$
$$+ \frac{p}{4} \left(
\begin{array}{cc|cc}
 0 & 0 & 0 & 0 \\
 1 & 0 & 0 & 1 \\ \hline
 1 & 0 & 0 & 1 \\
 0 & 0 & 0 & 0 \\
\end{array}
\right) + \frac{p}{4}  \left(
\begin{array}{cc|cc}
 0 & 1 & 1 & 0 \\
 0 & 0 & 0 & 0 \\ \hline
 0 & 0 & 0 & 0 \\
 0 & 1 & 1 & 0 \\
\end{array}
\right)
  $$
  $$
  +\frac{p}{4} \left(
\begin{array}{cc|cc}
 1 & 0 & 0 & 1 \\
 0 & 0 & 0 & 0 \\ \hline
 0 & 0 & 0 & 0 \\
 1 & 0 & 0 & 1 \\
\end{array}
\right) +\frac{p}{4} \left(
\begin{array}{cc|cc}
 1 & 0 & 1 & 0 \\
 0 & 0 & 0 & 0 \\ \hline
 0 & 0 & 0 & 0 \\
 1 & 0 & 1 & 0 \\
\end{array}
\right)
  $$
The numbers at the subscripts denote the cardinality of the respective classes.  

\vspace{0.2cm}

\noindent \underline{Restriction of $\mathbb{S}_{(p,p),(1/2,1/2)}$ by $\mathrm{NS}_4$:}

\vspace{0.2cm}

We take all the extremal distributions obtained in the previous analysis and find that upon further restricting to the hyperplane $\left<\mathrm{NS}_4,\mathbf{x}=0\right>$, using the sequence of steps described in \cref{sec:NS_sub}, there is only one class of extremal states with a cardinality of 16. A representative of this class is

\begin{equation} \label{eq:rep_class}
\frac{1}{2} \left(\begin{array}{cc|cc}
        1 & 0 & p & 1-p \\
        0 & 1 & 1-p & p \\  \hline
        1 & 0 & 1-p & p \\
        0 & 1 & p & 1-p \\
    \end{array}\right).
\end{equation}

In the following, we present one element from each extremal class. A convex decomposition, like the one obtained in the previous section, can also be obtained. However, an easier way of noticing that they are not extremal is by the fact that they are all local. The subscripts represent the number of elements present in each class.
$$
\left(
\begin{array}{cc|cc}
 0 & 0 & 0 & 0 \\
 0 & 1 & \frac{1}{2} & \frac{1}{2} \\ \hline
 0 & 0 & 0 & 0 \\
 0 & 1 & \frac{1}{2} & \frac{1}{2} \\
\end{array}
\right)_{8}, \left(
\begin{array}{cc|cc}
 0 & 0 & 0 & 0 \\
 0 & 1 & \frac{1}{2} & \frac{1}{2} \\ \hline
 0 & \frac{1}{2} & \frac{1-p}{2} & \frac{p}{2} \\
 0 & \frac{1}{2} & \frac{p}{2} & \frac{1-p}{2} \\
\end{array}
\right)_{16}, \left(
\begin{array}{cc|cc}
 0 & 0 & 0 & 0 \\
 \frac{1}{2} & \frac{1}{2} & \frac{1}{2} & \frac{1}{2} \\ \hline
 0 & 0 & 0 & 0 \\
 \frac{1}{2} & \frac{1}{2} & \frac{1}{2} & \frac{1}{2} \\
\end{array}
\right)_{4},\left(
\begin{array}{cc|cc}
 0 & 0 & 0 & 0 \\
 \frac{1}{2} & \frac{1}{2} & \frac{1}{2} & \frac{1}{2} \\ \hline
 \frac{1}{2} & 0 & \frac{1-p}{2} & \frac{p}{2} \\
 0 & \frac{1}{2} & \frac{p}{2} & \frac{1-p}{2} \\
\end{array}
\right)_{16}
$$
$$
\left(
\begin{array}{cc|cc}
 0 & \frac{1}{2} & \frac{p}{2} & \frac{1}{6} \\
 0 & \frac{1}{2} & \frac{1}{6} & \frac{p}{2} \\ \hline
 0 & \frac{1}{2} & \frac{1-p}{2} & \frac{p}{2} \\
 0 & \frac{1}{2} & \frac{p}{2} & \frac{1-p}{2} \\
\end{array}
\right)_{8}, \left(
\begin{array}{cc|cc}
 \frac{1-p}{3} & \frac{p}{3} & 0 & \frac{1}{3} \\
 \frac{p}{3} & \frac{2-p}{3} & \frac{1}{3} & \frac{1}{3} \\ \hline
 \frac{1-p}{3} & \frac{2(1-p)}{3} & \frac{1-p}{3} & \frac{2(1-p)}{3} \\
 \frac{p}{3} & \frac{2p}{3} & \frac{p}{3} & \frac{2p}{3} \\
\end{array}
\right)_{32}, \left(
\begin{array}{cc|cc}
 \frac{1-p}{3} & \frac{p}{3} & 0 & \frac{1}{3} \\
 \frac{p}{3} & \frac{2-p}{3} & \frac{1}{3} & \frac{1}{3} \\ \hline
 \frac{p}{3} & \frac{1-p}{3} & 0 & \frac{1}{3} \\
 \frac{1-p}{3} & \frac{1+p}{3} & \frac{1}{3} & \frac{1}{3} \\
\end{array}
\right)_{32},  
$$
$$
\left(
\begin{array}{cc|cc}
 \frac{1-p}{3} & \frac{p}{3} & \frac{p}{3} & \frac{1-p}{3} \\
 \frac{p}{3} & \frac{2-p}{3} & \frac{1-p}{3} & \frac{1+p}{3} \\ \hline
 \frac{1-p}{3} & \frac{p}{3} & 0 & \frac{1}{3} \\
 \frac{p}{3} & \frac{2-p}{3} & \frac{1}{3} & \frac{1}{3} \\
\end{array}
\right)_{32},
\left(
\begin{array}{cc|cc}
 \frac{1-p}{3} & \frac{p}{3} & \frac{2-p}{6} & \frac{p}{6} \\
 \frac{p}{3} & \frac{2-p}{3} & \frac{p}{6} & \frac{4-p}{6} \\ \hline
 \frac{p}{3} & \frac{1-p}{3} & \frac{p}{6} & \frac{2-p}{6} \\
 \frac{1-p}{3} & \frac{1+p}{3} & \frac{2-p}{6} & \frac{2+p}{6} \\
\end{array}
\right)_{32}, \left(
\begin{array}{cc|cc}
 \frac{1-p}{3} & \frac{p}{3} & \frac{p}{3} & \frac{1-p}{3} \\
 \frac{p}{3} & \frac{2-p}{3} & \frac{2-p}{3} & \frac{p}{3} \\ \hline
 \frac{p}{3} & \frac{1-p}{3} & \frac{1-p}{3} & \frac{p}{3} \\
 \frac{1-p}{3} & \frac{1+p}{3} & \frac{1+p}{3} & \frac{1-p}{3} \\
\end{array}
\right)_{32},
$$
$$
\left(
\begin{array}{cc|cc}
 \frac{p}{3} & \frac{1-p}{3} & 0 & \frac{1}{3} \\
 \frac{1-p}{3} & \frac{1+p}{3} & \frac{1}{3} & \frac{1}{3} \\ \hline
 \frac{2-p}{6} & \frac{1-p}{3} & \frac{1-p}{3} & \frac{2-p}{6} \\
 \frac{p}{6} & \frac{1+p}{3} & \frac{p}{3} & \frac{2+p}{6} \\
\end{array}
\right)_{32},
\left(
\begin{array}{cc|cc}
 \frac{p}{3} & \frac{p}{6} & \frac{p}{3} & \frac{p}{6} \\
 \frac{p}{6} & \frac{3-2p}{3} & \frac{1-p}{3} & \frac{4-p}{6} \\ \hline
 \frac{p}{3} & \frac{1-p}{3} & 0 & \frac{1}{3} \\
 \frac{p}{6} & \frac{4-p}{6} & \frac{1}{3} & \frac{1}{3} \\
\end{array}
\right)_{32}, \left(
\begin{array}{cc|cc}
 \frac{1-p}{2} & \frac{p}{2} & 0 & \frac{1}{2} \\
 \frac{p}{2} & \frac{1-p}{2} & \frac{1}{4} & \frac{1}{4} \\ \hline
 \frac{p}{4} & \frac{2-p}{4} & \frac{1}{4} & \frac{1}{4} \\
 \frac{2-p}{4} & \frac{p}{4} & 0 & \frac{1}{2} \\
\end{array}
\right)_{16}, $$
$$\left(
\begin{array}{cc|cc}
 \frac{1-p}{2} & \frac{p}{2} & 0 & \frac{1}{2} \\
 \frac{p}{2} & \frac{1-p}{2} & \frac{4-3p}{8} & \frac{3p}{8} \\ \hline
 \frac{3p}{8} & \frac{4-3p}{8} & \frac{1-p}{2} & \frac{p}{2} \\
 \frac{4-3p}{8} & \frac{3p}{8} & \frac{p}{8} & \frac{p}{8} \\
\end{array}
\right)_{32}, \left(
\begin{array}{cc|cc}
 \frac{1-p}{2} & 1-p & \frac{1-p}{2} & 1-p \\
 2p-1 & \frac{1-p}{2} & \frac{5p-3}{4} & \frac{1+p}{4} \\ \hline
 \frac{1-p}{2} & \frac{1-p}{3} & \frac{1-p}{2} & 1-p \\
 2p-1 & \frac{1-p}{2} & \frac{5p-3}{4} & \frac{1+p}{4} \\
\end{array}
\right)_{16},
$$
$$
\left(
\begin{array}{cc|cc}
 \frac{1-p}{2} & \frac{p}{2} & \frac{1-p}{2} & \frac{p}{2} \\
 \frac{p}{2} & \frac{1-p}{2} & \frac{p}{8} & \frac{4-p}{8} \\ \hline
 \frac{3p}{8} & \frac{4-3p}{8} & \frac{4-3p}{8} & \frac{3p}{8} \\
 \frac{4-3p}{8} & \frac{3p}{8} & 0 & \frac{1}{2} \\
\end{array}
\right)_{32}, \left(
\begin{array}{cc|cc}
 \frac{p}{4} & \frac{2-p}{4} & \frac{1-p}{2} & \frac{p}{2} \\
 \frac{2-p}{4} & \frac{p}{4} & \frac{p}{2} & \frac{1-p}{2} \\ \hline
 \frac{p}{4} & \frac{2-p}{4} & \frac{13-2p}{28} & \frac{1+2p}{28} \\
 \frac{2-p}{4} & \frac{p}{4} & \frac{2p+1}{28} & \frac{13+2p}{28} \\
\end{array}
\right)_{16}, \left(
\begin{array}{cc|cc}
 \frac{p}{4} & 1-p & \frac{p}{4} & 1-p \\
 \frac{p}{2} & \frac{p}{4} & \frac{p}{2} & \frac{p}{4} \\ \hline
 \frac{3p}{8} & \frac{3p}{8} & \frac{3p}{8} & \frac{3p}{8} \\
 \frac{3p}{8} & \frac{8-9p}{8} &\frac{3p}{8} & \frac{8-9p}{8}\\
\end{array}
\right)_{12},
$$
$$
\left(
\begin{array}{cc|cc}
 \frac{p}{4} & \frac{p}{2} & \frac{3p}{8} & \frac{3p}{8} \\
 1-p & \frac{p}{4} & \frac{8-9p}{8} & \frac{3p}{8} \\ \hline
 \frac{8-9p}{8} & \frac{3p}{8} & \frac{8-7p}{8} & \frac{p}{8} \\
 \frac{3p}{8} & \frac{3p}{8} & \frac{p}{8} & \frac{5p}{8} \\
\end{array}
\right)_{8}, \left(
\begin{array}{cc|cc}
 \frac{3p}{8} & \frac{3p}{8} & 0 & \frac{3p}{4} \\
 \frac{4-3p}{8} & \frac{4-3p}{8} & \frac{4-3p}{8} & \frac{4-3p}{8} \\ \hline
 \frac{4-3p}{8} & \frac{4-3p}{8} & \frac{4-3p}{8} & \frac{4-3p}{8} \\
 \frac{3p}{8} & \frac{3p}{8} & 0 & \frac{3p}{4} \\
\end{array}
\right)_{8},
$$
$$
\left(
\begin{array}{cc|cc}
 \frac{3p-1}{4} & \frac{3(1-p)}{4} & \frac{1-p}{2} & \frac{p}{2} \\
 \frac{3p-1}{4} & \frac{3(1-p)}{4} & \frac{1-p}{4} & \frac{1+p}{4} \\ \hline
 \frac{3p-1}{4} & \frac{3(1-p)}{4} & \frac{1-p}{2} & \frac{p}{2} \\
 \frac{3p-1}{4} & \frac{3(1-p)}{4} & \frac{1-p}{4} & \frac{1+p}{4} \\
\end{array}
\right)_{8}, \left(
\begin{array}{cc|cc}
 \frac{3p}{8} & \frac{3p}{8} & \frac{p}{4} & \frac{p}{2} \\
 \frac{3p}{8} & \frac{8-9p}{8} & 1-p & \frac{p}{4} \\ \hline
 \frac{p}{4} & 1-p & \frac{8-9p}{8} & \frac{3p}{8} \\
 \frac{p}{2} & \frac{p}{4} & \frac{3p}{8} & \frac{3p}{8} \\
\end{array}
\right)_{8},\left(
\begin{array}{cc|cc}
 \frac{3p-1}{4} & \frac{3-3p}{4} & \frac{1-p}{2} & \frac{p}{2} \\
 \frac{3p-1}{4} & \frac{3-3p}{4} & \frac{p}{2} & \frac{1-p}{2} \\ \hline
 \frac{3p-1}{4} & \frac{3-3p}{4} & \frac{1+p}{4} & \frac{1-p}{4} \\
 \frac{3p-1}{4} & \frac{3-3p}{4} & \frac{1-p}{4} & \frac{1+p}{4} \\
\end{array}
\right)_{16}
$$
$$
\left(
\begin{array}{cc|cc}
 \frac{1}{4} & \frac{1}{4} & \frac{1}{4} & \frac{1}{4} \\
 \frac{1}{4} & \frac{1}{4} & \frac{1}{4} & \frac{1}{4} \\ \hline
 \frac{1}{4} & \frac{1}{4} & \frac{1}{4} & \frac{1}{4} \\
 \frac{1}{4} & \frac{1}{4} & \frac{1}{4} & \frac{1}{4} \\
\end{array}
\right)_1.
$$
We therefore find that the non-signalling subspace of $\mathbb{S}_{(p,p),(1/2,1/2)}$ is equal to the convex hull of the extremal class \eqref{eq:rep_class}.

\section{Proof of \cref{lem:self-test,lem:2_ways}} \label{app:proof1}

\noindent \textbf{Proposition 2.} \textit{Let $p \in (2/5,1)$ and $T_{p}$ be the Bell functional
    \begin{equation}
        T_{p}(\mathrm{P}) = \langle B_{01} - B_{00} \rangle + \frac{1-p}{p}\langle B_{10} - B_{11} \rangle  + \langle A_{0}(B_{0} +B_{1}) \rangle + \frac{1-p}{p} \langle A_{1}(B_{0}  - B_{1} )\rangle \label{eq:tilt_chsh_app}
    \end{equation}
    Then the following properties hold:
    \begin{enumerate}
        \item For all non-signaling local behaviors $\mathrm{P}_{\mathrm{L}}^{\mathrm{NS}}$, $T_{p}(\mathrm{P}_{\mathrm{L}}^{\mathrm{NS}}) \leq 2(|2p-1|+1-p)/p$.
        \item For all one-way signaling local behaviors $\mathrm{P}_{\mathrm{L}}^{\mathrm{S}} \in \mathbb{S}_{(p,p),(1/2,1/2)}$, $T_{p}(\mathrm{P}_{\mathrm{L}}^{\mathrm{S}}) \leq 2\max\Big\{p , \Big| -p + 2\frac{(1-p)^2}{p}\Big|\Big\} + 2(1-p)$.
        \item For all non-signaling quantum behaviors $\mathrm{P}_{\mathrm{Q}}^{\mathrm{NS}}$, $T_{p}(\mathrm{P}_{\mathrm{Q}}^{\mathrm{NS}}) \leq 2\sqrt{2}\sqrt{\frac{p}{3p-1}}$.
        \item Up to local isometries, there is a unique non-signaling quantum strategy that achieves $T_{p}(\mathrm{P}_{\mathrm{Q}}^{\mathrm{NS}}) = 2\sqrt{2}\sqrt{\frac{p}{3p-1}}$, given by
        \begin{equation}
        \begin{gathered}
            \rho = \ketbra{\psi_{\theta}}{\psi_{\theta}}, \ \ket{\psi_{\theta}} = \cos(\theta) \ket{00} + \sin(\theta) \ket{11},\\
            A_{0} = \sigma_{X}, \ \ \ A_{1} = \sigma_{Z}, \\
            B_{0} = \cos(\phi) \, \sigma_{Z} + \sin(\phi) \, \sigma_{X}, \\
            B_{1} = -\cos(\phi) \, \sigma_{Z} + \sin(\phi) \, \sigma_{X},
        \end{gathered}
        \end{equation}
        where
        \begin{equation}
            \begin{gathered}
                \theta = \frac{1}{2} \mathrm{arctan}\big[f_{1}/g_{1}\big] - \pi/2,\ \
                \phi = \mathrm{arctan}\big[f_{2}/g_{2}\big], \\
                f_{1} = \frac{\sqrt{p(-2 + 5p)}}{1 - 3p}, \ \ g_{1} = \frac{1-2p}{-1+3p},\\
                f_{2} = -\sqrt{\frac{-2+5p}{-2+6p}}, \ \ g_{2} = \sqrt{\frac{p}{-2+6p}}. 
            \end{gathered} \label{eq:angs_app}
        \end{equation}
        when $p \in (2/5,1/2) \cup (1/2,1)$ and $\theta = \phi= -\pi/4$ when $p = 1/2$.
    \end{enumerate}}
\begin{proof}

\vspace{0.2cm}

\noindent \underline{Part 1:} The maximum value is achieved by a non-signaling local deterministic behavior, and in the minimal Bell scenario this is characterized by a tuple $(\mathsf{a}_{0},\mathsf{a}_{1},\mathsf{b}_{0},\mathsf{b}_{1}) \in \{0,1\}^{4}$ with $\langle A_{x} \rangle = (-1)^{\mathsf{a}_{x}}$, $\langle B_{y} \rangle = (-1)^{\mathsf{b}_{y}}$ and $\langle A_{x}B_{y}\rangle  = (-1)^{\mathsf{a}_{x} + \mathsf{b}_{y}}$. With this, the Bell expression \eqref{eq:new_ineq_ns} reduces to
\begin{equation}
    T_{p}(\mathrm{P}^{\text{NS}}_{\mathrm{L}}) = \frac{1-2p}{p}((-1)^{\mathsf{b}_{0}}-(-1)^{\mathsf{b}_{1}}) + (-1)^{\mathsf{a}_{0}}((-1)^{\mathsf{b}_{0}}+(-1)^{\mathsf{b}_{1}}) + \frac{1-p}{p} (-1)^{\mathsf{a}_{1}}((-1)^{\mathsf{b}_{0}}-(-1)^{\mathsf{b}_{1}}). \label{eq:expr_local}
\end{equation}
We now consider two cases. First, suppose $\mathsf{b}_{0} = \mathsf{b}_{1}$. Then \eqref{eq:expr_local} reduces to $2(-1)^{\mathsf{a}_{0} + \mathsf{b}_{0}} \leq 2$ as claimed. If $\mathsf{b}_{0} = \mathsf{b}_{1} \oplus 1$, then \eqref{eq:expr_local} reduces to 
\begin{equation*}
    2\frac{1-2p}{p}(-1)^{\mathsf{b}_{0}} + 2\frac{1-p}{p}(-1)^{\mathsf{a}_{1}+\mathsf{b}_{0}} \leq 2\frac{|2p-1|}{p} + 2\frac{1-p}{p},
\end{equation*}
which is greater than 2 for $p \in (2/5,1/2)$, and equal to 2 for $p \in [1/2,1)$, establishing the claim.

\vspace{0.2cm}

\noindent \underline{Part 2:} According to \cref{eq:p_sig_model}, the maximum of $T_{p}(\mathrm{P}_{\mathrm{L}}^{\mathrm{S}})$ is achieved at a vertex
\begin{equation*}
    \p(a,b|x,y) = \Bigg(\sum_{\tilde{y}}\p_{\tilde{Y}|Y}(\tilde{y}|y) \, \delta_{a,f_{A}(x,\tilde{y})}\Bigg)\Bigg(\sum_{\tilde{x}}\p_{\tilde{X}|X}(\tilde{x}|x) \, \delta_{b,f_{B}(\tilde{x},y)}\Bigg)
\end{equation*}
where $f_{A},f_{B}:\{0,1\}^{2} \to \{0,1\}$ are deterministic functions. In the case of symmetric one-way signaling from Alice to Bob, $\p_{\tilde{Y}|Y}(\tilde{y}|y) = 1/2$ and $\p_{\tilde{X}|X}(\tilde{x}|x) = p$ if $\tilde{x} =x$ and $1-p$ if $\tilde{x} = x \oplus 1$. We therefore have
\begin{equation*}
    \p(a,b|x,y) = \delta_{a,f_{A}(x)} \Big( p \, \delta_{b,f_{B}(x,y)} + (1-p)\, \delta_{b,f_{B}(x\oplus 1,y)} \Big),
\end{equation*}
where $f_{A}:\{0,1\} \to \{0,1\}$ is a deterministic function. Let $\mathsf{a}_{x} = f_{A}(x)$ and $\mathsf{b}_{\tilde{x},y} = f_{B}(\tilde{x},y)$. We then have
\begin{equation*}
    \begin{aligned}
        \langle A_{x} \rangle &= (-1)^{\mathsf{a}_{x}}, \\ 
        \langle B_{xy} \rangle &= p(-1)^{\mathsf{b}_{x,y}} + (1-p)(-1)^{\mathsf{b}_{x\oplus 1,y}} \ \text{and}\\
        \langle A_{x}B_{y} \rangle &= p(-1)^{\mathsf{a}_{x} + \mathsf{b}_{x,y}} + (1-p)(-1)^{\mathsf{a}_{x} + \mathsf{b}_{x\oplus 1,y}}.
    \end{aligned}
\end{equation*}
Inserting this into the expression $T_{p}$,
\begin{multline*}
    T_{p}(\mathrm{P}_{\mathrm{L}}^{\mathrm{S}}) = p(-1)^{\mathsf{b}_{0,1}} + (1-p)(-1)^{\mathsf{b}_{1,1}} - p(-1)^{\mathsf{b}_{0,0}} - (1-p)(-1)^{\mathsf{b}_{1,0}} \\ + \frac{1-p}{p}\big( p(-1)^{\mathsf{b}_{1,0}} + (1-p)(-1)^{\mathsf{b}_{0,0}} - p(-1)^{\mathsf{b}_{1,1}} - (1-p)(-1)^{\mathsf{b}_{0,1}}\big) \\ + p(-1)^{\mathsf{a}_{0} + \mathsf{b}_{0,0}} + (1-p)(-1)^{\mathsf{a}_{0} + \mathsf{b}_{1,0}} + p(-1)^{\mathsf{a}_{0} + \mathsf{b}_{0,1}} + (1-p)(-1)^{\mathsf{a}_{0} + \mathsf{b}_{1,1}} \\
    + \frac{1-p}{p}\big( p(-1)^{\mathsf{a}_{1} + \mathsf{b}_{1,0}} + (1-p)(-1)^{\mathsf{a}_{1} + \mathsf{b}_{0,0}} - p(-1)^{\mathsf{a}_{1} + \mathsf{b}_{1,1}} - (1-p)(-1)^{\mathsf{a}_{1} + \mathsf{b}_{0,1}}\big).
\end{multline*}
By pairing up terms, the above is equal to
\begin{multline*}
    (-1)^{\mathsf{b}_{0,0}}\Big(-p + \frac{(1-p)^2}{p} + p(-1)^{\mathsf{a}_{0}} + (-1)^{\mathsf{a}_{1}}\frac{(1-p)^2}{p}\Big) \\ 
    + (-1)^{\mathsf{b}_{0,1}}\Big( p - \frac{(1-p)^2}{p} + p(-1)^{\mathsf{a}_{0}} - (-1)^{\mathsf{a}_{1}}\frac{(1-p)^2}{p}\Big) \\
    + (1-p)(-1)^{\mathsf{b}_{1,0}}\Big( (-1)^{\mathsf{a}_{0}} + (-1)^{\mathsf{a}_{1}}\Big)  +(1-p)(-1)^{\mathsf{b}_{1,1}}\Big((-1)^{\mathsf{a}_{0}} -  (-1)^{\mathsf{a}_{1}}\Big).
\end{multline*}
We now consider two cases. First, suppose $\mathsf{b}_{0,0} = \mathsf{b}_{0,1}$. Then the above reduces to
\begin{multline*}
    2p(-1)^{\mathsf{b}_{0,0}+\mathsf{a}_{0}}
    + (1-p)(-1)^{\mathsf{b}_{1,0}}\Big( (-1)^{\mathsf{a}_{0}} + (-1)^{\mathsf{a}_{1}}\Big)  +(1-p)(-1)^{\mathsf{b}_{1,1}}\Big((-1)^{\mathsf{a}_{0}} -  (-1)^{\mathsf{a}_{1}}\Big) \\
    \leq 2p(-1)^{\mathsf{b}_{0,0}+\mathsf{a}_{0}}
    + 2(1-p)(-1)^{\mathsf{b}_{1,0}+\mathsf{a}_{0}} \leq 2
\end{multline*}
as claimed. If $\mathsf{b}_{0,0} = \mathsf{b}_{0,1} \oplus 1$, we have
\begin{multline*}
    2(-1)^{\mathsf{b}_{0,0}}\Big(-p + \frac{(1-p)^2}{p} + (-1)^{\mathsf{a}_{1}}\frac{(1-p)^2}{p}\Big) + (1-p)(-1)^{\mathsf{b}_{1,0}}\Big( (-1)^{\mathsf{a}_{0}} + (-1)^{\mathsf{a}_{1}}\Big)  +(1-p)(-1)^{\mathsf{b}_{1,1}}\Big((-1)^{\mathsf{a}_{0}} -  (-1)^{\mathsf{a}_{1}}\Big) \\ \leq 2(-1)^{\mathsf{b}_{0,0}}\Big(-p + \frac{(1-p)^2}{p} + (-1)^{\mathsf{a}_{1}}\frac{(1-p)^2}{p}\Big) + 2(1-p) \\
    \leq 2\max\Big\{p , \Big| -p + 2\frac{(1-p)^2}{p}\Big|\Big\} + 2(1-p).
\end{multline*}
Noting that the right hand side of the above equation is greater than or equal to 2 for $p\in (2/5,1)$ proves the claim.

\vspace{0.2cm}

\noindent \underline{Part 3:} We prove Part 3 via a sum of squares decomposition. Let $\mathcal{H}_{Q_{A}},\mathcal{H}_{Q_{B}},\{M_{a|x}\},\{N_{b|y}\},\rho_{Q_{A}Q_{B}}$ be a non-signaling quantum strategy, and $\mathrm{P}_{\mathrm{q}}$ the induced behavior . Let $A_{x} = \frac{1}{2}(\id_{Q_{A}} + (-1)^{a}M_{a|x})$ and $B_{y} = \frac{1}{2}(\id_{Q_{B}} + (-1)^{b}B_{b|y})$ denote Alice's and Bob's observables, such that $\langle A_{x}B_{y}\rangle = \tr[(A_{x}\otimes B_{y})\rho]$, $\langle A_{x} \rangle = \tr[(A_{x} \otimes \id_{Q_{B}})\rho]$ and $\langle B_{y} \rangle = \tr[(\id_{Q_{A}}\otimes B_{y})\rho]$. Then the following Bell operator
\begin{equation*}
    B_{p} = \frac{1-2p}{p}\Big( \id_{Q_A} \otimes B_{0} - \id_{Q_A} \otimes B_{1} \Big) + A_{0} \otimes B_{0} + A_{0} \otimes B_{1} + \frac{1-p}{p}\Big( A_{1} \otimes B_{0} - A_{1} \otimes B_{1}\Big) 
\end{equation*}
satisfies $T_{p}(\mathrm{P}_{\mathrm{q}}) = \tr[B_{p}\rho]$. Furthermore, if there exists a real number $\eta$ and set of polynomials $P_{i}$ in the operators $A_{x}\otimes \id_{Q_B}$, $\id_{Q_A} \otimes B_{y}$ and $A_{x}\otimes B_{y}$ such that
\begin{equation*}
    \eta \, \id_{Q_{A}Q_{B}} - B_{p} = \sum_{i}P_{i}^{\dagger}P_{i},
\end{equation*}
$\eta$ is an upper bound on the maximum quantum value of $T_{p}(\mathrm{P}_{\mathrm{q}})$, i.e., 
\begin{equation*}
    \eta - T_{p}(\mathrm{P}_{\mathrm{q}}) = \tr[(\eta \, \id_{Q_{A}Q_{B}} - B_{p})\rho] = \sum_{i}\tr[P_{i}^{\dagger}P_{i}\rho] \geq 0.
\end{equation*}
Importantly, the polynomials $P_{i}$ are independent of the actual realization of the operators $A_{x}$ and $B_{y}$, implying that the upper bound holds for all non-signaling quantum strategies. 

As mentioned in the main text, the family of Bell inequalities \eqref{eq:tilt_chsh_app} are a special case of a family discovered by Barizien, Sekatski and Bancal~\cite{Barizien_2024}. An SOS decomposition was provided in Ref.~\cite{Barizien_2024}, and we restate it here for the convenience of the reader. Let $p\in (2/5,1)$, and $(\theta,\phi)$ be given by \cref{eq:angs_app}. Note that for $p\in (2/5,1)$,
\begin{equation*}
    \begin{aligned}
        \sin(\phi) &= -\sqrt{\frac{2-5p}{2-6p}} \in (-\sqrt{3}/2,0), \\
        \cos(\phi) &= \sqrt{\frac{p}{-2+6p}} \in (0,1), \\
        \sin(2\theta) &= \frac{\sqrt{p(-2+5p)}}{1-3p} \in (-\sqrt{3}/2,0) \ \text{and} \\
        \cos(2\theta) &= \frac{1-2p}{-1+3p} \in (-1/2,1).
    \end{aligned}
\end{equation*}
Define the following polynomials:
\begin{equation*}
    \begin{aligned}
        N_{0} &= A_{1} \otimes \id_{Q_{B}} - \id_{Q_{A}} \otimes \frac{B_{0}-B_{1}}{2\cos(\phi)},\\
        N_{1} &= A_{0} - \sin(2\theta)\, \id_{Q_{A}} \otimes \frac{B_{0} + B_{1}}{2\sin(\phi)} - \cos(2\theta) \, A_{0} \otimes \frac{B_{0}-B_{1}}{2\cos(\phi)}.
    \end{aligned}
\end{equation*}
Let $\lambda^{2} > 0$ be such that
\begin{equation*}
    \frac{1}{\lambda^{2}}  = \frac{\sin^{2}(2\theta)}{\tan^{2}(\phi)} - \cos^{2}(2\theta) = \frac{1-p}{-1+3p} \in (1,3).
\end{equation*}
Then by direct calculation under the projective measurement assumption (i.e., $A_{x}^{2} = \id_{Q_{A}}$ and $B_{y}^{2} = \id_{Q_{B}}$), we find
\begin{multline}
    N_{0}^{2} + \lambda^{2}\, N_{1}^{2} = \Big(1 + \lambda^{2} + \frac{1}{2\cos(\phi)} + \frac{\lambda^2 \cos^{2}(2\theta)}{2\cos^{2}(\phi)} + \frac{\lambda^2 \sin^{2}(2\theta)}{2\sin^{2}(\phi)}\Big) \id_{Q_{A}Q_{B}} \\- \Big( A_{1}\otimes \frac{B_{0} - B_{1}}{\cos(\phi)} + \lambda^{2}\, \sin(2\theta)\, A_{0} \otimes \frac{B_{0} + B_{1}}{\sin(\phi)} + \lambda^{2} \, \cos(2\theta) \, \id_{Q_{A}} \otimes \frac{B_{0} - B_{1}}{\cos(\phi)}\Big). \label{eq:big_expr}
\end{multline}
By multiplying both sides of the above expression by $\sin(\phi)/(\lambda^{2}\sin(2\theta)) >0 $ and performing trigonometric simplifications, we recover 
\begin{equation}
    2\sqrt{2}\sqrt{\frac{p}{3p-1}} \, \id_{Q_{A}Q_{B}} - B_{p} = \frac{\sin(\phi)}{\lambda^{2} \,\sin(2\theta)}\Big(N_{0}^{2} + \lambda^{2} \, N_{1}^{2}\Big) \geq 0
\end{equation}
as desired.

\vspace{0.2cm}

\noindent \underline{Part 4:} The self-testing properties of the family of Bell expressions \eqref{eq:big_expr} were proven in~\cite[Appendix B]{Barizien_2024}. For fully non-signaling quantum behaviors, \eqref{eq:tilt_chsh_app} is a special case of \eqref{eq:big_expr}, and the self-testing proof follows immediately. 
    
\end{proof}

\noindent \textbf{Proposition 3.} \textit{Let $p,r \in (0,1)$ and $W_{p,r}$ be the Bell functional 
    \begin{equation}
    W_{p,r}(\mathrm{P}) = 1 - \p(11|00) - (1-r)\p(01|01) \\ - (1-p)\p(10|10) + (1-p)(1-r)\p(11|11). \label{eq:2way_app}
    \end{equation}
    Then the following properties hold:
    \begin{enumerate}
        \item For all two-way signaling local behaviors $\mathrm{P}_{\mathrm{L}}^{\mathrm{S}} \in \mathbb{S}_{(p,1),(r,1)}$, $W_{p,r}(\mathrm{P}_{\mathrm{L}}^{\mathrm{S}}) \leq 1$.
        \item There exists a non-signaling quantum behavior $P_{\mathrm{Q}}^{\mathrm{NS}}$ such that $W_{p,r}(\mathrm{P}_{\mathrm{Q}}^{\mathrm{S}}) > 1$. 
    \end{enumerate}}

\begin{proof}
\noindent \underline{Part 1:} As done in the proof of \cref{lem:self-test}, we can parameterise a LHV model using 8 $\pm 1$ valued variables $\alpha_{x,y} = (-1)^{\mathsf{a}_{x,y}}$ and $\beta_{x,y} = (-1)^{\mathsf{b}_{x,y}}$ for $x,y\in \{0,1\}$, where $\mathsf{a}_{x,y}$ and $\mathsf{b}_{x,y}$ are the binary response functions. We then define 
\begin{equation*}
\begin{aligned}
    \langle A_{x,0} \rangle &= r\alpha_{x,0} + (1-r)\alpha_{x,1} \\
    \langle A_{x,1} \rangle &= \alpha_{x,1} \\
    \langle B_{0,y} \rangle &= p\beta_{0,y} + (1-p)\beta_{1,y} \\
    \langle B_{1,y} \rangle &= \beta_{1,y},
\end{aligned}
\end{equation*}
and the respective probabilities are given by
\begin{equation*}
    \p(a,b|x,y) = \frac{1}{4}\big( 1 + (-1)^a \langle A_{x,y} \rangle\big)\big( 1 + (-1)^b \langle B_{x,y} \rangle \big).
\end{equation*}
Inserting this into \cref{eq:2way_app}, $W_{p,r}(\mathrm{P}) \leq 1$ is equivalent to
\begin{multline}
         (1-p)(1-r)(1-\langle A_{1,1}\rangle)(1-\langle B_{1,1}\rangle) \leq (1-\langle A_{0,0}\rangle)(1-\langle B_{0,0}\rangle) \\ + (1-r)(1+\langle A_{0,1}\rangle)(1-\langle B_{0,1}\rangle) + (1-p)(1-\langle A_{1,0}\rangle)(1+\langle B_{1,0}\rangle). \label{eq:big1}
\end{multline}
Note that 
\begin{equation}
\begin{aligned}
    (1-\langle A_{0,0}\rangle)(1-\langle B_{0,0}\rangle) &= \big(r(1-\alpha_{0,0}) + (1-r)(1-\alpha_{0,1})\big)\big(p(1-\beta_{0,0}) + (1-p)(1-\beta_{1,0})\big) \\
    &\geq (1-r)(1-p)(1-\alpha_{0,1})(1-\beta_{1,0}), \label{eq:sub1}
\end{aligned}
\end{equation}
where we used the fact that $1 \pm \alpha_{x,y} \geq 0$ and $1 \pm \beta_{x,y} \geq 0$. Similarly we have
\begin{equation}
    (1-r)(1+\langle A_{0,1}\rangle)(1-\langle B_{0,1}\rangle) \geq (1-r)(1-p)(1+\alpha_{0,1})(1-\beta_{1,1}) \label{eq:sub2}
\end{equation}
and 
\begin{equation}
    (1-p)(1-\langle A_{1,0}\rangle)(1+\langle B_{1,0}\rangle) \geq (1-r)(1-p)(1-\alpha_{1,1})(1+\beta_{1,0}). \label{eq:sub3}
\end{equation}
Note that these bounds can be achieved by choosing $\alpha_{0,0} = \beta_{0,0} = \beta_{0,1} = \alpha_{1,0} = +1$. Inserting \cref{eq:sub1,eq:sub2,eq:sub3} into the right hand side of \cref{eq:big1}, $W_{p,r}(\mathrm{P}) \leq 1$ if the following inequality holds:
\begin{equation}
     (1-p)(1-r)(1-\langle A_{1,1}\rangle)(1-\langle B_{1,1}\rangle) \leq (1-p)(1-r)\Big( (1-\alpha_{0,1})(1-\beta_{1,0}) + (1+\alpha_{0,1})(1-\beta_{1,1}) + (1-\alpha_{1,1})(1+\beta_{1,0})\Big). \label{eq:newIneq}
\end{equation}
Expanding and rearranging (recall that $p,r<1$ so $(1-p)(1-r) > 0$), \cref{eq:newIneq} is equivalent to 
\begin{equation}
    \alpha_{1,1}(\beta_{1,1} + \beta_{1,0}) + \alpha_{0,1}(\beta_{1,1} - \beta_{1,0}) \leq 2.
\end{equation}
This is equivalent to the CHSH inequality~\cite{CHSH}, and thus holds for all choices of $\alpha_{1,1}, \, \beta_{1,1}, \, \beta_{1,0}$ and $\alpha_{0,1}$. Moreover, it can be saturated by choosing, for example, $\alpha_{1,1} = \alpha_{0,1} = \beta_{1,1} = \beta_{1,0} = +1$. This completes the proof of Part 1.

\vspace{0.2cm}

\noindent \underline{Part 2:} To prove the second part, we provide a non-signaling quantum behavior $\mathrm{P}^{\text{NS}}_{\text{Q}}$ that satisfies $W_{p,r}(\mathrm{P}^{\text{NS}}_{\text{Q}}) > 1$ for all $p<1$ and $r < 1$. Define the following parameters:
\begin{equation}
    \begin{aligned}
        \phi_{0} &= \pi + 2 \, \arcsin\Bigg( \sqrt{\frac{2}
        {2^{\frac{2}{3}} + 2}} \Bigg) \\
        \phi_{1} &= \arctan(2\sqrt{2})\\
        \theta &= \pi - \arctan\Bigg( \frac{2^{\frac{2}{3}}}{2} \Bigg).
        \end{aligned}
\end{equation}
The two-qubit non-signaling quantum strategy is given by
\begin{equation}
    \begin{gathered}
        \ket{\psi} = \cos(\theta)\ket{00} + \sin(\theta) \ket{11}\\
        A_{x} = \cos(\phi_{x}) \, \sigma_{Z} + \sin(\phi_{x}) \,\sigma_{X} \\
        B_{y} = A_{y},
    \end{gathered}
\end{equation}
and its induced behavior $\mathrm{P}^{\text{NS}}_{\text{Q}}$ satisfies, for $p,r\in (0,1)$,
\begin{multline}
    W_{p,r}(\mathrm{P}^{\text{NS}}_{\text{Q}}) -1 = \frac{4 + 11\cdot 2^{\frac{1}{3}} + 3 \cdot 2^{\frac{2}{3}} + 2(2^{\frac{1}{3}}-1)(r-1)p-2(2^{\frac{1}{3}}-1)r}{3(2+3\cdot 2^{\frac{1}{3}} + 2^{\frac{2}{3}})} -1 \\= \frac{1}{45}(16 + 2\cdot 2^{\frac{3}{2}} - 11 \cdot 2^{\frac{2}{3}})(1-p)(1-r) > 0,
\end{multline}
where we noted that the constant term $\frac{1}{45}(16 + 2\cdot 2^{\frac{3}{2}} - 11 \cdot 2^{\frac{2}{3}}) \approx 0.0235$ is strictly positive. Therefore $W_{p,r}(\mathrm{P}^{\text{NS}}_{\text{Q}}) > 1$, completing the proof.
\end{proof}

\newpage 
 \section{Additional plots} \label{app:plots}
    
\begin{figure}[hbt!]
    \centering
    \includegraphics[width=0.5\textwidth]{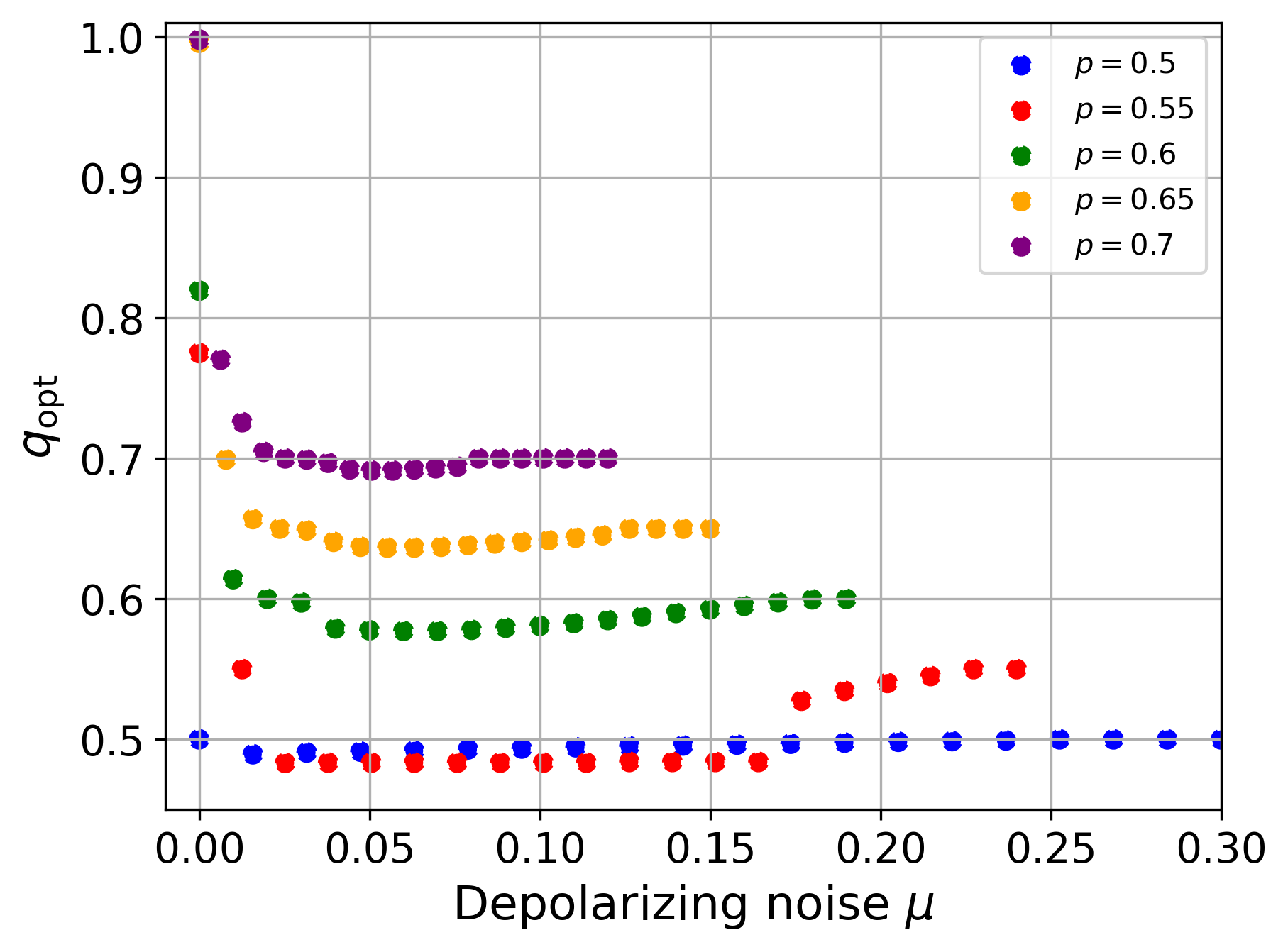}
    \caption{Value of the parameter $q \in (2/5,1)$ that certifies the maximum randomness from the family of Bell expressions $T_{q}$ as a function of depolarizing noise and the probability that Alice's input is sent to Bob, $p \in [1/2,1]$.}
    \label{fig:rand_3}
\end{figure}

\end{document}